\documentclass[a4paper,12pt]{article}

\RequirePackage{natbib}
\usepackage{amsmath}
\usepackage{amssymb}
\usepackage{latexsym}
\usepackage{ntheorem}
\usepackage{ifthen}
\usepackage{xcolor}
\usepackage{newcent}
\usepackage[all,cmtip]{xy}

\usepackage{smartdiagram}
\usepackage{centernot}
\usepackage{tikz-cd}
\usepackage {mathtools}
\mathtoolsset{showonlyrefs=true}

\usepackage{setspace}

\usepackage[margin=2.0cm]{geometry}

\theoremstyle{plain}

\theoremheaderfont{\bf}\theorembodyfont{\sl}

\newtheorem{theorem}{Theorem}[section]
\newtheorem{proposition}[theorem]{Proposition}
\newtheorem{lemma}[theorem]{Lemma}
\newtheorem{corollary}[theorem]{Corollary}

\theoremheaderfont{\it}\theorembodyfont{\rm}

\newtheorem{remark}[theorem]{Remark}

\theoremheaderfont{\bf}\theorembodyfont{\rm}

\newtheorem{definition}[theorem]{Definition}
\newtheorem{example}[theorem]{Example}

\newtheorem{assumption}[theorem]{Assumption}

\theoremstyle{nonumberplain}

\theoremheaderfont{\bf}\theorembodyfont{\sl}

\newenvironment{proof}[1][]
{\ifthenelse{\equal{#1}{}}{\smallskip\noindent\textsl{Proof. }}{\smallskip
\noindent\textsl{Proof #1. }}}{\hfill$\Box$}

\def\Bc{{\cal B}}

\def\Gc{{\cal G}}

\def \l{\lambda}

\def \p{\pi}

\def \o{\omega}
\def \O{\Omega}

\def \e{\varepsilon}
\begin{document}
\thispagestyle{empty}
\title{No-arbitrage with multiple-priors in discrete time}

\author {
{Romain} {Blanchard}, E.mail~:  romblanch@hotmail.com\\
\and
{Laurence} {Carassus}, E.mail~: laurence.carassus@devinci.fr \\
L\' eonard de Vinci P\^ole Universitaire, Research Center, \\
92916 Paris La D\'efense, France and 
 \\
LMR, FRE 2011 Universit\'e Reims Champagne-Ardenne.\\
}


\maketitle

\begin{abstract}
In a discrete time and multiple-priors setting, we propose a new characterisation of the condition of quasi-sure no-arbitrage  which has become a standard assumption. This  characterisation shows that it is indeed a well-chosen condition being equivalent to several previously used alternative notions of no-arbitrage and allowing the proof of 
important results in mathematical finance.
We also revisit the so-called  geometric and quantitative no-arbitrage conditions and   explicit two important examples where all these  concepts are illustrated.  
\end{abstract}

\section{Introduction}
\label{intro}
The concept of no-arbitrage  is fundamental in the modern theory of mathematical finance. Roughly speaking, it means that one cannot hope to make a profit without taking some risk. In a classical uni-prior setting,  
the Fundamental Theorem of Asset Pricing (FTAP in short)  makes the link between  an appropriate notion of no-arbitrage  and the existence of  equivalent risk-neutral probability measures.  
This result is essential   for pricing issues, namely for the superreplication price  which is  for a given claim the minimum selling price needed  to superreplicate  it by trading in the market. 
The FTAP was initially formalised in \citep{Hakr79}, \citep{HaPL81} and \citep{Kr81} while  \citep{dmw}  established it  in a general discrete-time setting and  \citep{DeS94} in  continuous time models.  The literature on the subject is huge and we refer to \citep{DelSch05} for a general overview.\\
However, the reliance on a single probability measure has long been questioned in the economic literature  and is often referred to as Knightian uncertainty, in reference to \citep{Kni}. In a financial context, it is  called model-risk and also has  a long history. The  financial crisis together with the  evolution  of the structure and behaviour of financial markets, have made these issues even  more acute for both academics and practitioners. In particular, this has motivated further research to find good notions of no-arbitrage allowing to  extend the FTAP and the superreplication price characterisation 
while accounting for model uncertainty. A typical example of such endeavor, directly motivated by concrete situations, is to find no-arbitrage prices for some exotic derivative products (such as barrier options, lookback options, double digit options,...)    using as input the prices of actively traded european  options,   without making any  assumptions on the dynamic of the underlying.  This is the so-called model-independent approach, pioneered in \citep{Ho982}.  We refer to \citep{Ho11} for a detailed presentation including the related  Skorokhod embedding problem. Importantly,   \citep{DaHo07}   have shown  that  the expected dichotomy between the existence of a suitable martingale measure and the existence of a model-independent arbitrage might not hold.  \citep{ABPW13} have also  established a FTAP   in a model-independent framework 
under a fairly weak notion of no-arbitrage\footnote{An arbitrage is a strategy with a strictly positive terminal payoff in all  states of the world.},  but assuming  the existence of a traded option with a super-linearly growing payoff-function.   \\ 
An alternative way of  modeling uncertainty   is to replace the single probability measure of the classical setting with a set of priors representing all the possible models: This is the so-called quasi-sure or multiple-priors approach. As the set can vary between a singleton  and all the probability measures on a given space, this formulation encompasses  a wide range of  settings, including the classical one.  As the set of priors is not assumed to  be dominated, this has raised challenging  mathematical  questions  and has lead to the development of innovative  tools such as quasi-sure stochastic analysis, non-linear expectations and G-Brownian motions. On these topics,  we refer among  others to \citep{pg07, pg11}, \citep{DM06}, \citep{DHP11}, \citep{NuHa13},   \citep{SoToZa11} and \citep{SoToZa11b}.\\
Following this approach,  \citep{BN} have introduced in a discrete-time setting with  time  horizon $T$,   a no-arbitrage condition called the $NA(\mathcal{Q}^{T})$ condition (where $\mathcal{Q}^{T}$ represents all the possible models). It states that  if the terminal value of a   trading strategy  is non-negative $\mathcal{Q}^{T}$-quasi-surely, then it  always equals  $0$ $\mathcal{Q}^{T}$-quasi-surely (see Definition \ref{NAQTARB}).  This is a natural  extension of the classical uni-prior where almost sure equality and inequality are replaced with their quasi-sure pendant. \citep{BN} established a generalisation of the  FTAP together with a Superhedging Theorem.  This framework has also been used to study a large range of related problems (FTAP with transaction cost, american options, worst-case optimal investment, ...) and we refer  among others  to \citep{BN2},  \citep{Bay15}, \citep{BC16}  and \citep{Bart216}.\\
Finally, the so-called pathwise approach is an other fruitful modeling approach:  In this setting, uncertainty is introduced by describing  a  subset of  relevant events or scenarii  without  references to   any probability measure and without  specifying their relative weight. In a discrete-time setting, \citep{ClassS}, \citep{Bur16b} introduce  a set of scenarii $\mathcal{S}$   representing the agent beliefs and an Arbitrage de la Classe $\mathcal{S}$ is a trading strategy  leading to a terminal value that is always non-negative for all the events in $\mathcal{S}$ and positive for a least one event in $\mathcal{S}$. A corresponding FTAP is then obtained.  Note that by choosing different sets $\mathcal{S}$, different definitions of no-arbitrage  can be considered and in particular  the model independent approach previously mentioned can  be recovered by choosing  the whole space for $\mathcal{S}$. Importantly,  \citep{OW18} have recently unified  the quasi-sure and the pathwise approaches  showing that under  technical assumptions both approaches are actually equivalent (see Metatheorem 1.1, see also Remark \ref{OW}).\\

In this paper we follow the multiple-priors approach of \citep{BN}.  Despite its success, one might still wonder if the $NA(\mathcal{Q}^{T})$ condition is the ``right'' one.   Indeed, at first sight at least, under this condition it is not even clear if there exists a model $P \in \mathcal{Q}^{T}$  
satisfying the uni-prior no-arbitrage condition $NA(P)$. Theorem \ref{PPstar} will prove that this is in fact possible.  But as  Lemmata \ref{exex} and \ref{lembin} show, $\mathcal{Q}^{T}$ might still contain some models that are not arbitrage free. 
This means that an  agent may not be able to delta-hedge  a simple vanilla option  using different levels of volatility in a arbitrage free way.   
So instead of $NA(\mathcal{Q}^{T})$ one may assume that every model is arbitrage free i.e. that the  $NA(P)$ condition holds true for every model $P \in \mathcal{Q}^{T}$. We call this $sNA(\mathcal{Q}^{T})$ for strong no-arbitrage, see  Definition \ref{sNA}. This alternative condition has appeared 
in recent results on robust utility maximisation of unbounded functions, see for instance \citep{BC16} and \citep{RaMe18}.  
Our main result provides a characterisation of the $NA(\mathcal{Q}^{T})$  condition that gives some kind of definitive answer to these questions and  confirms that the $NA(\mathcal{Q}^{T})$  condition is indeed the ``right" condition in the quasi-sure setting.  More precisely, Theorem \ref{TheoS} shows that the $NA(\mathcal{Q}^{T})$ condition is equivalent to the existence of 
a subclass of priors $\mathcal{P}^{T} \subset \mathcal{Q}^{T}$ such that  $\mathcal{P}^{T}$ and $ \mathcal{Q}^{T}$ have the same polar sets (roughly speaking the same relevant events)  
 and  such that the $sNA(\mathcal{P}^{T})$ hold true.  
In addition to enable a better economic comprehension of $NA(\mathcal{Q}^{T})$, 
Theorem \ref{TheoS}  also provides several interesting results. 
First, it allows for a short  proof of  
a refinement of the FTAP of \citep{BN} using the classical
Dalang-Morton-Willinger Theorem (see Corollary \ref{coroerhan}  and 
\citep[Theorem 2.1]{Bay17}). 
Then, Theorem \ref{TheoS}  provides tractable theorems for the existence of solutions in the problem of  robust utility
maximisation. Indeed it allows to prove the equivalence  between  $NA(\mathcal{Q}^{T})$  and two other conditions previously used in the litterature for solving this problem.  
The first one is the no-arbitrage  condition introduced in  \cite{BCK19} which states that for every prior $Q \in \mathcal{Q}^{T}$  there exist a prior $P \in \mathcal{Q}^{T}$ such that $Q \ll P$  and $NA(P)$ holds true (see Corollary \ref{corodaniel}).   
The second one is the condition used by \citep{RaMe18} which requires the 
existence of a model  $P^* \in \mathcal{Q}^{T}$   satisfying    $NA(P^*)$ and such that for this model the affine space generated by the conditional support always equals $\mathbb{R}^{d}$ (see  Theorem \ref{PPstar}, Remark \ref{andrea} and also \citep{Bay17}  in a one period setup).  
Finally, Theorem \ref{TheoS} allows to show that one may replace the set $\mathcal{Q}^{T}$ by the set $\mathcal{P}^{T}$ in the problem of maximisation of robust  expected utility without changing the value function (see Lemma \ref{lemutil} and Corollary  \ref{coroutil}).

We then introduce local characterisations  of  the $NA(\mathcal{Q}^{T})$  condition called   
  the geometric  and the quantitative  conditions (see  Definition \ref{NAGarb}, \ref{NAQarb} and Theorem \ref{ARBequival}). The geometric condition goes back in the uni-prior setup to \citep[Theorem 3 g)]{JS98} and provides some  geometric intuition.  Theorem \ref{ARBequival} generalises the preceding result to the quasi-sure setting. The geometric condition is an important tool in the multiple-priors literature. It has been used in different setups by  \citep{OW18} and by \citep{Burz18}. It is also efficient to prove concretely that the $NA(\mathcal{Q}^{T})$ condition holds true. 
The quantitative no-arbitrage goes back to   \citep[Proposition 3.3]{RS05} and is used to solve optimisation problems using the dynamic programming principle. For example,  it provides explicit bounds on the optimal strategies in the problem of maximisation of expected utility, see Remark \ref{rembornes}. 
Again Theorem \ref{ARBequival} generalises \citep[Proposition 3.3]{RS05} to the quasi-sure setting. Together with Propositions \ref{alphatmesARB} and  \ref{finally}, this fills a gap opened in \citep[Proposition 2.3]{BC16}, proving difficult measurability results and opening the possibility to solve, in the setting of \citep{BN}, the problem of multi-prior optimal investment for unbounded utility function defined on the whole real-line (see Remark \ref{alphabeta}). \\
Finally, Proposition \ref{CasD} explicits the relation between the different notions of no-arbitrage in the dominated case 
while Proposition \ref{DOM} is used to build examples of   sets  of probability measures $\mathcal{Q}^{T}$ which are not dominated. 

The proofs follow the same idea:  We first study  a one-period problem with deterministic initial data  where we rely on separation theorem and elementary geometric consideration  in finite dimension. Then we extend the results to the multi-period setting relying on  advanced  measurable selections arguments. The proof of Proposition \ref{DOM} relies  also on  relatively  recent topological results. 

Finally,  these theoretical results are complemented by two concrete and useful examples. The first one proposes a multiple-priors binomial model and the second  one  a  generic way of  introducing uncertainty for the discretised dynamics of a  diffusion process. In both cases, we show that the $NA(\mathcal{Q}^{T})$ conditions holds true and  provide explicit expressions for the parameters introduced in the geometric and quantitative versions of the $NA(\mathcal{Q}^{T})$ condition and for  the set $\mathcal{P}^{T}$.  

The paper is structured as follows: Section \ref{Unc} presents  the framework and notations  needed in the paper.  Different definitions of  conditional support  which are at the heart of our study are  introduced and important measurability results established.  Section \ref{NAQNA} contains the different definitions of no-arbitrage  together with  our main result. In Section \ref{ExAp} we propose two detailed examples  illustrating the previous results and also how to build  set of probability measures which are not dominated. 
Finally,    Section \ref{Apen} collects the missing proofs.

\section{The Model}
\label{Unc}
This section presents our multiple-priors  framework and gives introductory definitions. 
\subsection{Uncertainty modeling}
The construction of the global probability space is based on a product of the local (between time $t$ and $t+1$) ones using measurable selection under Assumption \ref{QanalyticARB} below. This is tailor made for the dynamic programming approach.  

We fix  a time horizon $T\in \mathbb{N}$ and  introduce  a sequence $\left(\Omega_t\right)_{1 \leq t \leq T}$  of Polish spaces. Each $\O_{t+1}$ contains all possible scenarii between time $t$ and $t+1$. For some $1 \leq t \leq T$, we set $\Omega^{t}:=\O_{1} \times \dots \times \O_{t}$  (with the convention that $\Omega^{0}$ is reduced to a singleton),  $\mathcal{B}(\O^{t})$  its Borel sigma-algebra and  $\mathfrak{P}(\O^{t})$ the set of all probability measures on $(\O^{t},\mathcal{B} (\O^{t}))$.  An element of $\Omega^{t}$ will be denoted by $\o^{t}=(\omega_{1},\dots, \omega_{t})=(\o^{t-1},\o_{t})$ for $(\o_{1},\dots,\o_{t}) \in \Omega_{1}\times\dots\times\Omega_{t}$. 
 We also introduce the universal sigma-algebra $\mathcal{B}_{c}(\O^{t})$ which is the intersection of all possible completions of $\mathcal{B}(\O^{t})$.\\
Let $S:=\left\{S_{t},\ 0\leq t\leq T\right\}$ be a 
$\mathbb{R}^d$-valued process
where for all $0\leq t\leq T$, $S_{t}=\left(S^i_t\right)_{1 \leq i \leq d}$ represents the   price of $d$ risky securities at time $t$. We assume that there is a riskless asset whose  price is constant and equals  $1$. We also make  the following assumptions already stated in \citep{BN} to which we refer for further details and motivations on the framework.
\begin{assumption}
\label{SassARB}
The process $S$ is $\left(\mathcal{B}(\Omega^{t})\right)_{0 \leq t \leq T}$-adapted.
 \end{assumption}
Trading 
strategies are represented by  $\left(\mathcal{B}_{c}(\Omega^{t-1})\right)_{1 \leq t \leq T}$-measurable and $d$-dimensional  processes $\phi:=\{ \phi_{t}, 1 \leq t \leq T\}$ where for all $1 \leq t \leq T$, $\phi_{t}=\left(\phi^{i}_{t}\right)_{1 \leq i \leq d}$ represents the
investor's holdings in  each of the $d$ assets at time $t$.  
The set of all such trading
strategies is denoted by $\Phi$.
 The notation $\Delta
S_t:=S_t-S_{t-1}$ will often be used.
If $x,y\in\mathbb{R}^d$ then
the concatenation $xy$ stands for their scalar product. The symbol $|\cdot|$ denotes the Euclidean norm
on $\mathbb{R}^d$ (or on $\mathbb{R})$. Trading is assumed to be self-financing and  the value at time $t$ of a portfolio $\phi$ starting from
initial capital $x\in\mathbb{R}$ is given by
$$
V^{x,\phi}_t=x+\sum_{s=1}^t  \phi_s \Delta S_s.
$$

We construct the set $\mathcal{Q}^{T}$ of all possible priors in the market.  For all $0\leq t\leq T-1$, let $\mathcal{Q}_{t+1}: \Omega^t \twoheadrightarrow \mathfrak{P}(\O_{t+1})$\footnote{The notation $\twoheadrightarrow$ stands for set-valued mapping.}   where $\mathcal{Q}_{t+1}(\o^{t})$  can be seen as the set of all possible  priors for the $t$-th period given the state $\o^{t}$ until time $t$.
\begin{assumption}
\label{QanalyticARB}
For all $0\leq t\leq T-1$,  $\mathcal{Q}_{t+1}$ is a non-empty and convex valued random set such that
\begin{align*}
\mbox{graph}(\mathcal{Q}_{t+1})=\left\{(\omega^{t},P) \in \Omega^{t}\times \mathfrak{P}(\Omega_{t+1}),\; P \in \mathcal{Q}_{t+1}(\omega^{t})\right\} 
\end{align*}
is an analytic set.
\end{assumption}
Let $X$ be a  Polish space. 
An analytic set of $X$  is the continuous image of some Polish space, see  \citep[Theorem 12.24 p447]{Hitch}. We denote by $\mathcal{A}(X)$ the set of analytic sets of $X$ and recall some key properties that  will  often be used without further reference in the rest of the paper. The projection of an analytic set is an analytic set  see  (\citep[Proposition 7.39 p165]{BS}), a countable union or intersection of analytic sets is an analytic set (see \citep[Corollary 7.35.2 p160]{BS}),  the Cartesian product of analytic sets is an analytic set (see \citep[Proposition 7.38 p165]{BS}), the image or pre-image of an analytic set is an analytic set (see \citep[Proposition 7.40 p165]{BS}) and (see  \citep[Proposition 7.36 p161, Corollary 7.42.1 p169]{BS})
 \begin{align}
 \label{analyticsetARB}
 \mathcal{B}(X) \subset \mathcal{A}(X)  \subset \mathcal{B}_{c}(X).
  \end{align} However the complement of an analytic set does not need to be an  analytic set. \\
  We will also use without further references a particular case of the Projection Theorem (see  \citep[Theorem 3.23 p75]{CV77}) and of the Auman's Theorem  (see  \citep[Corollary 1]{bv}) which we recall  for sake of completeness.
Let $(X,\mathcal{T})$ be a measurable space and $Y$ be some Polish space. 
If $G \in \mathcal{T} \otimes \mathcal{B}(Y)$, then the projection of $G$ on $X$ 
$ \mbox{Proj}_{X}(G)$ belongs to $ \mathcal{T}_{c}(X),$ the completion of $\mathcal{T}$ with respect to any probability measures on $(X,\mathcal{T})$.   Let $\Gamma: \, X 
\twoheadrightarrow Y$ be such that $\mbox{graph}(\Gamma) \in \mathcal{T} \otimes \mathcal{B}(Y).$ Then there exist a $\mathcal{T}_{c}(X)-\mathcal{B}(Y)$ measurable selector $\sigma: \, X \to Y$ such that $\sigma(x) \in \Gamma(x)$ for all $x \in \{\Gamma \neq \emptyset\}$. 

From the Jankov-von Neumann Theorem (see  \citep[Proposition 7.49 p182]{BS}) and Assumption \ref{QanalyticARB},  there exists some $\mathcal{B}_{c}(\O^{t})$-measurable $q_{t+1}: \Omega^{t} \to \mathfrak{P}(\O_{t+1})$ such that for all $\o^{t} \in \O^{t}$, $q_{t+1}(\cdot,\o^{t}) \in \mathcal{Q}_{t+1}(\o^{t})$ (recall that for all $\o^{t} \in \O^{t}$, $\mathcal{Q}_{t+1}(\o^{t}) \neq \emptyset$). 
For all $1\leq t \leq T$ let  $\mathcal{Q}^{t} \subset \mathfrak{P}\left(\Omega^{t}\right)$ be defined by
\begin{align}
\label{QstarARB}
 \mathcal{Q}^{t}:=\bigl\{ Q_{1} \otimes q_{2} \otimes \dots \otimes q_{t},&\;  Q_{1} \in \mathcal{Q}_{1},\;  q_{s+1} \in \mathcal{S}K_{s+1}, \\
\nonumber & q_{s+1}(\cdot,\o^{s}) \in \mathcal{Q}_{s+1}(\o^{s}),\; \forall \,  \o^{s} \in \O^{s}, \; \forall \,  1\leq s\leq t-1\; \bigr\},
\end{align}
where $Q^{t}:=Q_{1} \otimes q_{2} \otimes \dots \otimes q_{t}$ denotes the $t$-fold application of Fubini's Theorem (see  \citep[Proposition 7.45 p175]{BS}) which defines a measure on $\mathfrak{P}\left(\Omega^{t}\right)$ and $\mathcal{S}K_{t+1}$ is the set of universally-measurable stochastic kernel on $\O_{t+1}$ given $\O^{t}$ (see \citep[Definition 7.12  p134, Lemma 7.28 p174]{BS}).\\

Apart from Assumption \ref{QanalyticARB}, no specific assumptions on the set of priors are made:  $\mathcal{Q}^{T}$  is neither assumed to be  dominated by a given  probability measure nor to be weakly compact.
This setting allows for various general definitions of the  sets $\mathcal{Q}^T$.  Section \ref{ExAp} presents some concrete examples of non-dominated settings. We  refer also to \citep{Bart16}  for other examples.

\subsection{Multiple-priors conditional supports}

\label{Multi}

The following definitions  are at the heart of our study.
 \begin{definition}
\label{DefDARB} Let $P \in \mathfrak{P}\left(\Omega^{T}\right)$ with the fixed disintegration $P:=Q_1 \otimes q_{2} \otimes \cdots \otimes q_{T}$ where $q_{t} \in \mathcal{SK}_{t}$ for all $1 \leq t \leq T$. 
For all  $0 \leq t \leq T-1$,  the random sets ${E}^{t+1}: \; \Omega^{t} \times \mathfrak{P}(\O_{t+1}) \twoheadrightarrow \mathbb{R}^{d}$, ${D}^{t+1}, 
\ {D}_{P}^{t+1} \; : \Omega^{t} \twoheadrightarrow \mathbb{R}^{d}$  are defined  for $\o^{t} \in \O^{t}$, $p \in \mathfrak{P}(\O_{t+1})$ by
\begin{align}
\label{DefE}
{E}^{t+1}(\o^{t},p)&:=\bigcap  \left\{ A \subset \mathbb{R}^{d},\; \mbox{closed}, \; p\left(\Delta S_{t+1}(\o^{t},.) \in A\right) =1\right\},\\
\label{dedD}
{D}^{t+1}(\o^{t})&:=\bigcap  \left\{ A \subset \mathbb{R}^{d},\; \mbox{closed}, \; p\left(\Delta S_{t+1}(\o^{t},.) \in A\right)=1, \; \forall \,p \in \mathcal{Q}_{t+1}(\o^{t}) \right\},\\
\label{DefPDARB1}
{D}_{P}^{t+1}(\o^{t})&:=\bigcap  \left\{ A \subset \mathbb{R}^{d},\; \mbox{closed}, \; q_{t+1}\left(\Delta S_{t+1}(\o^{t},.) \in A,\o^{t}\right) =1\right\}.
\end{align}

\end{definition}
\begin{remark}
\label{supppp}
As  $\mathbb{R}^{d}$ is second countable, 
$p\left(\Delta S_{t+1}(\o^{t},\cdot) \in E^{t+1}(\o^{t},p)\right)=1,$ see  \citep[Theorem 12.14]{Hitch} and  $p\left(\Delta S_{t+1}(\o^{t},\cdot) \in D^{t+1}(\o^{t}\right))=1$ for all $p \in \mathcal{Q}_{t+1}(\o^{t}),$ see  \citep[Lemma 4.2]{BN}.
\end{remark}
\begin{remark}
\label{DomainInc}
It is easy to verify  that for all $\o^{t} \in \O^{t}$, $p \in \mathcal{Q}_{t+1}(\o^{t})$
\begin{align}
\label{EE}
E^{t+1}(\o^{t},p) &\subset D^{t+1}(\o^{t}). 
\end{align}
Recall that any probability $P \in \mathfrak{P}(\O^{T})$ can be decomposed using Borel-measurable stochastic kernel, see for instance \citep[Corollary 7.27.2 p139]{BS}. 
Then for some fixed disintegration of $P\in\mathcal{Q}^{T}$,  $P:=Q_1 \otimes q_{2} \otimes \cdots \otimes q_{T}$,  all $0 \leq t \leq T-1$ and all $\o^{t} \in \O^{t}$ 
\begin{align}
\label{DvsE}
{D}_{P}^{t+1}(\o^{t})={E}^{t+1}(\o^{t},q_{t+1}(\cdot,\o^{t}))  \subset {D}^{t+1}(\o^{t})
\end{align}
as $q_{t}(\cdot,\o^{t}) \in \mathcal{Q}_{t+1}(\o^{t})$ for all $\o^{t} \in \O^{t}$ (see \eqref{QstarARB}).
\end{remark}

The following lemma 
establishes some important measurability properties of the random sets previously introduced and uses   the following notations. For  some $R \subset \mathbb{R}^{d}$, let \begin{align*}
\mbox{Aff}(R)&:= \bigcap\{A \subset \mathbb{R}^{d},\; \mbox{affine subspace}, \; R \subset A\}, \\  {{\mbox{Conv}}}(R)&:= \bigcap\{C\subset \mathbb{R}^{d},\; \mbox{ convex},\; R \subset C\},\;  \quad
  {\overline{\mbox{Conv}}}(R):= \bigcap\{C\subset \mathbb{R}^{d},\; \mbox{closed convex},\; R \subset C\}.
  \end{align*} 
Recall that  $\mbox{Conv}(R)= \left\{\sum_{i=1}^{n} \l_{i} p_{i},\; n \geq 1,\; p_{i} \in R,\; \sum_{i=1}^{n} \l_{i} =1, \l_{i} \geq 0\right\}$ see \citep[Theorem 2.3 p12]{cvx} and that  ${\overline{\mbox{Conv}}}(R)=\overline{{\mbox{Conv}}(R)}$. \\
 For a random set $R: \Omega \twoheadrightarrow \mathbb{R}^{d}$, 
 ${\overline{\mbox{Conv}}}\left(R\right)$ and $\mbox{Aff}\left(R\right)$ are the  random  sets defined for all $\o \in \O$ by
$ \overline{{\mbox{Conv}}}\left(R \right)(\o):=\overline{{\mbox{Conv}}}\left(R(\o)\right)  \;  \mbox{and} \; \mbox{Aff}\left(R\right)(\o):=\mbox{Aff}\left(R(\o)\right).$

\begin{lemma}
\label{Dmeasurability}
Let Assumptions \ref{SassARB} and \ref{QanalyticARB} hold true and let $0 \leq t \leq T-1$ be fixed. 
Let $P\in\mathcal{Q}^{T}$ with a fixed disintegration  $P:=Q_1 \otimes q_{2} \otimes \cdots \otimes q_{T}$. 
\begin{itemize}
\item The random sets ${E}^{t+1}$, $\overline{\mbox{Conv}}\left({E}^{t+1}\right)$, $\mbox{Aff}\left(E^{t+1}\right)$ are non-empty, closed valued and $\mathcal{B}(\O^{t}) \otimes \mathcal{B}(\mathfrak{P}\left(\O_{t+1})\right)$-measurable\footnote{See \citep[Definition 14.1]{rw}.} with  $\mbox{graphs}$   in $\mathcal{B}(\O^t) \otimes\mathcal{B}\left( \mathfrak{P}(\O_{t+1})\right) \otimes \mathcal{B}(\mathbb{R}^{d})$. 
\item The random sets  ${D}^{t+1}$, 
${D}_{P}^{t+1}$, $\overline{\mbox{Conv}}\left(D^{t+1}\right)$, 
$\overline{\mbox{Conv}}\left({D}_{P}^{t+1}\right)$, $\mbox{Aff}\left(D^{t+1}\right)$ 
and $\mbox{Aff}\left(D_P^{t+1}\right)$ are  non-empty, closed valued 
and $\mathcal{B}_{c}(\O^{t})$-measurable. Furthermore their $\mbox{graphs}$ 
belong to $\mathcal{B}_{c}(\O^t) \otimes \mathcal{B}(\mathbb{R}^{d})$.
\end{itemize}
\end{lemma}
\begin{proof}
The measurability of ${D}^{t+1}$ follows from \citep[Lemma 2.2]{BC16}.   Fix some open set $O \subset \mathbb{R}^{d}$.  Assumption \ref{SassARB}  and \citep[Proposition 7.29 p144]{BS} imply that $(\o^{t},p) \to  p\left(\Delta S_{t+1}(\o^{t},.) \in O\right)$ is $\mathcal{B}(\O^{t}) \otimes \mathcal{B}(\mathfrak{P}(\O_{t+1}))$-measurable. The measurability of $ {E}^{t+1}$  
and ${D}_{P}^{t+1}$ follows from
\begin{small}
 \begin{align*}
 \left\{(\o^{t},p)
,\; {E}^{t+1}(\o^{t},p)  \cap O \neq \emptyset \right\} &=\left\{(\o^{t},p), \; p\left(\Delta S_{t+1}(\o^{t},.) \in O\right)>0\right\}
  \in \mathcal{B}(\O^{t}) \otimes \mathcal{B}(\mathfrak{P}(\O_{t+1})),\\ 
\{\o^{t},\; {D}_{P}^{t+1}(\o^{t}) \cap O \neq \emptyset \} &=\left\{\o^{t}, \; \exists \, q \in \mathfrak{P}(\O_{t+1}),\;  q_{t+1}(\cdot,\o^{t})=q,\; {E}^{t+1}(\o^{t},q) \cap O \neq \emptyset\right\}\\
&= \mbox{Proj}_{\O^{t}} \left\{(\o^{t},q),\; q_{t+1}(\cdot,\o^{t})=q,\;{E}^{t+1}(\o^{t},q) \cap O \neq \emptyset\right\}
 \in \mathcal{B}_{c}(\O^{t}),
\end{align*}
\end{small}
\noindent where we have used Assumption \ref{QanalyticARB} and the Projection Theorem as $(\o^{t},q) \to  q_{t+1}(\cdot,\o^{t})-q$  is $\mathcal{B}_{c}(\O^{t}) \otimes \mathfrak{P}(\O_{t+1})$-measurable.\\
Then,  \citep[Proposition 14.2, Exercise 14.12]{rw} implies that   $\overline{\mbox{Conv}}\left({E}^{t+1}\right)$,  $\mbox{Aff}\left(E^{t+1}\right)$ are  $\mathcal{B}(\O^{t}) \otimes \mathcal{B}(\mathfrak{P}(\O_{t+1}))$-measurable and 
that  $\overline{\mbox{Conv}}\left(D^{t+1}\right)$, 
$\overline{\mbox{Conv}}\left({D}_{P}^{t+1}\right)$, $\mbox{Aff}\left(D^{t+1}\right)$ 
and $\mbox{Aff}\left(D_P^{t+1}\right)$  are  $\mathcal{B}_{c}(\O^{t})$-measurable.\\ 
Finally,   \citep[Theorem 14.8]{rw} implies that the $\mbox{graphs}$ of $E^{t+1}$, $\overline{\mbox{Conv}}\left({E}^{t+1}\right)$ and $\mbox{Aff}\left(E^{t+1}\right)$    belong to $\mathcal{B}(\O^t) \otimes\mathcal{B}\left( \mathfrak{P}(\O_{t+1})\right) \otimes \mathcal{B}(\mathbb{R}^{d})$ while  the $\mbox{graphs}$ of ${D}^{t+1}$, ${D}_{P}^{t+1}$, $\overline{\mbox{Conv}}\left(D^{t+1}\right)$, 
$\overline{\mbox{Conv}}\left({D}_{P}^{t+1}\right)$, $\mbox{Aff}\left(D^{t+1}\right)$, 
and $\mbox{Aff}\left(D_P^{t+1}\right)$
 belong to $\mathcal{B}_{c}(\O^{t}) \otimes \mathcal{B}(\mathbb{R}^{d}).$
\end{proof}\\
\section{No-arbitrage characterisations}
\label{NAQNA}
\subsection{Global no-arbitrage condition 
and main result}
\label{alternativeNA}
In  the uni-prior case, for any $P \in \mathcal{P}^{T},$ the no-arbitrage $NA(P)$ condition holds true if $V_{T}^{0,\phi} \geq 0$ $P$-a.s. for some $\phi \in \Phi$ implies that $V_{T}^{0,\phi}=0$ $P$-a.s. In the multiple-priors setting,
 the  no-arbitrage condition  $NA(\mathcal{Q}^{T})$, also referred as quasi-sure no-arbitrage, was  introduced in  \citep{BN}. Our main message will be that it is indeed a good assumption. Besides being a natural  extension of the classical uni-prior arbitrage condition, we will show that it is equivalent to several conditions previously used in the literature. 
\begin{definition}
\label{NAQTARB}
The $NA(\mathcal{Q}^{T})$ condition holds true if
$V_{T}^{0,\phi} \geq 0 \; \mathcal{Q}^{T}\mbox{-q.s. for some $\phi  \in \Phi$}$ implies that $ V_{T}^{0,\phi}  = 0 \;\mathcal{Q}^{T}\mbox{-q.s. }$
\end{definition}
Recall that for a given $\mathcal{P} \subset  \mathfrak{P}(\O^T)$, 
a set $N \subset \O^T$ is called a $\mathcal{P}$-polar  if for all $P \in \mathcal{P}$, there exists some $A_{P} \in \mathcal{B}_c(\O^T)$ such that $P(A_{P})=0$ and $N \subset A_{P}$. A property holds true $\mathcal{P}$-quasi-surely (q.s.), if it is true outside a $\mathcal{P}$-polar set. Finally  a set is of $\mathcal{P}$-full measure  if its complement is a $\mathcal{P}$-polar set.

\citep{BN} proves that 
Definition \ref{NAQTARB} allows a FTAP  generalisation. The  $NA(\mathcal{Q}^{T})$ is equivalent to the following: For all $Q \in \mathcal{Q}^{T}$, there exists some $P \in \mathcal{R}^{T}$ such that $ Q \ll P$ where 
\begin{align}
\label{Mart}
\mathcal{R}^{T}:=\{P \in \mathfrak{P}(\O^{T}),\; \exists \, Q^{'} \in \mathcal{Q}^{T}, P \ll Q^{'} \; \mbox{and $P$ is a martingale measure}\}.
\end{align}
The next result is straightforward.
\begin{lemma}
\label{PRpolar}
Let $\mathcal{P}$ and $\mathcal{M}$ be two sets of probability measures on $\mathfrak{P}(\O^{T})$ such that $\mathcal{P}$ and $\mathcal{M}$ have the same polar sets. Then  the $NA(\mathcal{P})$ and the $NA(\mathcal{M})$ conditions  are equivalent.
\end{lemma}


Nevertheless, it is not true that under the $NA(\mathcal{Q}^{T})$ condition,  the $NA(P)$ condition holds true for all $P \in \mathcal{Q}^{T},$ see Lemma \ref{exex} below. 
This condition is called the   ``strong no-arbitrage" or  $sNA(\mathcal{Q}^{T})$. 
\begin{definition}
\label{sNA}
The  $sNA(\mathcal{Q}^{T})$ condition holds true if  the $NA(P)$ holds true for all $P \in \mathcal{Q}^{T}$.
\end{definition}
\begin{remark} The $sNA(\mathcal{Q}^{T})$ is a strong condition. But it is   related to  practical situations in finance:  If it does not hold true, there exists  a model $P \in \mathcal{Q}^{T}$ and a strategy  $\phi  \in \Phi$   such that $V_{T}^{0,\phi} \geq 0 \; {P}\mbox{-a.s.}$ and $ P(V_{T}^{0,\phi}  > 0)>0$ and  an agent  having sold some derivative product may not be able to use different arbitrage free models to  manage the resulting position (think for instance of different volatility level to delta-hedge  a simple vanilla option).  \\
The $sNA(\mathcal{Q}^{T})$ condition is also useful to obtain tractable theorems on multiple-priors expected utility maximisation for unbounded function, see \citep[Theorem 3.6]{BC16} and \citep[Theorem 3.9]{RaMe18}. \\
Finally,   this definition seems also  relevant in a continuous time setting for studying the no-arbitrage characterisation, see \citep[Definition 2.1, Theorem 3.4]{Bia16}. 
\end{remark}

In the spirit of the model-dependent arbitrage introduced in \citep{DaHo07} (see also  Remark \ref{andrea}) we  introduce  the  notion of ``weak no-arbitrage". \begin{definition}
\label{wNA}
The  $wNA(\mathcal{Q}^{T})$ condition holds true if there exists some $P \in \mathcal{Q}^{T}$ such that the $NA(P)$ holds true.
\end{definition}
\begin{remark}
The contraposition of the $wNA(\mathcal{Q}^{T})$ condition is that for all models $P \in \mathcal{Q}^{T}$, there exists a strategy $\phi_{P}$ such that $V_{T}^{0,\phi_{P}} \geq 0 \; {P}\mbox{-a.s.}$ and $ P(V_{T}^{0,\phi_{P}}  > 0)>0$. A concrete example of  a such model-dependent arbitrage is given in  \citep{DaHo07}.  
\end{remark}
We  illustrate now the obvious relations between the three no-arbitrage conditions introduced 
(see also Figure 2). The more subtle one will be addressed in Theorems \ref{TheoS} and  \ref{PPstar}. This last theorem  shows that the $NA(\mathcal{Q}^{T})$ condition implies the $wNA(\mathcal{Q}^{T})$ one.   
\begin{lemma}
\label{exex}
\begin{enumerate}
\item  Assume that   $\mathcal{Q}^{T}=\{P\}$ for some $P \in \mathfrak{P}(\O^{T})$. Then the  $NA(\mathcal{Q}^{T}),$ $sNA(\mathcal{Q}^{T}),$ $wNA(\mathcal{Q}^{T})$ and $NA(P)$ conditions are equivalent.
\item Assume that there exists a dominating  probability measure $\widehat{P} \in \mathcal{Q}^{T}$. Then  the $NA(\mathcal{Q}^{T})$  and  $NA(\widehat{P})$ conditions are equivalent.
\item The $sNA(\mathcal{Q}^{T})$ condition  implies the $wNA(\mathcal{Q}^{T})$  but the converse does not hold true. 
\item The $sNA(\mathcal{Q}^{T})$ condition  implies  the $NA(\mathcal{Q}^{T})$ 
but  the converse does not true. 
\item The $wNA(\mathcal{Q}^{T})$ condition does not imply  the $NA(\mathcal{Q}^{T})$ condition.  
\end{enumerate}
\end{lemma}
\begin{proof}
The first item is clear. The second one follows from  Lemma \ref{PRpolar}. The first part of item 3 is trivial and it easy to construct simple counter-example for the second part (see   Example \ref{EXEX} below). We now prove  item 4. 
If the $NA(\mathcal{Q}^{T})$  condition fails, there exists some $\phi  \in \Phi$ and $P \in \mathcal{Q}^{T}$ such that $V_{T}^{0,\phi} \geq 0 \; {\mathcal{Q}^{T}}\mbox{-q.s.}$ and $P(V_{T}^{0,\phi}  > 0)>0:$ The $sNA(\mathcal{Q}^{T})$ condition also fails.  Now consider a one-period model with one risky asset  $S_0=0$, $S_1: \O \to \mathbb{R}$ (for some Polish space $\O$). Let  $P_{1}$ such that $P_{1} \left( \pm \Delta S_1 >0 \right)>0$ and   $P_{2}$ such that $P_2(\Delta S_1 \geq 0)=1$ and $P_{2}(\Delta S_1>0)>0$ and set $\mathcal{Q}=\{\l P_{1}+ (1-\l) P_{2}, \; 0< \l \leq 1\}$. Then $NA(P_{2})$ fails while $NA(\mathcal{Q})$ holds true. Note that Lemma \ref{lembin} provides another counter-example. 
Finally for item 5,  consider  a  one period model with  two risky assets $S^{1}_0=S^{2}_0=0$ and $S^{1,2}_{1}:\O \to \mathbb{R}$. Let $P_{1}$ be such that $P_{1}(\Delta S_{1}^{1} \geq 0)=1$,  $P_{1}(\Delta S_{1}^{1} > 0)>0$ and  $P_{2}$ such that $P_{2}(\Delta S_{1}^{1} = 0)=1$, $P_{2}(\pm \Delta S_{1}^{2} >0)>0$ and set $\mathcal{Q}=\{\l P_{1}+ (1-\l) P_{2}, \; 0< \l \leq 1\}$. Then  the $NA(P_{2})$  and  thus the $wNA(\mathcal{Q})$ conditions are  clearly verified. But the $NA(\mathcal{Q})$ condition does not hold true.  Indeed, let $h=(1,0)$.
     Then $h \Delta S_{1} \geq 0$ $\mathcal{Q}$-q.s. but $P_{1}(h \Delta S_{1} > 0)>0$.  Note that $\mbox{Aff}(D)=\mathbb{R}^{2}$ and $\mbox{Aff}\left(D_{P_{2}}\right)=\{0\} \times \mathbb{R}$.
\end{proof}\\

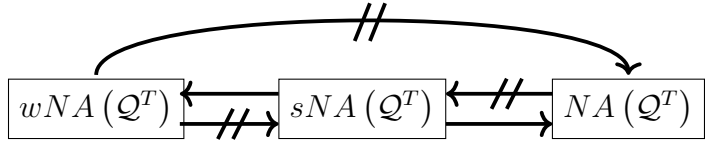
\begin{figure}
\begin{center}
 \begin{tikzpicture}
  \node[draw] (A1) at (0,0) {$wNA\left(\mathcal{Q}^{T}\right)$};
  \node[draw] (A2) at (3.5,0) {$sNA\left(\mathcal{Q}^{T}\right)$};
  \node[draw] (A3) at (7,0) {$NA\left(\mathcal{Q}^{T}\right)$};
\draw[->,black,line width=.5mm](2.4,0.2)--(1.1,0.2);  
\draw[<-,black,line width=.5mm](2.4,-0.2)--(1.1,-0.2);  
\draw[black,line width=.5mm](1.6,-0.4)--(1.8,-0.0);  
\draw[black,line width=.5mm](1.8,-0.4)--(2,-0.0);  

\draw[->,black,line width=.5mm] (0,.45)  to [out=90,in=90, looseness=.35] (7,0.45);
\draw[black,line width=.5mm](3.4,.9)--(3.6,1.4);  
\draw[black,line width=.5mm](3.6,.9)--(3.8,1.4);  
  \draw[<-,black,line width=.5mm](4.6,0.2)--(6.0,0.2);  
  \draw[->,black,line width=.5mm](4.6,-0.2)--(6.0,-0.2);  
\draw[black,line width=.5mm](5.2,0.0)--(5.4,0.4);  
\draw[black,line width=.5mm](5.4,0.0)--(5.6,0.4);  
    \end{tikzpicture}
\end{center}
 \caption{Relations between the  no-arbitrage definitions, see Lemma\ref{exex}.}
 \end{figure}



The following theorem  is our main result. 
\begin{theorem}
\label{TheoS}
Assume that Assumptions  \ref{SassARB} and \ref{QanalyticARB} hold true. The following conditions are equivalent.  \begin{itemize}
\item The $NA(\mathcal{Q}^{T})$ condition holds true.
\item There exists some $\mathcal{P}^{T} \subset \mathcal{Q}^{T}$  such that $\mathcal{P}^{T}$ and  $\mathcal{Q}^{T}$ have  the  same polar-sets  and such that  the $sNA(\mathcal{P}^{T})$ condition holds true.
\end{itemize}
\end{theorem}
Let $P^*$ as in Theorem \ref {PPstar} below with the fix  disintegration $P^*:=P_{1}^{*}\otimes p_{2}^{*} \otimes \cdots \otimes p_{T}^{*}$. The set  
$\mathcal{P}^{T}$ is defined recursively as follows: 
For all $1 \leq t \leq T-1$ 
  \begin{align}
  \begin{split}
  \mathcal{P}^{1}&:=\left\{\l P_{1}^*+ (1-\l)P,\; 0<\l\leq 1,\; P \in \mathcal{Q}^{1}  \right\},
\end{split}\\
\label{PstarARB}
\begin{split}
 \mathcal{P}^{t+1}&:=\Bigl\{P \otimes \left( \l p^*_{t+1}+ (1-\l) q_{t+1}\right),\;  0<\l\leq 1, 
\\
 & \quad \quad \quad \quad \quad P \in \mathcal{P}^{t},\;q_{t+1}(\cdot,\o^{t}) \in \mathcal{Q}_{t+1}(\o^{t}) \; \mbox{for all $\o^{t} \in \O^{t}$}  \Bigr\}.
\end{split}
\end{align}
\begin{proof}
See Section \ref{setwoARB}. 
\end{proof}
\begin{remark}
\label{marco}
\citep[Theorem 4]{Burz18} delivers a similar message but  in a completely different setup which does not rely on  a set of priors and under  the no open-arbitrage  assumption. The set $\mathcal{P}^{T}$ is replaced by the set of probability measures with full support. 
\end{remark}
\begin{remark}
\label{static}
In previous studies on robust pricing and hedging, it is often assumed that there exists some additional assets  available only for static trading (buy and hold), see for instance \citep[Theorem 5.1]{BN}. This  raises the mathematical difficulties as, roughly speaking, its breaks the dynamic consistency between time zero and future times and might prevent from obtaining a dynamic programming principle. A typical illustration of the issue arising is the so-called duality gap for American options, where  the superhedging price for an American option may be strictly larger than the supremum of its expected (discounted) payoff over all stopping times and all (relevant) martingale measures (see  for instance \citep{Bay15}, \citep{HobNeu16}, \citep{Bay17}).\\
 In our setting all assets are dynamically traded and some of them may be  derivatives products. Obviously the level of uncertainty regarding the behaviours of each assets might depend on its nature and this will be reflected in the set of prior $\mathcal{Q}^{T}$.  This  follows the spirit of the original approach developed in  \citep{Ho982} where the prices of actively traded options is taken as input. Furthermore, from a pure practical point of view, 
 we think that additional financial assets which provide useful informations for pricing should be traded at least on a daily basis. 
Hence, introducing trading constraints or transactions cost could be a better way to reflect the potential difference of liquidity between assets and derivatives. From a theoretical perspective,  \citep{AKS18} shows that   any setup as in \citep{BN} can be lifted to a setup with dynamic trading in all assets  in a way which does not introduce arbitrage (see \citep[Lemma 3.1]{AKS18}).  
 The idea is to assume that  the options are traded dynamically and to choose a  set of priors $\mathcal{Q}^{T}$ which does not impose any  assumptions about their dynamics other than these resulting from no arbitrage in the initial  setup. An admissible pricing measure in the original setup can be used to define dynamic  options prices via conditional expectations and can thus be lifted to a martingale measure in the extended setup.
\end{remark}

We now propose three applications of Theorem \ref{TheoS} which show how usefull it is. 

The first application establishes the equivalence  between the $NA(\mathcal{Q}^{T})$ condition and  the no-arbitrage condition introduced by \citep{BCK19} which studies the problem of robust maximisation of expected utility using medial limits.
  \begin{corollary}
\label{corodaniel}
Assume that Assumptions  \ref{SassARB} and \ref{QanalyticARB} hold true. The following conditions are equivalent  
\begin{itemize}
\item The $NA(\mathcal{Q}^{T})$ condition holds true.
\item For all $Q \in \mathcal{Q}^{T}$, there exists some $P \in \mathcal{P}^{T}$ such that $Q \ll P$ and such that $NA(P)$ holds true.
\item For all $Q \in \mathcal{Q}^{T}$, there exists some $P \in \mathcal{Q}^{T}$ such that $Q \ll P$ and such that $NA(P)$ holds true.
\end{itemize}
\end{corollary}
\begin{proof}
Assume that the $NA(\mathcal{Q}^{T})$ condition holds true and choose some $Q \in \mathcal{Q}^{T}$ with the fixed disintegration  $Q:=Q_1 \otimes q_{2} \otimes \cdots \otimes q_{T}.$ Let 
$$P:=\left(\frac12 P_1^* +\frac12 Q_1\right)\otimes \left(\frac12 p_2^* +\frac12 q_2\right)\otimes  \ldots \otimes  \left(\frac12 p_T^* +\frac12 q_T\right),$$
where $P^*$ is given in Theorem \ref {PPstar} with the fixed disintegration $P^*:=P_{1}^{*}\otimes p_{2}^{*} \otimes \cdots \otimes p_{T}^{*}.$ 
Then \eqref{PstarARB} implies that $P\in \mathcal{P}^{T}$ and obviously $Q \ll P.$  Now, Theorem \ref{TheoS} implies that the  $NA(P)$ condition holds true and the second assertion is proved. As $\mathcal{P}^{T} \subset \mathcal{Q}^{T},$ 
the second assertion implies the third one. 
Assume now that the third assertion holds true and let $\phi \in \Phi$ such that $V_{T}^{0,\phi} \geq 0$ $\mathcal{Q}^{T}$-q.s. Fix some $Q \in \mathcal{Q}^{T}$. Then there exists $P \in \mathcal{Q}^{T}$ such that $Q \ll P$ and such that $NA(P)$ holds true. Thus $V_{T}^{0,\phi} = 0$ $P$-a.s and also $Q$-a.s.  As this is true for all $Q \in \mathcal{Q}^{T}$, we get that  $V_{T}^{0,\phi} = 0$ $\mathcal{Q}^{T}$-q.s. 
\end{proof}\\

The  second application  allows to prove the robust FTAP from  the classical one. Our proof uses  the one-period arguments of \citep[Theorem 2.1]{Bay17} adapted to the multi-period setting. 
Let 
\begin{align}
\label{Mart2}
{\mathcal{K}}^{T}:=\{P \in \mathfrak{P}(\O^{T}),\; \exists \, Q^{'} \in \mathcal{P}^{T}, P \sim Q^{'} \; \mbox{and $P$ is a martingale measure}\}.
\end{align}
\begin{corollary}
\label{coroerhan}
Assume that Assumptions  \ref{SassARB} and \ref{QanalyticARB} hold true. The following conditions are equivalent  \begin{itemize}
\item The $NA(\mathcal{Q}^{T})$ condition holds true.
\item For all $Q \in \mathcal{Q}^{T}$, there exists some $P \in{\mathcal{K}}^{T}$   such that $Q \ll P.$
\item For all $Q \in \mathcal{Q}^{T}$, there exists some $P \in \mathcal{R}^{T}$ (see \eqref{Mart}) such that $Q \ll P.$ 
\end{itemize}
\end{corollary}
Note that this is a refinement of the version of \citep{BN} as we  have more information about the measure $P.$ \\
\begin{proof}
Assume that the $NA(\mathcal{Q}^{T})$ condition holds true. 
Corollary \ref{corodaniel} implies   that for  all $Q \in \mathcal{Q}^{T}$ there  exists some $Q' \in \mathcal{P}^{T}$ such that $Q \ll Q'$ and such that $NA(Q')$ holds true. Now the classical FTAP (see \citep{dmw}) establishes  the existence of some $P \sim Q'$ such that $P$ is a martingale measure.  Thus $P \in {\mathcal{K}}^{T}$. 
As  $Q \ll P$,  the second assertion holds true. As ${\mathcal{K}}^{T} \subset  \mathcal{R}^{T}$, the second assertion implies the third one. Assume now that the third assumption holds true and 
let $\phi \in \Phi$ such that $V_{T}^{0,\phi} \geq 0$ $\mathcal{Q}^{T}$-q.s. Fix some $Q \in \mathcal{Q}^{T}$. Then there exists $P  \in \mathfrak{P}(\O^{T})$ and $Q^{'} \in \mathcal{Q}^{T}$ such that $Q \ll P$, $P \ll Q^{'}$ and $P$ is a martingale measure.    As  $V_{T}^{0,\phi} \geq 0$ $Q'$-a.s and thus $P$-a.s. and 
$E_{P} (V_{T}^{0,\phi})= 0$, we get that 
 $V_{T}^{0,\phi} = 0$ $P$-a.s and also $Q$-a.s.  As this is true for all $Q \in \mathcal{Q}^{T}$, we obtain that  $V_{T}^{0,\phi} = 0$ $\mathcal{Q}^{T}$-q.s. 
\end{proof}\\

Lastly, Theorem \ref{TheoS} allows to obtain a tractable theorem on maximisation of expected utility under the $NA(\mathcal{Q}^{T})$ condition avoiding the difficult \citep[Assumption  2.1]{BC16}.  Note  that the no-arbitrage condition is indeed related to the utility maximisation problem in the  uni-prior case  (see  for instance \citep{Rog94}).  In the robust case, it is not clear whether a similar approach could work. This is the subject of further research.\\
A {random utility} $U$ is a function  defined on $\Omega^T \times (0,\infty)$ taking values in $\mathbb{R}\cup \{-\infty\}$ such that for every $x\in \mathbb{R}$,  $U \left(\cdot,x\right)$  is $\mathcal{B}(\Omega^{T})$-measurable and for every $\o^{T} \in {\Omega}^{T}$, $U(\o^{T}, \cdot)$ is proper\footnote{There exists $x \in(0,+\infty)$ such that  $U(\o^{T}, x)>-\infty$ and $U(\o^{T}, x)<+\infty$ for all $x \in(0,+\infty)$.}, 
non-decreasing and  concave  on $(0,+\infty)$. We extend $U$ by (right) continuity in $0$ and set  $U(\cdot,x)=-\infty$ if $x<0$. \\
Fix some $x\geq 0$.
For $P \in \mathfrak{P}(\O^{T})$ fixed, we denote by  $\Phi(x,U,P)$ the set of all strategies $\phi \in \Phi$ such that $V_{T}^{x,\phi}(\cdot)\geq 0$ $P$-a.s. and  such that either $E_{P}U^{+}(\cdot,V_{T}^{x,\phi}(\cdot))<\infty$ or $E_{P}U^{-}(\cdot,V_{T}^{x,\phi}(\cdot))<\infty$. 
Then 
$
\Phi(x,U,\mathcal{Q}^{T}):= \bigcap_{P \in \mathcal{Q}^{T}} \Phi(x,U,P). 
$ 
The set $\Phi(x,U,\mathcal{P}^{T})$ is defined similarly changing $\mathcal{Q}^{T}$ by $\mathcal{P}^{T}$  where $\mathcal{P}^{T}$ is defined in \eqref{PstarARB}. 
The {multiple-priors portfolio problem} with initial wealth $x \geq 0$  is
	\begin{align}\label{eq:OP}
		u(x) :=  \sup_{\phi \in \Phi(x,U,\mathcal{Q}^{T})} \inf_{P \in \mathcal{Q}^{T}} E_{P} U(\cdot,V^{x,\phi}_{T}(\cdot)).
	\end{align}
We also define
\begin{align}
		u^{\mathcal{P}}(x) :=  \sup_{\phi \in \Phi(x,U,\mathcal{P}^{T})} \inf_{P \in \mathcal{P}^{T}} E_{P} U(\cdot,V^{x,\phi}_{T}(\cdot)).
\end{align}
Let for all $1 \leq t \leq T$
 $$\mathcal{W}_{t}:= \bigcap_{r>0} \left \{ X: \Omega^{t} \to \mathbb{R}\cup \{\pm \infty\}, \; \mbox{$\mathcal{B}(\O^{t})$-measurable},\; \sup_{P \in \mathcal{Q}^{t}}E_{P} |X|^{r} <\infty \right\}.$$
\begin{assumption}
\label{assW}
We have that  $U^{+}(\cdot,1), U^{-}(\cdot,\frac{1}{4}) \in \mathcal{W}_{T}$  and   $\Delta S_{t}, {1}/{\alpha^{P}_{t}} \in \mathcal{W}_{t}$ for all $1 \leq t \leq T$ and $P \in \mathcal{P}^{t}$  (see Remark \ref{singleP2} for the definition of $\alpha_{t}^{P}$).
\end{assumption}
The first lemma shows the equality between both value functions.  
\begin{lemma}
\label{lemutil}
Assume that the $NA(\mathcal{Q}^{T})$ condition and  Assumptions \ref{SassARB} and  \ref{QanalyticARB} hold true. Furthermore, assume that 
$U$ is either  bounded from above or that Assumption \ref{assW} holds true.  
Then $u(x)=u^{\mathcal{P}}(x)$ for all $x \geq 0.$
\end{lemma}
\begin{proof}
Fix $x\geq 0$. Theorem \ref{TheoS} will be  in force. 
Let $P^*$ be given by Theorem \ref {PPstar} with the fixed disintegration $P^*:=P_{1}^{*}\otimes p_{2}^{*} \otimes \cdots \otimes p_{T}^{*}.$ 
First we show that $\Phi(x,U,\mathcal{Q}^{T})=\Phi(x,U,\mathcal{P}^{T})$. 
The first inclusion follows from 
$\mathcal{P}^{T} \subset \mathcal{Q}^{T}.$ As $\mathcal{P}^{T}$ and $ \mathcal{Q}^{T}$ have the same polar sets, 
$V_{T}^{x,\phi}(\cdot)\geq 0$ $\mathcal{Q}^{T}$-q.s. and $V_{T}^{x,\phi}(\cdot)\geq 0$ $\mathcal{P}^{T}$-q.s. are equivalent. 
So to prove the reverse inequality it is enough to show that 
for  $\phi \in \Phi(x,U,\mathcal{P}^{T})$ 
$E_{Q}U^{+}(\cdot,V_{T}^{x,\phi}(\cdot))<\infty$ or $E_{Q}U^{-}(\cdot,V_{T}^{x,\phi}(\cdot))<\infty$ for any $Q \in \mathcal{Q}^{T}$.  It is obviously true
if $U$ is bounded from above. Assume now that Assumption \ref{assW} holds true.  Let $Q \in \mathcal{Q}^{T}$ with  the fixed disintegration $Q:=P_1\otimes q_2\otimes  \ldots \otimes q_T$ and choose
$$R:=\left(\frac12 P_1^* +\frac12 P_1\right)\otimes \left(\frac12 p_2^* +\frac12 q_2\right)\otimes  \ldots \otimes  \left(\frac12 p_T^* +\frac12 q_T\right).$$
Then $R\in \mathcal{P}^{T},$ see \eqref{PstarARB}.  Assume that $E_{R}U^{+}(\cdot,V_{T}^{x,\phi}(\cdot))<\infty$ (the same argument applies to the negative part). Then 
\begin{align}
\frac1{2^T}E_QU^{+}(\cdot,V_{T}^{x,\phi}(\cdot)) & \leq E_RU^{+}(\cdot,V_{T}^{x,\phi}(\cdot)) <\infty. 
\end{align}
Thus 
\begin{align}
		u(x) =  \sup_{\phi \in \Phi(x,U,\mathcal{P}^{T})} \inf_{P \in \mathcal{Q}^{T}} E_{P} U(\cdot,V^{x,\phi}_{T}(\cdot)).
	\end{align}
Next we show that for all $x\geq 0$ and $\phi \in  \Phi(x,U,\mathcal{P}^{T})$ 
\begin{align}
\label{white}
		u(x,\phi) :=  \inf_{P \in \mathcal{Q}^{T}} E_{P} U(\cdot,V^{x,\phi}_{T}(\cdot))= \inf_{P \in \mathcal{P}^{T}} E_{P} U(\cdot,V^{x,\phi}_{T}(\cdot)) =:u^{\mathcal{P}}(x,\phi). 
	\end{align}
As $\mathcal{P}^{T} \subset \mathcal{Q}^{T},$ $u^{\mathcal{P}}(x,\phi)  \geq u(x,\phi)$. 	
Let $Q \in \mathcal{Q}^{T}$  with  the fixed disintegration $Q:=P_1\otimes q_2\otimes  \ldots \otimes q_T.$ Let 
$$P^n:=\left(\frac1n P_1^* +\left( 1- \frac1n \right) P_1\right)\otimes \left(\frac1n p_2^* +\left( 1- \frac1n \right)q_2\right)\otimes  \ldots \otimes  \left(\frac1n p_T^* +\left( 1- \frac1n \right) q_T\right).$$
Then \eqref{PstarARB} implies that $P^n \in \mathcal{P}^{T},$ 
\begin{align}
\label{black}
u^{\mathcal{P}}(x,\phi)  \leq 
 E_{P^n} U(\cdot,V^{x,\phi}_{T}(\cdot))
 \end{align}
and 
the only term in $E_{P^n}U(\cdot,V_{T}^{x,\phi}(\cdot))$  that is not multiplied  by $1/n$ is $(1-1/n)^{T}E_Q U(\cdot,V_{T}^{x,\phi}(\cdot)).$
Moreover, \eqref{PstarARB} implies that all the others probability measures appearing in $E_{P^n}U(\cdot,V_{T}^{x,\phi}(\cdot))$ belongs to $\mathcal{P}^{T}.$  
Fix $R \in \mathcal{P}^{T}$  as one of this measures and note that $\phi \in \phi(x,U,R)$. 
Theorem \ref{TheoS} implies that the $sNA(\mathcal{P}^{T})$   and also the $NA(R)$ conditions  hold true. 
We first prove that 
$E_{R}U^{+}(\cdot, V_{T}^{x,\phi}(\cdot))< \infty.$ If $U$ is bounded from above this is immediate. Assume that Assumption \ref{assW} holds true. Then  \citep[Theorem 4.17]{BCR18} shows that for ${R}$-almost all $\o^{T} \in \O^{T}$, 
\begin{align}
\label{rich}
|V_{T}^{x,\phi}(\o^{T})| \leq   \prod_{s=1}^{T}\left(x+ \frac{|\Delta S_{s}(\o^{s})|}{\alpha^{R}_{s-1}(\o^{s-1})}\right)=: \frac\l 2 \in \mathcal{W}_{T}
\end{align}
as  $\Delta S_{s}, \; \frac{1}{\alpha_{s}^{R}} \in \mathcal{W}_{s}$ for all $s \geq 1$. 
Suppose that $x\geq 1$ else by  monotonicity of $U^{+}$, one may replace $x$ by 1. Then \citep[Proposition 3.24]{BC16} (as $\l \geq 1$) implies that 
\begin{align*}
E_{R}U^{+}(\cdot, V_{T}^{x,\phi}(\cdot)) &\leq  4 E_{R} \left(\prod_{s=1}^{T}\left(x+ \frac{|\Delta S_{s}(\cdot)|}{\alpha^{R}_{s-1}(\cdot)}\right)\left( U^{+}(\cdot,1) +U^{-}(\cdot,\frac{1}{4})\right)\right) < \infty,
\end{align*}
as $U^{+}(\cdot,1)$, $U^{-}(\cdot,\frac{1}{4})$ $\in$ $\mathcal{W}_{T}.$\\
Now if  $E_{R}U^{-}(\cdot, V_{T}^{x,\phi}(\cdot)) =-\infty$, as $R \in \mathcal{P}^{T},$ we get that 
 $u(x,\phi) \leq u^{\mathcal{P}}(x,\phi)  =-\infty.$ Thus $u^{\mathcal{P}}(x,\phi)=u(x,\phi)$. Else 
letting $n$ go to infinity in \eqref{black} we obtain that  
$u^{\mathcal{P}}(x,\phi) \leq 
 E_{Q} U(\cdot,V^{x,\phi}_{T}(\cdot))$ and  taking the infimum over all $Q \in \mathcal{Q}^{T}$,   $u^{\mathcal{P}}(x,\phi) \leq u(x,\phi)$: \eqref{white} is proved. \\
Finally taking in \eqref{white} the supremum over all  $\phi \in \Phi(x,U,\mathcal{P}^{T})$, we get that $u(x)=u^{\mathcal{P}}(x).$ 
\end{proof}\\

To state the corollary on the existence of an optimal solution for \eqref{eq:OP}, we need two additional assumptions. 
\begin{assumption}
\label{Sass2}
There exists some $0 \leq s <\infty$ such that $-s \leq S^{i}_{t}(\o^{t}) <+\infty $ for all $1 \leq i \leq d$, $\o^{t} \in \O^{t}$ and  $0 \leq t \leq T$.
 \end{assumption}
\begin{assumption}\label{Uminus}
For all $r \in \mathbb{Q}$, $r>0,$ 
$\sup_{P \in \mathcal{Q}^{T}} E_{P} U^{-}(\cdot,r) <+\infty.$
\end{assumption}
\begin{corollary}
\label{coroutil}
Assume that the $NA(\mathcal{Q}^{T})$ condition and  Assumptions \ref{SassARB},  \ref{QanalyticARB}, \ref{Sass2} and  \ref{Uminus} 
hold true. Furthermore, assume that 
$U$ is either  bounded from above or that Assumption \ref{assW} holds true.     Let $x\geq 0$. Then,
there exists some optimal strategy $\phi^* \in \Phi(x,U, \mathcal{Q}^{T})$ such that
$$u(x) =  \inf_{P \in \mathcal{Q}^{T}} E_{P} U(\cdot,V^{x,\phi^*}_{T}(\cdot))<\infty.$$
\end{corollary}
\begin{proof}
Fix some $x \geq 0.$ Theorem \ref{TheoS} implies that $sNA(\mathcal{P}^{T})$ holds true. 
 So \citep[Theorem 3.6]{BC16}  gives the existence of an optimal strategy for $u^{\mathcal{P}}(x)$.  Lemma \ref{lemutil} allows to conclude since  $u(x)=u^{\mathcal{P}}(x)$. 
\end{proof}

\subsection{Local no-arbitrage conditions and further 
results} 
\label{local}
We now turn to local conditions which are at the heart of the proofs due to the structure of the model. 
We  recall the first  part of \citep[Theorem 4.5]{BN} which establishes the essential link between the global version   $NA(\mathcal{Q}^{T})$ and its local version.

\begin{theorem}
\label{bnlocal}
Assume that Assumptions  \ref{SassARB} and \ref{QanalyticARB} hold true. Then the following  statements are equivalent.\\
1. The $NA(\mathcal{Q}^{T})$ condition hold true.\\
2. For all $0 \leq t \leq T-1$, there exists a $\mathcal{Q}^{t}$-full measure set $\O^{t}_{NA} \in \mathcal{B}_{c}(\O^{t})$ such that for all $\o^{t} \in \O^{t}_{NA}$, $h\Delta S_{t+1}(\o^{t},\cdot) \geq 0 \; \mathcal{Q}_{t+1}(\o^{t}) \mbox{-q.s.}$ for some $h \in \mathbb{R}^{d}$ implies that $h\Delta S_{t+1}(\o^{t},\cdot) = 0 \;\mathcal{Q}_{t+1}(\o^t)\mbox{-q.s.}$
\end{theorem}

We present two other local  definitions of  no-arbitrage and establish their equivalence with  the $NA(\mathcal{Q}^{T})$ conditions in Theorem  \ref{ARBequival}  which is   an analogous of Theorem \ref{bnlocal}. \\
The first definition proposes a geometric view of the no-arbitrage.  Theorem \ref{ARBequival} extends the uni-prior result of  \citep[Theorem 3g)]{JS98}, see also \citep[Proposition 2.1.6]{KaS}.  Note that the geometric no-arbitrage has appeared in different multiple-priors contexts, see  \citep[Proposition 6.4]{OW18} and \citep[Corollary 21]{Burz18}.  A similar idea was already exploited in \citep[Lemma 3.3]{BN}. Theorem \ref{ARBequival} will also allow us to prove  Proposition \ref{alphatmesARB} and  Theorem \ref{PPstar}.\\
Recall that for a convex set $C \subset \mathbb{R}^{d}$, the relative interior of $C$ (see \citep[Section 6]{cvx}) is $\mbox{Ri}(C)=\{y \in C, \, \exists \, \varepsilon >0,\;  \mbox{Aff} (C) \cap B(y,\varepsilon) \subset C\}$ 
where $B(y,\varepsilon)$ is the open ball in $\mathbb{R}^{d}$ centered in $y$ with radius $\varepsilon$.
Moreover for a convex-valued random set $R,$   $\mbox{Ri}\left(R\right)$ is the random set defined by $\mbox{Ri}\left(R\right)(\o):=\mbox{Ri}\left(R(\o)\right)$ for $\o \in \O$.
 \begin{definition}
 \label{NAGarb} The geometric no-arbitrage condition holds true if  for all $0\leq t\leq T-1$, there exists some $\mathcal{Q}^{t}$-full measure set $\Omega^{t}_{gNA} \in \mathcal{B}_{c}(\O^{t})$   such that for all $\omega^{t} \in \Omega^{t}_{gNA}$,  $0 \in \mbox{Ri}\left({\mbox{Conv}}(D^{t+1})\right)(\o^{t})$. In this case for all $\o^{t} \in \Omega^{t}_{gNA}$, there exists  $\varepsilon_{t}(\omega^{t})>0$ such that
 \begin{align}
\label{vamakiARB}
 B(0,\varepsilon_{t}(\o^{t})) \cap \mbox{Aff}\left(D^{t+1}\right)(\o^{t}) \subset {{\mbox{Conv}}}\left(D^{t+1}\right)(\o^{t}).
 \end{align}
 \end{definition}The  geometric (local) no-arbitrage condition is indeed practical: Together with Theorem \ref{ARBequival} it allows to check whether the (global) $\mbox{NA}(\mathcal{Q}^T)$ condition  holds true or not.  
As 
$\mathcal{Q}^T$ and for all $1\leq t \leq T,$ $\Delta S_{t+1}$ are given one gets $ \mbox{Ri}\left({\mbox{Conv}}(D^{t+1})\right)(\cdot)$ and it is easy to check whether $0$ is in it or not (see Section \ref{ExAp} for examples of such a reasoning). 

Secondly,  in the spirit of  \citep[Proposition 3.3]{RS05}   (see also \citep[Proposition 2.3]{BC16}), we introduce the so-called  quantitative  no-arbitrage condition. 
 \begin{definition}
 \label{NAQarb} The quantitative no-arbitrage condition holds true if  for all $0\leq t\leq T-1$, there exists some $\mathcal{Q}^{t}$-full measure set $\Omega^{t}_{qNA} \in \mathcal{B}_{c}(\O^{t})$   such that for all $\omega^{t} \in \Omega^{t}_{qNA}$,  
there exists $\beta_{t}(\omega^{t}),\kappa_{t}(\omega^{t}) \in (0,1)$ such that for all  $h \in  \mbox{Aff}\left({D}^{t+1}\right)(\omega^{t})$ , $h \neq 0$  there exists $p_{h} \in \mathcal{Q}_{t+1}(\o^{t})$ satisfying
\begin{align}
\label{valakiARB}
p_{h}\left({h}\Delta S_{t+1}(\omega^{t},\cdot) < -\beta_{t}(\omega^{t}){|h|}\right) \geq \kappa_{t}(\omega^{t}).
\end{align}
\end{definition}
In the case where there is only one risky asset and one period,  \eqref{valakiARB} is interpreted as follows~: There exists a prior  $p^{+}$ for which the price of the risky asset increases enough and an other one $p^{-}$ for which it decreases i.e. $p^{\mp} \left( \pm \Delta S(\cdot)<-\beta\right) \geq \kappa$ where $\beta, \kappa \in (0,1)$. The number $\kappa$ serves as a measure of the gain/loss probability  and the number $\beta$ of their size. 
\begin{remark}
\label{remcomm}
Definition \ref{NAQarb} is the direct adaptation to the multiple-priors set-up of \citep[Proposition 3.3]{RS05}: The probability measure depends of the strategy. For an agent buying or selling some quantity of risky assets, there is always a prior in which she is exposed to a potential loss.  Proposition \ref{finally} will show that one can in fact choose a comment prior for all strategies in Definition \ref{NAQarb}. 
\end{remark}
\begin{remark}
\label{rembornes}
Theorem \ref{TheoS} and Proposition  \ref{finally} are precious for solving the problem of maximisation of expected utility. For example when the utility function $U$ is defined on $(0,\infty)$ they provide natural bounds for the one step strategies or for $U(V_T^{x,\Phi})$, see \eqref{rich} and \citep[Lemma 3.11 and (44)]{BC16}.   This is used to prove the existence of the optimal strategy but it could also be used to compute it numerically. We propose in Section \ref{ExAp} explicit values for $\beta_{t}$ and $\kappa_{t}$. 
\end{remark}
\begin{remark}
\label{gammabeta}
In \eqref{valakiARB}, $\beta_{t}(\omega^{t})$ provides information on $D^{t+1}(\o^{t})$ while $\kappa_{t}(\o^{t})$ provides information on $\mathcal{Q}_{t+1}(\o^{t})$.
Moreover, Definition \ref{NAQarb}  can equivalently be formulated as follow:  For all $0\leq t\leq T-1$, there exists some $\mathcal{Q}^{t}$-full measure set $\Omega^{t}_{qNA} \in \mathcal{B}_{c}(\O^{t})$   such that for all $\omega^{t} \in \Omega^{t}_{qNA}$,  
there exists  $\alpha_{t}(\o^{t}) \in (0,1)$  such that for all  $h \in  \mbox{Aff}\left({D}^{t+1}\right)(\omega^{t})$ , $h \neq 0$  there exists $p_{h} \in \mathcal{Q}_{t+1}(\o^{t})$ satisfying
\begin{align}
\label{valakitwo}
p_{h}\left({h}\Delta S_{t+1}(\omega^{t},\cdot) < -\alpha_{t}(\omega^{t}){|h|}\right) \geq \alpha_{t}(\omega^{t}).
\end{align}
Indeed,  \eqref{valakitwo} implies  \eqref{valakiARB} and  assuming \eqref{valakiARB}, \eqref{valakitwo} is true with $\alpha_{t}(\o^{t})=\min(\kappa_{t}(\omega^{t}),\beta_{t}(\o^{t})) \in (0,1)$.
\end{remark}
 \begin{theorem}
\label{ARBequival}
Assume that Assumptions  \ref{SassARB} and \ref{QanalyticARB} hold true. Then the $NA(\mathcal{Q}^{T})$  condition  (see Definition \ref{NAQTARB}),  the geometric no-arbitrage (see Definition  \ref{NAGarb}) and  the quantitative no-arbitrage (see Definition  \ref{NAQarb})   are equivalent  and one can choose $\O^{t}_{NA}= \O^{t}_{qNA}=\O^{t}_{gNA}$ for all $0 \leq t \leq T-1$. Furthermore, one can choose $\beta_{t}={\varepsilon_{t}}/{2}$ in \eqref{valakiARB} (for $\varepsilon_{t}$  introduced in \eqref{vamakiARB}).
\end{theorem}
\begin{proof}
See Section \ref{proofmesepsi}. 
\end{proof}\begin{remark}
\label{ovect1}
Under Assumptions  \ref{SassARB} and \ref{QanalyticARB} and any of the no-arbitrage condition,  $0 \in \mbox{Conv}\left(D^{t+1}\right)(\o^{t})$ and  $\mbox{Aff}\left(D^{t+1}\right)(\o^{t})$ is a vector space for all $\o^{t} \in \O^{t}_{NA}$.
\end{remark}
The next proposition  is \citep[Theorem 3]{JS98} but could also be obtained as a direct application of Theorem \ref{ARBequival} together with   \citep[Lemma 7.28 p174]{BS} and  \citep[Theorem 12.28]{Hitch} in the specific setting where $\mathcal{Q}^{T}=\{P_{1} \otimes p_{2} \otimes \dots \otimes p_{T}\}$.  Indeed, Theorem \ref{ARBequival} does not apply directly as $\mbox{graph}(p_{t})$ belongs a priory to $\mathcal{B}_{c}\left(\O^{t} \times \mathfrak{P}(\O_{t+1})\right)$ and not to $\mathcal{A}\left(\O^{t} \times  \mathfrak{P}(\O_{t+1})\right)$, and one needs to build  some  Borel-measurable version of $p_{t}$. Proposition \ref{singleP}   will be used in the sequel to prove that the $NA(P)$ condition holds true.
\begin{proposition}
\label{singleP}
Assume that  Assumption \ref{SassARB} holds true and let $P \in   \mathfrak{P}(\O^{T})$ with the fixed disintegration $P:=P_{1} \otimes p_{2} \otimes \dots \otimes p_{T}$ where  $p_{t} \in  \mathcal{S}K_{t}$ for all  $1 \leq t \leq T$. Then the $NA(P)$ condition holds true  if and only if  $0 \in \mbox{Ri}\left(\mbox{Conv}\left(D_{P}^{t+1}\right)\right)(\cdot)$ $P^{t}$-a.s. for all $0 \leq t \leq T-1$.
\end{proposition}
\begin{remark}
\label{singleP2}
Similarly, under the assumption of Proposition \ref{singleP},   one can show that the $NA(P)$ condition holds true  if and only if  the quantitative no-arbitrage holds true for  $\mathcal{Q}^{T}=\{P\}$ which is exactly  \citep[Proposition 3.3]{RS05}. In this case, we denote $\alpha_t$ in \eqref{valakitwo} by $\alpha^P_t$.
\end{remark}

\noindent We now establish   some tricky measurability properties.
\begin{proposition}
\label{alphatmesARB}
Assume that Assumptions  \ref{SassARB} and \ref{QanalyticARB} hold true. Under one of the no-arbitrage conditions (see Definitions \ref{NAQTARB}, \ref{NAGarb} and \ref{NAQarb}) one can choose a $ \mathcal{B}_{c}(\O^{t})$-measurable version of $\varepsilon_{t}$ (in \eqref{vamakiARB}) and $\beta_t$ (in \eqref{valakiARB}). \end{proposition}
\begin{proof}
See Section \ref{proofmesepsi}. 
\end{proof}
\begin{remark}
\label{alphabeta}
The measurability of $\kappa_{t}$ cannot be directly inferred from the one of $\varepsilon_{t}$ but will be obtained  in  Proposition \ref{finally} as a consequence of Theorem \ref{TheoS}.  The measurability of $\kappa_{t}$  is  useful to solve the problem of multi-priors optimal investment for unbounded utility function defined on the whole real-line since the bounds on the optimal strategies depends on $\kappa_{t}$ 
 see for instance \citep[(17)]{RS05} in a non-robust setting and  \citep[Proof of Lemma 3.3]{RaMe18} in the robust context. 
\end{remark}

The next theorem is crucial. It is  a  first step towards Theorem \ref{TheoS}:  It gives the existence of the measure $P^*$  which allows to  build recursively the set 
$\mathcal{P}^{T}$ (see \eqref{PstarARB}).  But   it is also of own interest since it gives the equivalence between $NA(\mathcal{Q}^{T})$ and  a stronger form of $wNA(\mathcal{Q}^{T})$.

\begin{theorem}
\label{PPstar}
Assume that Assumptions  \ref{SassARB} and \ref{QanalyticARB} hold true. The $NA(\mathcal{Q}^{T})$ condition holds true if and only if there exists
 some $P^* \in \mathcal{Q}^{T}$ such that    $\mbox{Aff}\left(D_{P^*}^{t+1}\right)(\o^{t})=\mbox{Aff} \left(D^{t+1}\right)(\o^{t})$ and $0 \in \mbox{Ri} \left({\mbox{Conv}}(D_{P^*}^{t+1})\right)(\o^{t})$  for all $0 \leq t \leq T-1$, $\o^t \in \O_{NA}^{t}$ \footnote{The set $\O^{t}_{NA}$ was introduced in  Theorem \ref{bnlocal}, see also \eqref{OTNA}.}. 
\end{theorem}
\begin{proof}
See Section \ref{setwoARB2}. 
\end{proof}
\begin{remark}
\label{erhan}
Theorem \ref{PPstar} was proved in a one period setting in \citep[Lemma 2.2]{Bay17}. 
\end{remark}
\begin{remark}
\label{passeul}
The probability measure $P^*$ of Theorem \ref{PPstar} is not unique. In fact, under $NA(\mathcal{Q}^{T})$, all 
 $P \in \mathcal{P}^{T}$ satisfy   $\mbox{Aff}\left(D_{P}^{t+1}\right)(\o^{t})=\mbox{Aff} \left(D^{t+1}\right)(\o^{t})$ and $0 \in \mbox{Ri} \left({\mbox{Conv}}(D_{P}^{t+1})\right)(\o^{t})$  for all $0 \leq t \leq T-1$, $\o^t \in \O_{NA}^{t},$ see proof of Theorem \ref{TheoS} step 2 $ iii).$
 \end{remark}
 
\begin{remark}
\label{PstarRem}
The main (and difficult) point  in Theorem \ref{PPstar} is that $P^* \in \mathcal{Q}^{T}$.  Thus any $\mathcal{Q}^{t}$-null set is also a $P^*$-null set and in particular $\O^{t}_{NA}$ is of $P^*$-full measure (see Theorem  \ref{bnlocal}). So   $0 \in \mbox{Ri} \left({\mbox{Conv}}(D_{P^*}^{t+1})\right)(\o^{t})$ for $\o^{t} \in \O^{t}_{NA}$ and  the $NA(P^*)$ condition holds true (see  Proposition \ref{singleP}).
We have actually more since  $\O^{t}_{NA}$ is of $\mathcal{Q}^{t}$-full measure.
We will provide in Section \ref{ExAp} explicit form of $P^*$. 
\end{remark}

\begin{remark}
\label{OW}
Theorem \ref{PPstar} is related and complements \citep[Theorem 3.1]{OW18}. Indeed, in  both cases the main issue is to find some $p_{t+1}^*(\cdot,\o^{t}) \in \mathcal{Q}_{t+1}(\o^{t})$ such that $0 \in \mbox{Ri} \left({\mbox{Conv}}(D_{P^*}^{t+1})\right)(\o^{t}) \subset  \mbox{Ri} \left({\mbox{Conv}}(D^{t+1})\right)(\o^{t})$ (recall  \eqref{DvsE}). This is used in  \citep{OW18} to make the link with the quasi-sure setting and in our case to establish Theorem \ref{TheoS}.
 \end{remark}
 
 \begin{remark}
\label{andrea}
\citep[Assumption 2.1]{RaMe18} asserts that there exists at least one arbitrage free model (in the uni-prior sense) and that for this model the affine space generated by the conditional support always equals $\mathbb{R}^{d}$. Those are the conditions verified by $P^*$ in Theorem \ref{PPstar} and thus \citep[ Theorem 3.7]{RaMe18} which shows the existence in the problem of maximisation of expected   utility for bounded function defined on the whole real line works under $NA(\mathcal{Q}^{T}).$  

 \end{remark}
 \begin{example}
 \label{EXEX} 
 The probability measure $P^* \in \mathcal{Q}^{T}$ of Theorem \ref{PPstar} provides a kind of stronger  $NA(P^*)$. 
The counter example of the last item in Lemma \ref{exex} illustrates why  the condition $\mbox{Aff} \left(D_{P^*}^{t+1}\right)(\cdot)=\mbox{Aff} \left(D^{t+1}\right)(\cdot)$ $\mathcal{Q}^{t}$-q.s. is needed in  Theorem \ref{PPstar}.  However this is not enough to obtain equivalence with the $NA(\mathcal{Q}^{T})$ condition and the following counterexample illustrates why $0 \in \mbox{Ri} \left({\mbox{Conv}}(D_{P^*}^{t+1})\right)(\cdot )$ $\mathcal{Q}^{t}$-q.s. is needed and why $0 \in \mbox{Ri} \left({\mbox{Conv}}(D_{P^*}^{t+1})\right)(\cdot )$ ${P}_{t}^*$-p.s. is not enough.\\
Let $T=2$, $d=1$, $\O_{1}:=\O_{2}:=\{-1,0,1\}$, $S_0:=2$, $S_{1}(\o_{1}):=2+\omega_{1}$, $S_{2}(\o_{1},\o_{2}):=2+\omega_{1}+\o_{2}$. Let $P_{na}:=\frac{1}{2}(\delta_{-1}+\delta_{1}) $,  $P_0:=\delta_{0}$ and  $P_1:=\delta_{1}$ be three probability measures on $\mathfrak{P}(\O_{1})$.  Set $\mathcal{Q}_{1}:=\mbox{Conv}\left(P_0,P_{na}\right)$ and  define  $\mathcal{Q}_{2}(\cdot)$ as follow:  $\mathcal{Q}_{2}(\pm 1)=\{P_{na}\}$ and   $\mathcal{Q}_{2}(0)=\{P_1\}$. This is illustrated in Figure  2.
   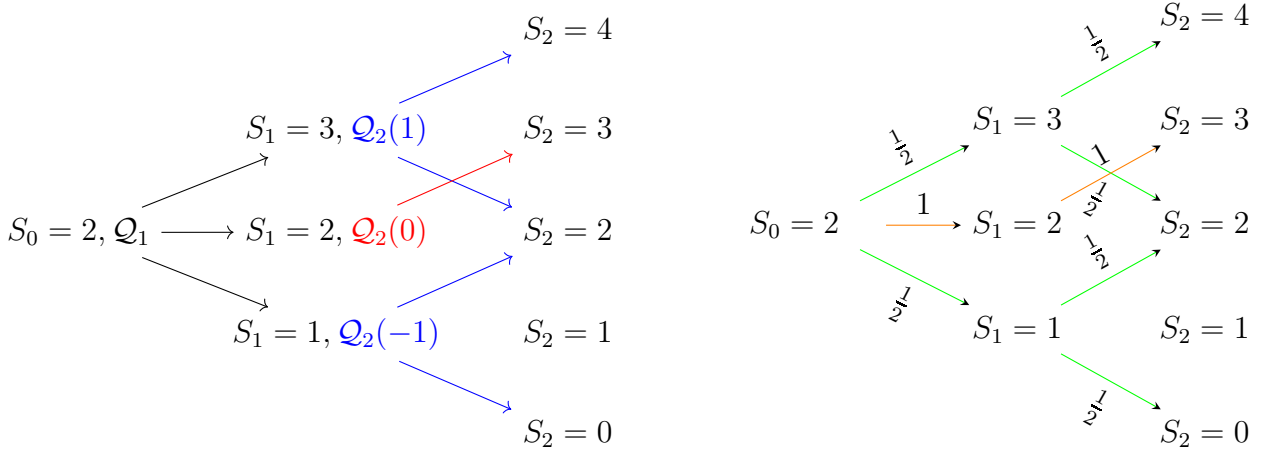
\begin{figure}
   \label{F1}
   \begin{center}
     \begin{tikzpicture}[scale=0.33]
    \matrix (tree) [
      matrix of nodes,
      minimum size=.4cm,
      column sep=.8cm,
      row sep=.65cm,
    ]
   {		
						&   					& $S_{2}=4$ \\
  						& $S_{1}=3,   \textcolor[rgb]{0.00,0.00,1.00}{\mathcal{Q}_{2}(1)}$ 			&  $S_{2}=3$ \\
      		$S_0=2,\mathcal{Q}_{1}$ 		& $S_{1}={2}, \textcolor[rgb]{1.00,0.00,0.00}{\mathcal{Q}_{2}(0)}$			& $S_{2}=2$ \\
  						& $S_{1}=1, \textcolor[rgb]{0.00,0.00,1.00}{\mathcal{Q}_{2}(-1)}$ 	& $S_{2}=1$ \\
     						&   					& $S_{2}=0$ \\
    };
    \draw[->,draw=black] (tree-3-1) -- (tree-2-2) ;
    \draw[->,draw=black]  (tree-3-1) -- (tree-3-2) ;
    \draw[->,draw=black]  (tree-3-1) -- (tree-4-2)  ;
    
   \draw[->,draw=blue]  (tree-2-2) -- (tree-1-3)  ;
  \draw[->,draw=blue] (tree-2-2) -- (tree-3-3)  ;
    
  \draw[->,draw=red]  (tree-3-2) -- (tree-2-3)  ;

  \draw[->,draw=blue]  (tree-4-2) -- (tree-3-3)  ;
  \draw[->,draw=blue]  (tree-4-2) -- (tree-5-3) ;

      \end{tikzpicture}
      \hspace{1cm}
           \begin{tikzpicture}[>=stealth,sloped]
    \matrix (tree) [%
      matrix of nodes,
      minimum size=.5cm,
      column sep=1cm,
      row sep=.75cm,
    ]
   {  
      &   					& $S_{2}=4$ \\
   				& $S_{1}=3$ 	&  $S_{2}=3$ \\
    $S_0=2\;\;\;\;$ 	& $S_{1}={2}$			& $S_{2}=2$ \\
      				& $S_{1}=1$ 	& $S_{2}=1$ \\
       				&   					& $S_{2}=0$ \\
    };
    \draw[->,draw=green] (tree-3-1) -- (tree-2-2) node [midway,above] {$\frac{1}{2}$};
    \draw[->,draw=orange]  (tree-3-1) -- (tree-3-2) node [midway,above] {$1$};
    \draw[->,draw=green]  (tree-3-1) -- (tree-4-2) node [midway,below] {$\frac{1}{2}$};
    
   \draw[->,draw=green]  (tree-2-2) -- (tree-1-3) node [midway,above] {$\frac{1}{2}$};

  \draw[->,draw=green] (tree-2-2) -- (tree-3-3) node [midway,below] {$\frac{1}{2}$};
 
  \draw[->, draw=orange] (tree-3-2) -- (tree-2-3) node [midway,above] {$1$};
    
  \draw[->,draw=green]  (tree-4-2) -- (tree-3-3) node [midway,above] {$ \frac{1}{2}$};

  \draw[->,draw=green]  (tree-4-2) -- (tree-5-3) node [midway,below] {$\frac{1}{2}$};

      \end{tikzpicture}
      \end{center}
  \caption{Left-hand side: The model. Right-hand side: In green $P^*$ and in orange $\overline{Q}$.}
  \end{figure}
It is clear that Assumptions  \ref{SassARB} and \ref{QanalyticARB} hold true. \\
Let $p_2(\cdot) \in \mathcal{Q}_{2}(\cdot)$ and set $P^*:=P_{na}\otimes p_{2} \in \mathcal{Q}^{2}$ 
  (see Figure 2).
It is immediate  that the $NA(P^{*})$ and thus the $wNA(\mathcal{Q}^2)$ conditions hold true. Furthermore $D_{P^*}^{2}(\pm 1)=D^{2}(\pm 1)=\{-1,1\}$ and $D_{P^*}^{2}(0)=D^{2}(0)=\{1\}$. Thus  for all $\o_{1}$, $\mbox{Aff}\left(D_{P^*}^{2}\right)(\o_{1})=\mbox{Aff}\left(D^{2}\right)(\o_{1})=\mathbb{R}$. As $0 \in \mbox{Ri} \left(\mbox{Conv} \left(D_{P^*}^{2}\right)\right)(\pm 1)$ and  $P_{1}^*\left(\{\pm 1\}\right)=1$,  $0 \in \mbox{Ri} \left(\mbox{Conv} \left(D_{P^*}^{2}\right)\right)(\cdot)$ $P_{1}^{*}$-a.s. Now let $\bar{Q}:=P_{0} \otimes p_{2} \in \mathcal{Q}^{2}$ (see Figure 2). Then  $\bar{Q}^{1}(\{0\})=1$ and  $0 \notin \mbox{Ri} \left(\mbox{Conv} \left(D_{P^*}^{2}\right)\right)(0)$ implies that $0 \in \mbox{Ri} \left(\mbox{Conv} \left(D_{P^*}^{2}\right)\right)(\cdot)$ $\bar{Q}^{1}$-p.s. and thus $0 \in \mbox{Ri} \left({\mbox{Conv}}(D_{P^*}^{2})\right)(\cdot )$ $\mathcal{Q}^{1}$-q.s. are not verified. \\ 
Let us check that the $NA(\mathcal{Q}^{2})$ condition does not hold true. Choose $\phi  \in \Phi$ such that  $\phi_{1}=0$ and $\phi_{2}(\o_{1})=1_{0}(\o_{1})$ and use again $\bar Q=P_{0} \otimes p_{2} \in \mathcal{Q}^2$. Then  $V_{2}^{0,\phi} \geq 0$ $\mathcal{Q}^2$-q.s. and $\bar Q \left(\{V_{2}^{0,\phi}>0\}\right)=\delta_{1}(\{\o_{2}>0\})=1$.\\ 
Now replace $\mathcal{Q}_{2}$ by $\widetilde{\mathcal{Q}}_{2}(\cdot):=\mbox{Conv}\left(P_{na},P_{1}\right)$ while keeping $\widetilde{\mathcal{Q}}_{1}=\mathcal{Q}_{1}$ as before and set  $\widetilde{P}^{*}:=P_{na}\otimes \widetilde{p}_{2}$, where  $\widetilde{p}_{2}(\cdot,\o_{1}):=P_{na}(\cdot)$ for all $\o_{1}$. Then $D_{\widetilde{P}^{*}}^{2}(\o_1)=\{-1,1\}$,  $\mbox{Aff}(D_{\widetilde{P}^*}^{2})(\o_{1})=\mbox{Aff}\left(D^{2}\right)(\o_{1})=\mathbb{R}$ and $0 \in \mbox{Ri} (\mbox{Conv} (D_{\widetilde{P}^{*}}^{2}))(\o_{1})$ for all $\o_{1}$.   One can directly check that the $NA(\widetilde{\mathcal{Q}}^{2})$ condition holds true. \\
Finally,   one may  build $\widetilde{\mathcal{P}}^{2}$ using  \eqref{PstarARB}  and $\widetilde{P}^{*}$.  It is clear that $\widetilde{\mathcal{P}}^{2}$  is strictly included in $\widetilde{\mathcal{Q}}^{2}$ since it   does not contain $\{P_0 \otimes q_{2}, q_{2}(\cdot,\o_{2}) \in \widetilde{\mathcal{Q}}_{2}(\o_2) \}$.
\end{example}
\noindent The following result provides an answer to the measurability issue raised in Remark \ref{alphabeta} and also provides a commun prior for all strategies.
\begin{proposition}
\label{finally}
Assume that Assumptions  \ref{SassARB} and \ref{QanalyticARB} as well as the $NA(\mathcal{Q}^{T})$ condition hold true.
Then for all $0\leq t\leq T-1$ 
there exists some $\mathcal{B}_{c}(\O^{t})$-measurable random variables $\beta_{t}(\cdot),\kappa_{t}(\cdot) \in (0,1)$ such that for all $\omega^{t} \in \Omega^{t}_{NA}$ and   $h \in  \mbox{Aff}\left({D}^{t+1}\right)(\omega^{t})$ , $h \neq 0$  
\begin{align}
\label{valakiARBtop}
p^*_{t+1}\left({h}\Delta S_{t+1}(\omega^{t},\cdot) < -\beta_{t}(\omega^{t}){|h|}, \omega^t\right) \geq \kappa_{t}(\omega^{t}),
\end{align}
where   $p_{t+1}^{*}(\cdot,\o^{t})$ is defined in  Theorem \ref{PPstar} with the fix disintegration $P^*:=P_{1}^* \otimes p_{2}^{*} \otimes \cdots \otimes p_{T}^{*}$. 
\end{proposition}
\begin{remark}
\label{kappa}
We have that $\beta_{t}(\o^t)=\kappa_{t}(\o^t)=1$ only if $D^{t+1}_{P^*}(\o^t)=\{0\}$. Indeed if  $\beta_{t}(\o^t)=\kappa_{t}(\o^t)=1$ and $D^{t+1}_{P^*}(\o^t)\neq\{0\}$, then for all $h \in  \mbox{Aff}\left({D}^{t+1}\right)(\omega^{t})$ with $|h|=1$ $p^*_{t+1}\left({h}\Delta S_{t+1}(\omega^{t},\cdot) < -1, \omega^t\right) =1.$ 
Fix such a $h$ and let  $F_h:=\{y \in \mathbb{R}^d, \; hy \leq -1\}.$ Then 
$p^*_{t+1}\left(\Delta S_{t+1}(\omega^{t},\cdot) \in F_{ \pm h},\o^t\right)=1$ and  $D^{t+1}_{P^*}(\o^t)=E^{t+1}(\o^t,p_{t+1}^{*}(\cdot,\o^{t}) )\subset F_{-h} \cap F_{+h}=\emptyset,$ see Remark \ref{DomainInc}.   Note that it is not easy to obtain this result for Theorem \ref{ARBequival} as the prior in \eqref{valakiARB} depends on $h$. 
\end{remark}
\begin{proof}
See Section \ref{setwoARBe}.
\end{proof} \\

Finally, if there exists a dominating  probability measure $\widehat{P} \in \mathcal{Q}^{T}$,  the following result holds true. 
\begin{proposition}
\label{CasD}
Assume that Assumptions   \ref{SassARB} and  \ref{QanalyticARB} hold true.  Assume furthermore that there exists some  dominating measure $\widehat{P} \in \mathcal{Q}^{T}$. Then the $NA(\widehat{P})$ and    the $NA(\mathcal{Q}^{T})$ conditions are equivalent. In this case, for all $0 \leq t \leq T-1,$
\begin{align}
\label{tropfull}
D_{\widehat{P}}^{t+1}(\cdot)=D^{t+1}(\cdot) \mbox{ and } 0 \in \mbox{Ri} \left({\mbox{Conv}}(D_{\widehat{P}}^{t+1})\right)(\cdot ) \;\; \mathcal{Q}^{t}\mbox{-q.s.}
\end{align}
 \end{proposition}
\begin{remark}
One can choose $P^*=\widehat{P}$ in Proposition \ref{finally} changing $\Omega^{t}_{NA}$ by the full-measure set where \eqref{tropfull} holds true. Moreover, $\mathcal{P}^{T}$ (see \eqref{PstarARB}) in Theorem \ref{TheoS} can be constructed from $\widehat{P}$.
\end{remark}
\begin{proof}
See Section \ref{setwoARBef}. 
\end{proof}

\section{Examples}
\label{ExAp}

This section proposes  concrete examples of multiple-priors setting  illustrating our results.   We  also use these examples to present  how to build sets of probability measures which are not dominated.   This relies on the following result. 

\begin{proposition}
\label{DOM}
Assume that Assumption \ref{QanalyticARB} holds true and that   there exists  some  $\widetilde{P} \in \mathcal{Q}^{T}$, some $0 \leq t \leq T-1$ and  some  $\O^{t}_{N} \in \mathcal{B}_{c}(\O^{t})$ such that  $\widetilde{P}^{t}(\O^{t}_{N})>0$  and such that the set $\mathcal{Q}_{t+1}(\o^{t})$ is not dominated for all $\o^{t} \in \O^{t}_{N}$. Then $\mathcal{Q}^{T}$ is not dominated.
\end{proposition}
\begin{proof}
See Section \ref{PRDOM}. 
\end{proof}

\subsection{Robust Binomial model}
\label{BI}

Suppose that $T \geq 1$, $d=1$ and   $\Omega_{t}=\mathbb{R}$ (or $(0,\infty)$) for all $1 \leq t \leq T$. 
The risky asset $\left(S_{t}\right)_{0 \leq t \leq T}$ is such that   $S_0=1$
 and $S_{t+1}= S_{t} Y_{t+1}$ where $Y_{t+1}$ is a  real-valued and  $\mathcal{B}(\O_{t+1})$-measurable random variable  such that $Y_{t+1}(\O_{t+1})=(0,\infty)$ for all $0 \leq t \leq T-1$ (if $\Omega_{t}=(0,\infty)$ you can think of $Y_{t}=\o_{t}$). The positivity of $Y_{t}$  implies that $S_{t}(\o^{t})>0$ for all $\o^{t} \in \O^{t}$.  It is clear that Assumption \ref{SassARB} is verified. Then,  for $0\leq t \leq T-1$ let
\begin{align*}\mathcal{B}_{t+1}(\o^{t}):=\{\pi\delta_{u} +(1-\p) \delta_{d},\; \p_{t}(\o^{t}) \leq \p \leq \Pi_{t}(\o^{t}),&\;u_{t}(\o^{t}) \leq u \leq U_{t}(\o^{t}),
\;d_{t}(\o^{t}) \leq d \leq D_{t}(\o^{t}) \},\end{align*}
where $\p_{t},\Pi_{t}, u_{t}, U_{t}, d_{t},D_{t}$ are real-valued $\mathcal{B}(\O^{t})$-measurable random variables such that $0 \leq \p_{t}(\o^{t}) \leq \Pi_{t}(\o^{t}) \leq 1$, $u_{t}(\o^{t}) \leq U_{t}(\o^{t})$ and $ d_{t}(\o^{t}) \leq D_{t}(\o^{t}) $ for all $\o^{t} \in \O^{t}$.\footnote{This could be generalised  by setting $\mathcal{B}_{t+1}(\o^{t}):=\{\p\delta_{u} +(1-\p) \delta_{d},\; \p  \in \mathcal{S}_{t}(\o^{t}),\; u \in   \mathcal{U}_{t}(\o^{t}), \; d \in  \mathcal{D}_{t}(\o^{t}) \}$, where $\mathcal{S}_{t}$, $\mathcal{U}_{t}$, $\mathcal{D}_{t}$  are Borel-measurable random sets $\Omega^t \twoheadrightarrow \mathbb{R}.$}
\begin{assumption}
\label{Abin}
We have that  $\p_{t}(\o^{t})<1$, $\Pi_{t}(\o^{t})>0$ and $0<d_{t}(
\o^{t})<1<U_{t}(\o^{t})$   for all $0\leq t \leq T-1$  and  $\o^{t} \in \O^{t}$. 
\end{assumption}
For all $0\leq t \leq T-1$ and $\o^{t} \in \O^{t},$ let 
\begin{align}
 \label{convexify}
 \widetilde{\mathcal{Q}}_{t+1}(\o^{t})&:=\left\{q \in \mathfrak{P}(\O_{t+1}),\; q\left(
Y_{t+1} \in \cdot \right) \in \mathcal{B}_{t+1}(\o^{t})\right\} \mbox{ and } 
 \mathcal{Q}_{t+1}(\o^{t}):=\mbox{Conv}\left(\widetilde{\mathcal{Q}}_{t+1}(\o^{t})\right),
\end{align}
where  $q\left(Y_{t+1} \in \cdot \right)$ is the law of $Y_{t+1}$ under $q$.  In words, at each step, the risky asset can go up or down and  there is uncertainty not only on the probability of the jumps but also on their sizes.
\begin{remark} 
The usual binomial model (see \citep{CRR79})  corresponds  to $\p_{t}=\Pi_{t}=\pi$, $u_{t}=U_{t}=u$ and $d_{t}=D_{T}=d$ where $0<\pi<1$, $d<1<u$. 
\end{remark}
\begin{lemma}
\label{cesconv}
Under Assumption \ref{Abin}, Assumption \ref{QanalyticARB} holds true. 
\end{lemma}
\begin{proof}
First, $\mathcal{Q}_{t+1}$ is convex valued by definition. Since $Y_{t+1}(\O_{t+1})=(0,\infty)$,  $\widetilde{\mathcal{Q}}_{t+1}(\o^{t}) \neq \emptyset$, hence ${\mathcal{Q}}_{t+1}(\o^{t}) \neq \emptyset$  for all $\o^{t} \in \O^{t}$.   
 We show  successively that  $\mbox{graph}\left(\mathcal{B}_{t+1}\right)$, $\mbox{graph}\left(\widetilde{\mathcal{Q}}_{t+1}\right)$ and $\mbox{graph}\left({\mathcal{Q}}_{t+1}\right)$ are analytic sets.  For $\o^{t} \in \O^{t},$ let 
  \begin{align*}
  E(\o^{t})&:=\{(u,d,\p) \in \mathbb{R}^{3}, \; \p_{t}(\o^{t}) \leq \p \leq \Pi_{t}(\o^{t}),\;u_{t}(\o^{t}) \leq u \leq U_{t}(\o^{t}), \;d_{t}(\o^{t}) \leq d \leq D_{t}(\o^{t})\}, \\
 F(\o^{t},& u,d,\p):=(\o^{t},\p \delta_{u} +(1-\p) \delta_{d}) \quad \mbox{for $(\o^{t}, u,d,\p) \in \O^{t} \times \mathbb{R}^{3}$}.
\end{align*} 
Then $F$ is Borel-measurable (see \citep[Corollary 7.21.1 p130]{BS}),    $\mbox{graph}(E) \in \Bc_c(\O^t) \otimes \Bc(\mathbb{R}^3)$ as  $\p_{t}, \,\Pi_{t}, \,u_{t}, \,U_{t}, \,d_{t}$ and $D_{t}$ are   Borel-measurable. We conclude that 
$\mbox{graph}\left(\mathcal{B}_{t+1}\right)$$= F\left(\mbox{graph}(E)\right)$ is analytic. 
Let $\Phi :\mathfrak{P}(\O_{t+1}) \to \mathfrak{P}(\mathbb{R})$  be defined by $\Phi(q):= q\left(
Y_{t+1}  \in \cdot\right)$. 
Using \citep[Propositions 7.29 p144 and 7.26 p134]{BS}, $\Phi$ is a Borel-measurable  stochastic kernel on $\mathbb{R}$ given  $\mathfrak{P}(\O_{t+1})$. So $\hat{\Phi}(\o^{t},q):=(\o^{t}, \Phi(q))$ is also Borel-measurable and    $\mbox{graph}(\widetilde{\mathcal{Q}}_{t+1})= \hat{\Phi}^{-1} \left(\mbox{graph}\left(\mathcal{B}_{t+1}\right) \right)$ is analytic. Then one can show as in  \citep[Proofs for Section 2.3]{Bart216} that $\mbox{graph}\left({\mathcal{Q}}_{t+1}\right)$ is analytic since ${\mathcal{Q}}_{t+1}$ is the convex hull of $\widetilde{\mathcal{Q}}_{t+1}$. 

\end{proof}
\begin{lemma}
\label{lembin}
Under Assumption \ref{Abin}, the $NA(\mathcal{Q}^{T})$ condition holds true and the $sNA(\mathcal{Q}^{T})$ condition might fails.
\end{lemma}
\begin{proof} 
 It is clear that for all $0\leq t \leq T-1$, all $\o^{t} \in \O^{t}$, 
$$\mbox{Conv} \left(D^{t+1}\right)(\o^{t})=[S_{t}(\o^{t})(d_{t}(\o^{t})-1),S_{t}(\o^{t})(U_{t}(\o^{t})-1)].$$ 
So the $NA(\mathcal{Q}^{T})$ condition holds true  as $0 \in \mbox{Ri}\left(\mbox{Conv}\left(D^{t+1}\right)\right)(\o^{t})$ for all $\o^{t} \in \O^{t}$ (see Theorem \ref{ARBequival}).
Under  Assumption \ref{Abin}, one may have that $u_{t}(\o^{t})<1$ for all $\o^{t} \in \O^{t}$, $0 \leq t \leq T-1$ and find some $a_{t}(\o^{t}) \in [u_{t}(\o^{t}), 1)$. For all $0 \leq t \leq T-1$ and  $\o^{t} \in \O^{t}$, let 
$$q_{t+1}( Y_{t+1} \in \cdot,\o^{t}):= r_{t}(\o^{t}) \delta_{a_{t}(\o^{t})}(\cdot)+  \left(1-r_{t}(\o^{t}) \right)\delta_{d_{t}(\o^{t})}(\cdot),$$ 
where $r_{t}(\o^{t}) \in [\p_{t}(\o^{t}),\Pi_{t}(\o^{t})]$.  Set  $Q:=Q_{1} \otimes q_{2}  \otimes \cdots \otimes q_{T} \in \mathcal{Q}^{T}$.   As  
$$\mbox{Conv}\left(D_{Q}^{t+1}\right)(\o^{t})=\left[S_{t}(\o^{t})(d_{t}(\o^{t})-1), S_{t}(\o^{t})(a_{t}(\o^{t})-1)\right],$$  $0 \notin \mbox{Conv}\left(D_{Q}^{t+1}\right)(\o^{t})$ for all $\o^{t} \in \O^{t}$ and Proposition \ref{singleP} implies that $NA({Q})$  and thus $sNA(\mathcal{Q}^{T})$  fail.
\end{proof}

We now provide some explicit expressions for  $\varepsilon_{t}$, $\beta_{t}$ and $\kappa_{t}$ of \eqref{vamakiARB} and \eqref{valakiARB} and exhibit a  candidate for the measure $P^*$  of Theorem \ref{PPstar}. 
\begin{lemma}
Assume that  Assumption \ref{Abin} holds true. For all $0\leq t \leq T-1$, all $\o^{t} \in \O^{t}$ let 
\begin{align*}
\bar \p_{t}(\o^{t}) &  :=\frac{\p_{t}(\o^{t})+\Pi_{t}(\o^{t})}{2} \in (0,1)\\
\frac{\varepsilon_{t}(\o^{t})}{2}& =\beta_{t}(\o^{t}):= \frac{S_{t}(\o{^t})}{N}\min\left(\frac{U_{t}(\o^{t})-1}{2},\frac{1-d_{t}(\o^{t})}{2}\right)>0,\\
\kappa_{t}(\o^{t})& : =\frac1M\min\left(\bar \p_{t}(\o^{t}) ,1-\bar \p_{t}(\o^{t}) \right)>0,\\
 a_t^+ (\o^{t})&: =U_{t}(\o^{t})>1, \;  \; \;\; b_t^{+}(\o^{t}) : =\min\left(D_{t}(\o^{t}),  \frac{d_{t}(\o^{t})+1}{2}\right)<1,  \\
a_t^{-}(\o^{t}) & : = \max \left(u_{t}(\o^{t}),  \frac{U_{t}(\o^{t})+1}{2}\right) >1, \;  \; \;\;b_t^{-}(\o^{t}) : =d_{t}(\o^{t})<1,\\
r_{t+1}^{\pm}(\cdot,\o^t) & : = \bar \p_{t}(\o^{t}) \delta_{a_t^{\pm}(\o^t)}(\cdot) + (1-\bar \p_{t}(\o^{t})) \delta _{b_t^{\pm}(\o^t)}(\cdot) \in \mathcal{B}_{t+1}(\o^{t}),\\
r_{t+1}^*(\cdot ,\o^{t})  &  := \frac{1}{2} \left(r_{t+1}^{+}(\cdot,\o^t)+ r_{t+1}^{-}(\cdot,\o^t)\right) \in 
\mathcal{B}_{t+1}(\o^{t}), \;\; p_{t+1}^{*}(Y_{t+1}\in \cdot,\o^t)  := r_{t+1}^*(\cdot ,\o^{t}) \in \mathcal{Q}_{t+1}(\o^{t})
\end{align*}
where  $N>1$ and $M> 1$ are fixed and allows to get sharper bound for $\e_{t}(\o^{t}),\;\beta_{t}(\o^{t})$ and $\kappa_{t}(\o^{t})$. Then  
\begin{align}
\label{quelstar} 
p_{t+1}^{*}\left(\pm \Delta S_{t+1}(\o^{t},\cdot) <   -\beta_{t}(\o^{t}),\o^t \right) 
\geq \kappa_{t}(\o^{t}),
\end{align}
and \eqref{valakiARB} is satisfied;  \eqref{vamakiARB} also holds true. \\
Moreover,  for $P^{*}:= P_{0}^*\otimes p^{*}_{1} \cdots \otimes p^{*}_{T} \in \mathcal{Q}^{T},$ $0 \in \mbox{Ri} \left({\mbox{Conv}}(D_{P^*}^{t+1})\right)(\o^{t})$ and  $\mbox{Aff}\left(D_{P^*}^{t+1}\right)(\o^{t})=\mbox{Aff} \left(D^{t+1}\right)(\o^{t})=\mathbb{R}$  for all $\o^t \in \Omega^t$.\\
Finally, assume that for some $0 \leq t \leq T-1$ and some $\o^{t} \in \O^{t}$,  $u_{t}(\o^{t})<U_{t}(\o^{t})$ or  $ d_{t}(\o^{t}) < D_{t}(\o^{t}).$ Then the set $\mathcal{Q}_{t+1}(\o^{t})$ is  not dominated and one can construct sets $\mathcal{Q}^{T}$ which are not dominated.
\end{lemma}
\begin{remark}
Note that $P^*$ is not unique. 
The (Borel) measurability of $\varepsilon_{t}, \beta_{t}$ and $\kappa_{t}$  are clear. Similarly they will inherit  any integrability conditions  imposed on $S_{t}$, $\p_{t}$, $\Pi_{t}$,  $d_{t}$, $D_{t},$ $u_{t}$ and $U_{t}$. For instance  if  they 
belong to $\mathcal{W}_{t}$ 
for all $1 \leq t \leq T$  so do $\varepsilon_{t}, \beta_{t}$ and $\kappa_{t}$. 
\end{remark}
\begin{proof}
Fix some $0 \leq t \leq T-1$, $\o^{t} \in \O^{t}$. 
Let $q_{t+1}^{\pm}(Y_{t+1}\in \cdot,\o^t) := r_{t+1}^{\pm}(\cdot ,\o^{t}) \in \mathcal{Q}_{t+1}(\o^{t})$. Then 
\begin{small}
\begin{align}
  \label{HH1}
q_{t+1}^{+}\left(\Delta S_{t+1}(\o^{t},\cdot) <   -\beta_{t}(\o^{t}), \o^t \right)
&\geq q_{t+1}^{+}\left(Y_{t+1}(\cdot)  <  \frac{d_{t}(\o^{t})+1}{2}, \o^t\right)
  \geq 1-\bar \pi_t(\o^t) \geq \kappa_{t}(\o^{t})\\
    \label{HH2}
q_{t+1}^{-}\left(\Delta S_{t+1}(\o^{t},\cdot) >   \beta_{t}(\o^{t}), \o^t \right) & \geq  q_{t+1}^{-}\left(Y_{t+1}(\cdot)  > \frac{U_{t}(\o^{t})+1}{2},\o^t\right)  \geq \bar \pi_t(\o^t) \geq \kappa_{t}(\o^{t})
\end{align}
\end{small}
and \eqref{quelstar} follows  while \eqref{vamakiARB} follows from Theorem \ref{ARBequival}.\\
As $p_{t+1}^*\in \mathcal{SK}_{t+1}$,  $P^{*} \in \mathcal{Q}^{T}$. 
From \eqref{quelstar},  the quantitative no-arbitrage  \eqref{valakiARB} holds true for all $\o^{t} \in \O^{t}$ with $p_{h}=p_{t+1}^{*}(\cdot, \o^{t})$ for all possible strategy $h$.  Therefore the $NA(P^*)$ condition holds true (see Remark \ref{singleP2}). 
Theorem \ref{ARBequival} 
implies also  that $0 \in \mbox{Ri} \left({\mbox{Conv}}(D_{P^*}^{t+1})\right)(\o^{t})$. Moreover $\mbox{Aff}\left(D_{P^*}^{t+1}\right)(\o^{t})=\mbox{Aff} \left(D^{t+1}\right)(\o^{t})=\mathbb{R}$  for all $\o^{t}$. \\
For the last item, assume that  for some $0 \leq t \leq T-1$ and some $\o^{t} \in \O^{t}$,  $u_{t}(\o^{t})<U_{t}(\o^{t})$ and that the set $\mathcal{Q}_{t+1}(\o^{t})$ is dominated by some measure $\widehat p$. 
For  $x \in(0,\infty)$ let 
$A_{x}:=\{Y_{t+1}^{-1}(\{x\})\}\neq \emptyset$ as $Y_{t+1}(\O^{t})=(0,\infty)$. 
Fix  $x(\o^{t}) \in(\min(1,u_{t}(\o^{t})),U_{t}(\o^{t}))$ and choose $a(\o^{t}) \in A_{x(\o^{t})}$ and 
$b(\o^{t}) \in A_{d_t(\o^{t})}$. Let    $r_x(.,\o^t):= \Pi_{t}(\o^{t}) \delta_{a(\o^{t})} + (1-\Pi_{t}(\o^{t}) ) \delta_{b(\o^{t})} \in \mathcal{B}_{t+1}(\o^{t})$ and $p_x(Y_{t+1} \in \cdot,\o^t):= r_x(.,\o^t) \in  \mathcal{Q}_{t+1}(\o^{t}).$
As  $r_x(\{a(\o^{t})\},\o^t)=\Pi_{t}(\o^{t})>0$, $\widehat p (\{a(\o^{t})\})>0$, which leads to an uncountable number of atoms for $\widehat p$.\\ 
 Then,  Proposition \ref{DOM} allows to build examples of  sets $\mathcal{Q}^{T}$ which are not dominated.
\end{proof}

\subsection{ Discretized $d$-dimensional diffusion}
We provide now an example   for  the discretized dynamics of a multi-dimensional diffusion process in the spirit of \citep[Example 8.2]{cr11}.\\
Fix a period $T \geq 1$ and $  n \geq d$. Denote by 
 $M_{n}$  the set of real-valued matrix with $n$ rows and $n$ columns. Choose some constant $Y_0 \in \mathbb{R}^{n}$  and  let   $Y_{t+1}$  be  defined by the following difference equation for all $0 \leq t \leq T-1$, $(\o^{t},\o_{t+1}) \in \O^{t} \times \O_{t+1}$
\begin{align}
\label{Yt}
Y_{t+1}(\o^{t},\o_{t+1})-Y_{t}(\o^{t})=\mu_{t+1}\left(Y_{t}(\o^{t}),\o^{t},\o_{t+1}\right) + \nu_{t+1}\left(Y_{t}(\o^{t}),\o^{t}\right) Z_{t+1}(\o^{t},\o_{t+1})
\end{align}
where $\mu_{t+1}: \mathbb{R}^{n} \times \O^{t} \times \O_{t+1} \to \mathbb{R}^{n}$, $\nu_{t+1}:  \mathbb{R}^{n} \times \O^{t} \to M_{n}$, $Z_{t+1}: \O^{t}\times \O_{t+1} \to \mathbb{R}^{n}$ are assumed to be Borel-measurable.  \\
Two cases will be studied: $S^{i}_{t}=Y^{i}_{t}$ and  $S^{i}_{t}=e^{Y^{i}_{t}}$ for all $1 \leq i \leq d$. In a uni-prior setting if the law of $Z_{t+1}$ is assumed to be normal, this corresponds to the popular normal and lognormal dynamic for the underlying assets.  Note that  in both cases if $d<n$ we may think that $Y^{i}_{t}$ for $i>d$ represents some non-traded assets or the evolution of some economic factors that will influence the market.\\
Assume that some $P^{0}\in \mathfrak{P}(\O^{T}) $ is given with  fixed disintegration $P^{0}:=P^{0}_{1}\otimes p_{2}^{0}\otimes \cdots \otimes p_{T}^{0}$, where  $p^{0}_{t+1} \in \mathcal{SK}_{t+1}$ for all $0 \leq t \leq T-1$: $P^0$ could be  an initial guess or estimate for the prior.    For all $0 \leq t \leq T-1$,  let $r_{t}$ and $q_{t}$ be  functions from $\O^t$ to $(0,\infty)$: $r_{t}$ will be the bound on the drift while $q_{t}$ guarantees that the diffusion is non-degenerated (in  dimension one it is  a lower bound on the volatility). 
We make the following assumptions on the dynamic of $Y$.
\begin{assumption}
\label{assdiff}
For all $0 \leq t \leq T-1$, $r_{t}$ is $\mathcal{B}(\O^{t})$-measurable. For all 
 $\o^{t} \in \O^{t}, \; x \in \mathbb{R}^{n}$, 
\begin{itemize}
\item  $\nu_{t+1}(x,\o^{t}) \in M_{n}^{{q_{t}(\o^{t})}}$  where 
$M^{\delta}_{n}:=\left\{M \in M_{n},\; \forall \, h \in \mathbb{R}^{n},\;h^{t}MM^{t}h  \geq {\delta} h^{t}h \right\}$ for $\delta>0$.
\item  $Z_{t+1}(\o^{t},\cdot)$ and $\mu_{t+1}(\o^{t},\cdot)$ are independent under $p^{0}_{t+1}(\cdot,\o^{t})$.
\item $p^{0}_{t+1}\left(\mu_{t+1}(Y_{t}(\o^{t}),\o^{t},\cdot) \in [-r_{t}(\o^{t}),r_{t}(\o^{t})]^{n},\o^{t}\right)=1$
\item  $D_{Z_{t+1}}^{t+1}(\o^{t})=\mathbb{R}^{n}$, where $D_{Z_{t+1}}^{t+1}(\o^{t})$ is the support of $Z_{t+1}(\o^{t},\cdot)$ under $p_{t+1}^{0}(\cdot,\o^{t}),$ see \eqref{DefPDARB1}. 
\end{itemize}
\end{assumption}
The model  uncertainty  on the laws  of   $\mu_{t+1}$ and   $Z_{t+1}$ is given by the folowing sets.
\begin{align}
\label{defQ1}
\mathcal{Q}^{1}_{t+1}(\o^{t})&:=\left\{p \in \mathfrak{P}(\O_{t+1}),\;  p\left(\mu_{t+1}(Y_{t}(\o^{t}),\o^{t},\cdot)\in [-r_{t}(\o^{t}),r_{t}(\o^{t})]^{n}\right)=1\right\},\\
\label{defQ2}
\mathcal{Q}^{2}_{t+1}(\o^{t})&:=\left\{p \in \mathfrak{P}(\O_{t+1}),\; F_{t}(p,\o^{t}) = 0\right\},\\
\nonumber
\mathcal{Q}_{t+1}(\o^{t})&:= \mathcal{Q}^{1}_{t+1}(\o^{t}) \bigcap \mathcal{Q}^{2}_{t+1}(\o^{t}),
\end{align}
where for some   $k \geq 1$,  $F_{t}: \mathfrak{P}(\O_{t+1}) \times \O^{t} \to \mathbb{R}^{k}$ is a Borel-measurable function such that $F_{t}(p_{t+1}^{0}(\cdot,\o^{t}),\o^{t})=0$ 
   for all $0 \leq t \leq T-1$, $\o^{t} \in \O^{t}$. 
By assumption  $p_{t+1}^{0}(\cdot,\o^{t}) \in \mathcal{Q}_{t+1}(\o^{t})$ for all $\o^{t} \in \O^{t}$ and thus $P^0 \in \mathcal{Q}^{T}$. Note  that for a given $p \in \mathcal{Q}_{t+1}(\o^{t})$ the law of $Z_{t+1}(\o^{t},\cdot)$ and $\mu_{t+1}(\o^{t},\cdot)$ under $p$ are not necessarily independent.

The financial interpretation is the following. The set $\mathcal{Q}^{1}_{t+1}(\o^{t})$ allows  the drift of the diffusion to be not only stochastic but with an unknown  distribution. It  is only assumed to be bounded. 
If $F_{t}(p,\o^{t})= 1_{\mbox{dist}_{t}\left(p\;,\;p^{0}_{t+1}(\cdot,\o^{t})\right) \leq b_{t}(\o^{t})}-1$ with $b_{t}(\o^t)>0$ and 
$\mbox{dist}_{t}$ some kind of distance function between probability measures, the set $\mathcal{Q}^{2}_{t+1}(\o^{t})$ contains models which are close enough from $p_{t+1}^{0}(\cdot,\o^{t})$.  This could happen if  the  physical measure is not known but estimated from data at each step. A popular choice for the $\mbox{dist}_{t}$ function is the Wasserstein distance. But one may also choose for the  coordinate $i$    of $F(p,\o^{t})$ (with $1\leq i \leq k$) the difference between the moments of order $i$ of  $Z_{t+1}(\o^{t},\cdot)$ under $p$ and under $p_{t+1}^{0}(\cdot,\o^{t})$ 
and incorporate all the models $p$ such that the moments of $Z_{t+1}(\o^{t},\cdot)$ under $p$ are  equals to the ones of $Z_{t+1}(\o^{t},\cdot)$ under $p^0_{t+1}(\cdot,\o^{t})$ up to order $k$.\\

\begin{lemma}  
Under Assumption \ref{assdiff}, Assumptions   \ref{SassARB} and \ref{QanalyticARB} are satisfied. 
\end{lemma}
\begin{proof}
Assumption \ref{SassARB} follows from the 
Borel measurability of $\mu_{t+1}$,  $\nu_{t+1}$, $Z_{t+1}$ and thus of $Y_{t+1}$.
As the function $(\o^{t},p) \to p\left(\mu_{t+1}(Y_{t}(\o^{t}),\o^{t},\cdot)\in [-r_{t}(\o^{t}),r_{t}(\o^{t})]^{n}\right)$ is Borel-measurable (see  \citep[Proposition 7.29 p144]{BS}), $\mbox{graph} \left(\mathcal{Q}^{1}_{t+1}\right)$ is analytic. The Borel-measurability of $F_{t}$ implies that $\mbox{graph} \left(\mathcal{Q}^{2}_{t+1}\right)$ is an analytic set and  so is $\mbox{graph} \left(\mathcal{Q}_{t+1}\right)$. It is clear that $\mathcal{Q}^{1}_{t+1}$ is convex valued. If  
$F_{t}(\cdot,\o^{t})$  is convex 
 for all $\o^{t} \in \O^{t}$,
 then $\mathcal{Q}^{2}_{t+1}$ is convex valued. Else one may consider the convex hull of $\mathcal{Q}^{2}_{t+1}$ whose analyticity can be established as in the proof of Lemma \ref{cesconv}.  Assumption \ref{QanalyticARB}  is proved. 
 \end{proof}\\
 Now we give explicit values for 
$\beta_{t}$ and $\kappa_{t}$ in \eqref{valakiARB}   with $p_{h}=p_{t+1}^{0}(\cdot,\o^{t})$ and prove $NA(\mathcal{Q}^{T})$. 
 \begin{lemma}
 Assume that Assumption \ref{assdiff} is satisfied and that $S^{i}_{t}=Y^{i}_{t}$ for all $1 \leq i \leq d$ and all $1 \leq t \leq T$. Then  $D^{t+1}(\o^{t})=\mathbb{R}^d$  for all $\o^t \in \Omega^t$ and $1 \leq t \leq T-1$ and $NA(\mathcal{Q}^{T})$ condition holds true. Let 
 \begin{align}
\label{bof}
\kappa_{t}(\o^{t}):=  \min_{k \in K} \left(p^{0}_{t+1}\left(G_{k}(\o^{t}),\o^{t}\right)\right)>0  \; \mbox{and} \;
\beta_{t}(\o^{t}):=\frac{\mbox{ln 2}}{\sqrt{n}}>0,\end{align}
where  $K$ is the (finite)  set of functions from   $\{1,\cdots, d\}$  to $\{-1,1\}$ and  for some $k \in K$  
\begin{align}
\label{gk}
G_{k}(\o^{t}):=\left\{k(i)  \Delta Y^{i}_{t+1}(\o^{t},\cdot) < -\mbox{ln 2}, \; 1 \leq i \leq d \right\}.
\end{align}
Then, for all $h \in \mathbb{R}^d$ with $|h|=1$
\begin{align}
\label{diffaoa}
p_{t+1}^{0}\left(h \Delta S_{t+1}(\o^{t},\cdot) < - \beta_{t}(\o^{t}),\o^{t}\right) \geq \kappa_{t}.
\end{align}
\end{lemma}
\begin{proof}
First, we  show that   for all $\o^{t} \in \O^{t},$ $D^{t+1}(\o^{t})=\mathbb{R}^d.$ To do that we prove that 
 \begin{align}
 \label{PY}
 D^{t+1}_{Y}(\o^{t}):= \bigcap\left\{ A \subset \mathbb{R}^{n},\; \mbox{closed}, \; p\left(\Delta Y_{t+1}(\o^{t},\cdot) \in A,\o^{t}\right)=1 \; \forall p \in \mathcal{Q}_{t+1}(\o^{t})\right\}= \mathbb{R}^{n}.
 \end{align}
Let $D_{Y,P_0}^{t+1}(\o^{t})$ be the support of $Y(\o^{t},\cdot)$ under $p_{t+1}^{0}(\cdot,\o^{t}),$ see \eqref{DefPDARB1}.
Using  \eqref{DvsE} $D^{t+1}_{Y,P_0}(\o^{t}) \subset D_Y^{t+1}(\o^{t})$ and it  is enough to prove that   $D^{t+1}_{Y,P_0}(\o^{t})=\mathbb{R}^n. $
Fix some $\o^{t} \in \O^{t}$. For ease of reading,  we  adopt the following notations. Let $\Delta Y(\cdot)= \Delta Y_{t+1}(\o^{t},\cdot)$, $R(\cdot)=\mu_{t+1}(Y_{t}(\o^{t}),\o^{t},\cdot)$, $X(\cdot)= \Delta Y(\cdot) -R(\cdot)$, $M=\nu_{t+1}(Y_{t}(\o^{t}),\o^{t})$, $Z(\cdot)=Z_{t+1}(\o^{t},\cdot)$ and $p^0(\cdot)=p^{0}_{t+1}(\cdot,\o^{t})$. As   $X(\cdot)=MZ(\cdot)$ (see \eqref{Yt}) and    $Z$ and $R$ are independent under $p^0$,  $X$ and $R$ are also independent under $p^0$.\\
Fix some $x_0 \in \mathbb{R}^{n}$, $\varepsilon>0$. By assumption $M$ is an invertible matrix:  There exists some $y_0 \in \mathbb{R}^{n}$, $\alpha>0$, such that $B(y_0,\alpha) \subset M^{-1}\left(B(x_0,\varepsilon)\right)$\footnote{$M^{-1}\left(B(x_0,\varepsilon)\right)$ is open in $\mathbb{R}^{n}$ and is not empty because $M$ is a bijective function on $\mathbb{R}^{n}$.}. The forth item of Assumption \ref{assdiff} together with Lemma \ref{supportcar} imply that\footnote{With the notation 
$p^0_{R}(A)=p^0(R \in A)$ for all $A \in \Bc(\mathbb{R}^{n})$.} 
\begin{align*}
p^0\left( X(\cdot) \in B(x_0,\varepsilon)\right) &=  p^0\left( Z(\cdot) \in M^{-1}\left(B(x_0,\varepsilon) \right)\right) \geq p^0\left( Z(\cdot) \in B(y_0,\alpha)\right)>0\\
p^0\left( \Delta Y(\cdot) \in B(x_0,\varepsilon)\right) & = p^0\left( X(\cdot)+R(\cdot)  \in B(x_0,\varepsilon)\right)
 =  \int_{\mathbb{R}} p^0\left( X(\cdot) \in B(x_0-u,\varepsilon)\right) p^0_{R}(du)>0,
 \end{align*}
as $X$ and $R$ are independent under $p^0$.
Lemma \ref{supportcar} implies that  the supports of $X $ and of $\Delta Y$ under $p^0$ are equal to  $\mathbb{R}^{n}$.

For all $0 \leq t \leq T-1$, $\o^{t} \in \O^{t}$, $D^{t+1}(\o^{t})=\mathbb{R}^{d}$ and   $0 \in \mbox{Ri} \left(\mbox{Aff}\left({D}^{t+1}\right)(\omega^{t})\right)$.  Theorem \ref{ARBequival} implies that the $NA(\mathcal{Q}^{T})$ condition is verified.\\
Fix now some $\o^{t} \in \O^{t}$ and  $h \in \mathbb{R}^{d}$ with $|h|=1$.
First, $D_{Y,P_0}^{t+1}(\o^{t})=\mathbb{R}^{n}$ implies that for all  $k \in K$, $\o^{t} \in \O^{t}$  \begin{align}
\label{PoG}
p_{t+1}^0\left(G_{k}(\o^{t}),\o^{t}\right)=p^{0}_{t+1}(\Delta Y_{t+1}(\o^{t},\cdot) \in \mathcal{O}_{h},\o^{t})>0,
\end{align}
where $\mathcal{O}_{h}:=\{z \in \mathbb{R}^{n},\; k(i)z_{i} < -\mbox{ln 2},\; \forall \, 1 \leq i \leq d\}$ is an open set of $\mathbb{R}^{n}$.
Set  $k^*(i):=\mbox{sign}(h_{i})$ for all $1 \leq i \leq d$, then $k^* \in K$. Let $\o_{t+1} \in G_{k^*}(\o^{t})$ as  \eqref{PoG} implies that $G_{k^*}(\o^{t})$ is not empty.  For all $1 \leq i \leq d$,  
$$h_{i} \Delta S^{i}_{t+1}(\o^{t},\o_{t+1})= |h_{i}| k^{*}(i) \Delta Y^{i}_{t+1}(\o^{t},\o_{t+1}) \leq -\mbox{ln 2} |h_{i}|  \leq 0.$$
As $|h|=1$ there exists $1 \leq i^* \leq d$ such that  $  \frac{1}{\sqrt{n}} \leq \frac{1}{\sqrt{d}}\leq |h_{i^*}| \leq 1$ and 
\begin{align*}
h \Delta S_{t+1}(\o^{t},\o_{t+1}) < - \frac{\mbox{ln 2}}{\sqrt{n}}+ \sum_{i \neq i^*}h_{i} \Delta Y^{i}_{t+1}(\o^{t},\o_{t+1}) \leq - \frac{\mbox{ln 2}}{\sqrt{n}}.
\end{align*}
Therefore 
$p_{t+1}^{0}\left(h \Delta S_{t+1}(\o^{t},\cdot) < - {\mbox{ln 2}}/{\sqrt{n}},\o^{t}\right) \geq \min_{k \in K} \left(p_{t+1}^{0}\left(G_{k}(\o^{t}),\o^{t}\right)\right)$.  
Recalling \eqref{PoG}, \eqref{diffaoa} is satisfied.  
\end{proof}\\
We now treat the log-normal case. 
\begin{lemma}
 Assume that Assumption \ref{assdiff} is satisfied and that $S^{i}_{t}=e^{Y^{i}_{t}}$  for all $1 \leq i \leq d$ and all $1 \leq t \leq T$. Then $D^{t+1}(\o^{t})=\mathbb{R}^d$  for all $\o^t \in \Omega^t$ and $1 \leq t \leq T-1$ and $NA(\mathcal{Q}^{T})$ condition holds true. Let 
 \begin{align}
\label{bof2}
\kappa_{t}(\o^{t}):=&\min_{k \in K} \left(p_{t+1}^0\left(G_{k}(\o^{t}),\o^{t}\right)\right) >0 \; \; \;
\beta_{t}(\o^{t}):=\frac{1}{2}\min\left({1}, \frac{\min_{1 \leq i\leq d}S^{i}_{t}(\o^{t})}{\sqrt{n}}\right)>0,
\end{align}
recall   \eqref{gk} for the definition of $G_{k}(\o^{t}).$ 
Then, for all $h \in \mathbb{R}^{d}$ with $|h|=1$
\begin{eqnarray}
\label{diffaoa2}
p_{t+1}^{0}\left(h \Delta S_{t+1}(\o^{t},\cdot) < - \beta_{t}(\o^{t}),\o^{t}\right) \geq \kappa_{t}.
\end{eqnarray}
\end{lemma}
\begin{proof}
Let $0 \leq t \leq T-1$ and fix  $\o^{t} \in \O^{t}$. Using  \eqref{DvsE} $D^{t+1}_{P_0}(\o^{t}) \subset D^{t+1}(\o^{t})$ and it  is enough to prove that   $D^{t+1}_{P_0}(\o^{t})=\mathbb{R}^d.$ This will follow from Lemma \ref{supportcar} if for any  open set $O$ of $\mathbb{R}^d,$  
$p^0\left( \Delta S_{t+1}(\cdot,\o^{t}) \in O,\o^{t}\right) >0.$
Fix an open set $O$ of $\mathbb{R}^d$ and let 
$F_{\o^{t}}: \mathbb{R}^{n} \to \mathbb{R}^{d}$ be  defined by $F_{\o^{t}}(x_1,\cdots,x_{n})=(e^{Y_{t}^{1}(\o^{t})}(e^{x_1}-1),\cdots,e^{Y_{t}^{d}(\o^{t})}(e^{x_d}-1))$. As $F_{\o^{t}}$ is continuous 
$F^{-1}_{\o^{t}} (O)$ is an open set of $\mathbb{R}^n$. Then
 \begin{align*}
 &p^0\left( e^{Y_{t+1}(\cdot,\o^{t})}- e^{Y_{t}(\o^{t})}  \in O,\o^{t}\right)\\
 &=p^0\left( \left(e^{Y^{1}_{t}(\o^{t})}  \left(e^{\Delta Y^{1}_{t+1}(\cdot,\o^{t})}-1\right),\cdots,  e^{Y^{d}_{t}(\o^{t})}  \left(e^{\Delta Y^{d}_{t+1}(\cdot,\o^{t})}-1\right)\right) \in O,\o^{t}\right)\\
& =p^0\left( \Delta Y_{t+1}(\cdot,\o^{t}) \in F^{-1}_{\o^{t}} (O),\o^{t} \right)>0,
 \end{align*}
using \eqref{PY} and Lemma \ref{supportcar} again. 
Thus for all $0 \leq t \leq T-1$, $\o^{t} \in \O^{t}$, $D^{t+1}(\o^{t})=\mathbb{R}^{d}$ and   $0 \in \mbox{Ri} \left(\mbox{Aff}\left({D}^{t+1}\right)(\omega^{t})\right)$.  Theorem \ref{ARBequival} implies that the $NA(\mathcal{Q}^{T})$ condition is verified.\\
Fix  a $\o^{t} \in \O^{t}$,  $h \in \mathbb{R}^{d}$ with $|h|=1$. Then
\begin{align}
\label{hisum}
h \Delta S_{t+1}(\o^{t},\o_{t+1})&
=\sum_{i=1}^{d} h_{i} S^{i}_{t}(\o^{t})\left( e^{ \Delta Y^{i}_{t+1}(\o^{t},\o_{t+1})}-1\right).
\end{align}
Let $k^* \in K$ as in the proof of the preceding lemma  
and let $\o_{t+1} \in G_{k^*}(\o^{t})$. First, for all $1 \leq i \leq d$, 
\begin{align}
\label{hipo}
h_{i} S^{i}_{t}(\o^{t}) \left( e^{\Delta Y^{i}_{t+1}(\o^{t},\o_{t+1})}-1\right) < 
\begin{cases}
  -\frac{|h_i|S^{i}_{t}(\o^{t})}{2} \; \mbox{if $k^*({i})=1$}\\
  - |h_{i}|{S^{i}_{t}(\o^{t})} \mbox{ if $k^*({i})=-1$}
\end{cases}
\leq 0.
\end{align}
As $|h|=1$ there is a component  $h_{i^*}$ such that $ \frac{1}{\sqrt{n}} \leq \frac{1}{\sqrt{d}} \leq |h_{i^*}| \leq 1$ and as $S^{i^*}_{t}(\o^{t}) > 0$,     \eqref{hisum} implies that
$$h \Delta S_{t+1}(\o^{t},\o_{t+1}) < -  \frac{S^{i^*}_{t}(\o^{t})}{2 \sqrt{n}} +\sum_{i \neq i^*} h_{i} S^{i}_{t}(\o^{t})\left( e^{ \Delta Y^{i}_{t+1}(\o^{t},\o_{t+1})}-1\right) \leq  -\frac{\min_{1 \leq i\leq d}S^{i}_{t}(\o^{t})}{2 \sqrt{n}}.$$
So, 
$$p_{t+1}^{0}\left(h \Delta S_{t+1}(\o^{t},\cdot) <   -\frac{\min_{1 \leq i\leq d}S^{i}_{t}(\o^{t})}{2 \sqrt{n}},\o^{t}\right) \geq \min_{k \in K} p_{t+1}^{0}\left(G_{k}(\o^{t}),\o^{t}\right),$$
and using  \eqref{PoG},  \eqref{diffaoa2} is satisfied. \\
\end{proof}
\begin{remark}
Note that in both cases  ($S^{i}_{t}=Y^{i}_{t}$ and  $S^{i}_{t}=e^{Y^{i}_{t}}$),  we can choose $P^*=P^0$ in Theorem \ref{PPstar}. 
\end{remark}

We now give a one dimension illustration of the previous setting where $\mathcal{Q}^T$ is not dominated. 
Take $n=d=1$ and $\O_{t}:=\O$ for some Polish space $\O$.  Let $Z$ be  some  real-valued random variable defined on  $\O$  and $p_0 \in \mathfrak{P}(\O)$ be such that  under $p_0$,  $Z$ is normally distributed with mean $0$ and standard deviation $1$. Set $P^0:=p_0\otimes \cdots \otimes p_0$ and $Z_{t+1}(\o^{t},\o_{t+1}):=Z(\o_{t+1})$ for all $0 \leq t \leq T-1$ and $\o^t \in \Omega^t$.  Define $F: \mathfrak{P}(\O) \to \mathbb{R}^{2}$ by $F(p):=\left(E_{p}(Z), E_{p}\left(Z-E_{p}(Z)\right)^2-1\right)$ and $F(\o^{t},\o_{t+1}):=F(\o_{t+1})$ for all $0 \leq t \leq T-1$ and $\o^t \in \Omega^t$. Finally,  set  $\mathcal{Q}_{t+1}(\o^{t}):=\{p \in \mathfrak{P}(\O), F(p)=0\}=:\mathcal{Q}$ for all $0 \leq t \leq T-1$, $\o^{t} \in \O^{t}$. For each $\o^{t}$, the law of the driving process $Z$ for the next period  is  centered with variance $1$ but not necessarily normally distributed.  \\
Assumption \ref{assdiff}  on the dynamic of $Y$  are verified if we choose $Y_0:=1$ and for all $0 \leq t \leq T-1$, $x \in \mathbb{R}$,   $(\o^{t},\o_{t+1}) \in  \O^{t} \times \O$ 
$$\mu_{t+1}(x,\o^{t},\o_{t+1}):= r_{t}(\o^{t}):=r, \;\;\; \nu_{t+1}(x,\o^{t}):=\sigma, \;\;\;q_{t}(\o^{t}):=\sigma^{2},$$   for some  $r \in \mathbb{R}$ and $\sigma >0$ fixed. \\
As $\Delta Y_{t}= r + \sigma Z$ and $Z$ is normally distributed with mean $0$ and standard deviation $1$ under $p_0$, \eqref{bof} (or \eqref{bof2}) implies that   $$\kappa_{t}=\kappa=\min\left( \Phi_{}\left(-\frac{\mbox{ln 2}+r}{\sigma}\right), 1-\Phi_{}\left(\frac{\mbox{ln 2}-r}{\sigma}\right)\right) $$ where $\Phi_{}$ is the cumulative distribution function  of some normal law with mean $0$ and standard deviation $1$. We have already seen that $\beta_{t}(\o^{t})= \beta={(\mbox{ln 2}})/{\sqrt{n}}$  when $S_{t}={Y_{t}}$.  In the other case, $S_{t}(\o^{t})=\mbox{exp}\left({Y_{t}(\o^{t})}\right)= \mbox{exp}\left(1+  r t + \sigma \sum_{i=1}^{t} Z(\o_{i})\right)$  and $\beta_{t}(\o^{t}) =(1/2)\min\left(1, S_{t}(\o^{t})\right)$ (see \eqref{bof2}).\\
 Finally, the set $\mathcal{Q}^{T}$ is not dominated. Indeed, we show that $\mathcal{Q}$ is not dominated and conclude  using Proposition \ref{DOM}.   Assume that there is some $\widehat{p} \in \mathfrak{P}(\O)$ which dominates $\mathcal{Q}$. For $x \neq 0$, let $q_{x} \in \mathfrak{P}(\O)$ such that $$q_{x}(Z=x)=\frac{1}{2x^{2}},  \, \, q(Z=-x)=\frac{1}{2x^{2}},  \,\, q(Z=0)=1-\frac{1}{x^{2}}.$$  Then $q_{x} \in \mathcal{Q}$ and  $\{x \in \mathbb{R},\; \widehat{p}(\{Z=x\})>0\}= \mathbb{R} \backslash\{0\}$, a contradiction.
\section{Proofs}
\label{Apen}

The first  section presents the  one-period  version of our  problems  with deterministic initial data. We will study  the different notions of arbitrage and their equivalence (see Proposition  \ref{Arbeqone}). We also prove Proposition \ref{bay} that will be used in the proof of Theorem \ref{PPstar}.
In the second  section  the multi-period  results are proved relying on the one-period results together with   measurable selections technics. Finally, the third section presents 
the proof of Proposition \ref{DOM}. 
\subsection{One-period model}
\label{OneOne}
 Let $(\overline{\Omega}, \Gc)$ be a measured space, $\mathfrak{P}(\overline{\Omega})$ the set of all probability measures  defined on $\mathcal{G}$  and $\mathcal{Q}$  a non-empty convex subset of $\mathfrak{P}(\overline{\Omega})$. For $P \in \mathcal{Q}$ fixed,  $E_{P}$ denotes the expectation under $P$. Let $Y$ be a $\Gc$-measurable $\mathbb{R}^{d}$-valued random variable.\\
The   following sets  are the pendant in the one-period case of the ones introduced in Definition \ref{DefDARB}. Let $P \in \mathcal{Q}$  \begin{align}
{E}(P)&:=\bigcap  \left\{ A \subset \mathbb{R}^{d},\; \mbox{closed}, \; P\left(Y(.) \in A\right) =1\right\}, \\
{D}&:=\bigcap\left\{ A \subset \mathbb{R}^{d},\; \mbox{closed},\; P \left(Y(\cdot) \in A \right)=1, \; \forall P \in \mathcal{Q} \right\}.
\end{align}

\noindent The next lemma will be used in the proof of Proposition \ref{alphatmesARB}. 

\begin{lemma}
\label{util}
Let $C$ be a convex set of $\mathbb{R}^{d}$ and fix some $\varepsilon>0$. Then  $B(0,\varepsilon) \cap \mbox{Aff} (C) \subset \overline{C}$  if and only if  $B(0,\varepsilon) \cap \mbox{Aff} (C) \subset {C}$.
\end{lemma}
\begin{proof}
The reverse implication is trivial. Assume that $B(0,\varepsilon) \cap \mbox{Aff} (C) \subset \overline{C}$ and  let $x \in B(0,\varepsilon) \cap \mbox{Aff} (C)$. As $|x|<\varepsilon$, there exists some $\delta>0$ such that $B(x,\delta)   \cap \mbox{Aff} (C) \subset  B(0,\varepsilon) \cap \mbox{Aff} (C) \subset  \overline{C}.$ Hence $x \in \mbox{Ri}(\overline{C})=\mbox{Ri}(C) \subset C$ (see \citep[Theorem 6.3 p46]{cvx}).
\end{proof}\\

This lemma allows an easy characterisation of the support and was used several time in the paper. 
\begin{lemma}
\label{supportcar}
 Let $h \in \mathbb{R}^{d}$ and   $P \in \mathfrak{P}(\overline{\O})$ be fixed. Then,  $h \in {E}(P)$ if and only if for all $n \geq 1$, $P\left( Y(\cdot) \in B\left(h,{1}/{n}\right)\right)>0$. Similarly,  $h \in {D}$ if and only if for all $n \geq 1$, there exists some $P^{n} \in \mathcal{Q}$, such that  $P^{n}\left( Y(\cdot) \in B\left(h,{1}/{n}\right)\right)>0$.
\end{lemma}
\begin{proof}
Fix some $h \in \mathbb{R}^{d}$. By definition $h \notin  {E}(P)$ if and only if there exists an open set $O \subset \mathbb{R}^{d}$ such that $h \in O$ and  $P( Y(\cdot)\in O)=0$ and the first item follows. 
Similarly, $h \notin  {D}$ if and only if there exists an open set $O \subset \mathbb{R}^{d}$ such that $h \in O$ and  $P( Y(\cdot) \in O)=0$ for all $P \in \mathcal{Q}$ and the second item follows. 
\end{proof}\\

Now, we introduce the  definitions  of no-arbitrage in this one period setting. The first one is the one-period pendant of the $NA(\mathcal{Q}^{T})$ condition while the two others are the pendant of Definitions \ref{NAGarb} and \ref{NAQarb}.
\begin{definition}
\label{AOAone}
The one-period  no-arbitrage condition holds true if   $h  Y(\cdot) \geq 0$ $\mathcal{Q}$-q.s. for some  $h \in \mathbb{R}^{d}$ implies that $h Y(\cdot)=0$ $\mathcal{Q}$-q.s.
\end{definition}

\begin{definition}
\label{NAGone}
The one-period  geometric no-arbitrage condition holds true if $0 \in \mbox{Ri}\left(\mbox{Conv}(D)\right).$ This is equivalent to $0 \in \mbox{Conv}(D)$ and  there exists some $\varepsilon>0$ such that
$B(0,\varepsilon) \cap { \mbox{Aff}(D)}\subset {\mbox{Conv}}(D).$
\end{definition}

\begin{definition}
\label{NAQone}
The one-period  quantitative no-arbitrage condition holds true if there exists some constants $\beta, \kappa \in (0,1]$ 
such that for all  $h \in  \mbox{Aff}(D)$, $h \neq 0$ there exists $P_{h} \in \mathcal{Q}$ satisfying
\begin{align}
\label{alphaone}
P_{h}( h Y(\cdot) < -\beta |h|) \geq  \kappa.
\end{align}
\end{definition}

\begin{remark}
\label{RIo}
We recall that if $0 \notin \mbox{Ri}\left({\mbox{Conv}}(D)\right)$ there exists some $h^* \in \mbox{Aff}(D)$, $h^* \neq 0$ such that $h^* Y(\cdot) \geq 0$ $\mathcal{Q}$-q.s. This is a classical exercise relying on separation arguments in $\mathbb{R}^{d}$, see  \citep[Theorems 11.1, 11.3 p97]{cvx} or \citep[Proposition A.1]{fs}. 
\end{remark}
Proposition \ref{Arbeqone}  establishes that these three preceding conditions are actually equivalent. 

\begin{proposition}
\label{Arbeqone}
Definitions \ref{AOAone}, \ref{NAGone} and \ref{NAQone} are equivalent. Moreover, one can choose $\beta={\varepsilon}/{2}$ in  \eqref{alphaone}  where $\varepsilon>0$ is such that
$B(0,\varepsilon) \cap { \mbox{Aff}(D)}\subset {\mbox{Conv}}(D)$ in Definition \ref{NAGone}. \end{proposition}
\begin{proof}
{\it Step 1 :  Definition \ref{AOAone} implies  Definitions  \ref{NAGone} and \ref{NAQone}.} \\
First we show by contradiction that for all $h \in \mbox{Aff}(D)$
\begin{align} \label{NPome2ARB}
hY(\cdot) \geq 0 \; \mathcal{Q}\mbox{-q.s.} \Rightarrow h = 0.
\end{align}
Assume that there exists some $h \in \mbox{Aff}(D)$, $h \neq 0$ such that $hY(\cdot) \geq 0 \; \mathcal{Q}$-q.s.  Definition \ref{AOAone} implies that $ hY(\cdot) = 0 \; \mathcal{Q}\mbox{-q.s.}$ and\footnote{${X}^{\perp}$ stands for the orthogonal space of some set $X.$}  
$$h \in \{h \in \mathbb{R}^{d}, hy=0 \; \mbox{for all $y \in {D}$}\}= {D}^{\perp}=\left( \mbox{Aff}(D) \right)^{\perp},$$
see for instance \citep[Proof of Lemma 2.6]{Nutz}.
This implies that $h \in  \mbox{Aff}(D) \cap \left( \mbox{Aff}(D) \right)^{\perp} \subset \{0\}$, a contradiction. \\
Now we show  that  Definition \ref{NAGone} holds true. If $0 \notin \mbox{Ri}\left( {\mbox{Conv}}(D)\right)$,  Remark \ref{RIo} 
implies that  there exists some $h^* \in \mbox{Aff}(D)$, $h^* \neq 0$ such that $h^* Y (\cdot) \geq 0$ $\mathcal{Q}$-q.s. which contradicts  \eqref{NPome2ARB}. 
  Then, we prove  that  Definition \ref{NAQone} holds also true. For all $n\geq 1,$ let
\begin{align}
\label{defAn}
A_{n}:=\left\{ h \in  \mbox{Aff}(D),\; |h|=1,\; P\left(hY(\cdot) < -\frac{1}{n}\right) < \frac{1}{n}  \; \forall P \in \mathcal{Q}\right\} \; \;   n_{0}:=\inf\{n \geq 1, A_{n}=\emptyset\}\end{align}
 with the convention that $\inf \emptyset=+\infty$. We have seen that Definition  \ref{NAGone} holds  true:  $0 \in \mbox{Ri}\left( {\mbox{Conv}}(D)\right)  \subset \mbox{Aff}(D)$ and $\mbox{Aff}(D)$ is a vector space.
 If $\mbox{Aff}(D) =\{0\}$, then $n_{0}=1 <\infty$.   Assume now that $\mbox{Aff}(D) \neq \{0\}$. We prove by
contradiction that  $n_{0}<\infty$.
Assume that $n_{0}=\infty$. For all $ n\geq 1$, there exists some $h_{n} \in A_n.$ 
By passing to a sub-sequence we can assume that $h_{n}$ tends to  some $h^{*}\in \mbox{Aff}(D)$ with $|h^{*}|=1$. Let $B_{n}:= \left\{ h_{n}Y(\cdot) <   -{1}/{n} \right\}$. Then
  $\{h^{*}Y(\cdot)<0 \} \subset \liminf_{n} B_{n}$ and Fatou's Lemma implies that for any  $P \in \mathcal{Q}$
\begin{align*}
P\left(h^*Y(\cdot) < 0\right) &
  \leq   \int_{\overline{\Omega}}
1_{\liminf_{n} B_{n}}(\o) P(d\o) \leq  \liminf_{n} \int_{\overline{\Omega}}
1_{B_n}(\o) P(d\o)=0.
\end{align*}
So $h^*Y(\cdot) \geq 0$ $\mathcal{Q}$-q.s. and  \eqref{NPome2ARB} implies that $h^{*}=0$
which contradicts $|h^{*}|=1$. Thus $n_{0}<\infty$ and we can set  $\beta=\kappa= {1}/{n_{0}}.$
It is clear that $\beta, \kappa \in (0,1]$ and by definition of $A_{n_{0}}$, \eqref{alphaone} holds true. \\

{\it Step 2 :  Definition \ref{NAQone} implies Definition \ref{NAGone}.} \\  
Else,  Remark \ref{RIo} implies that there exists  some $h^* \in \mbox{Aff}(D)$, $h^* \neq 0$ such that $h^* Y(\cdot) \geq 0$  $\mathcal{Q}$-q.s.: A contradiction with \eqref{alphaone}.\\ 
 
{\it Step 3 :  Definition \ref{NAGone} implies Definition \ref{AOAone}.} \\  Fix some $h \in \mathbb{R}^{d}$ such that $h Y(\cdot) \geq 0$ $\mathcal{Q}$-q.s.  Let  $p(h)$ be the orthogonal projection of $h$ on $\mbox{Aff}(D)$ (recall that  $\mbox{Aff}(D)$ is a vector space since    $0 \in \mbox{Ri}({\mbox{Conv}}(D)) \subset \mbox{Aff}(D)$). Assume for a moment  that $p(h)=0$. Remark \ref{supppp} shows that   $P( \{Y(\cdot) \in {D}\}) =1$ for all $P \in \mathcal{Q},$ $hY(\cdot)=p(h) Y(\cdot)=0$ $\mathcal{Q}$-q.s. and  Definition \ref{AOAone} is verified. \\ 
Next we show  that $hy \geq 0$ for all $y \in {D}$
and by convex combinations  for all $y \in {\mbox{Conv}}(D)$. Indeed if there exists $y_0 \in D$ such that $h y_0<0$, then there exists some $\delta>0$ such that $hy<0$ for all $y \in B(y_0,\delta)$.  But Lemma \ref{supportcar} implies the existence of some $P \in \mathcal{Q}$ such that $P(Y(\cdot) \in B(y_0, \delta))>0$, a contradiction. Now, if $p(h) \neq 0$,  as $0 \in \mbox{Ri}({\mbox{Conv}}(D))$, there exists some $\varepsilon>0$ such that $B(0,\varepsilon) \cap \mbox{Aff}(D) \subset  {\mbox{Conv}}(D),$ $-\varepsilon {p(h)}/{|p(h)|} \in {\mbox{Conv}}(D)$ and  $$-\varepsilon \frac{p(h)}{|p(h)|}h= -\varepsilon \frac{p(h)}{|p(h)|}p(h)<0,$$ a contradiction. \\

{\it Step 4: 
If  $B(0,\varepsilon) \cap \mbox{Aff}(D) \subset  {\mbox{Conv}}(D)$  one can choose $\beta={\varepsilon}/{2}$ in \eqref{alphaone}.}\\ This is similar to the proof of Definition \ref{NAGone} implies Definition \ref{NAQone}. The set  $A_{n}$ is modified by setting  $$A_{n}:=\left\{h \in \mbox{Aff}(D),\; |h|=1, P\left(h Y(\cdot)<-\frac{\varepsilon}{2}\right)< \frac{1}{n},\; \forall P \in \mathcal{Q}\right\}.$$  
The same arguments as before apply and 
if $n_0=\infty$ there exists  some $h^{*}\in \mbox{Aff}(D)$, $|h^{*}|=1$ such that $ h^{*} Y \geq -{\varepsilon}/{2}$ $\mathcal{Q}^{T}$-q.s.  We also get that $h^*y \geq -{\varepsilon}/{2}$ for all $y \in {\mbox{Conv}}(D)$.  Choosing $y=-(2/3)\varepsilon h^* \in B(0,\varepsilon) \cap \mbox{Aff}(D) \subset  {\mbox{Conv}}(D)$, we obtain  a contradiction. So, \eqref{alphaone} holds true with $\beta=\varepsilon/{2}$ and $\kappa={1}/{n_0}$\footnote{The same argument shows that one can set $\kappa=\inf_{h \in  \scriptsize{\mbox{Aff}}(D),\; |h|=1}\ \sup_{P \in \mathcal{Q}} P(h Y(\cdot) < -\frac{\varepsilon}{2})>0$  illustrating why  the  measurability of $\kappa$ cannot be directly obtained, see Remark \ref{alphabeta}.}.    \end{proof}\\

\noindent The next proposition follows from \citep[Lemma 2.2]{Bay17} and 
will be used in the proof of Theorem \ref{PPstar}. 
\begin{proposition}
\label{bay}
Assume that the one-period no-arbitrage condition (see Definition \ref{AOAone}) holds true. Then there exists some $P^* \in \mathcal{Q}$ such that $0 \in \mbox{Ri} \left({\mbox{Conv}}(E(P^*))\right)$  and $\mbox{Aff}(E(P^*))=\mbox{Aff}(D)$.
\end{proposition}
\begin{proof} \citep[Lemma 2.2]{Bay17} gives the existence of some $P^* \in \mathcal{Q}$ such that $NA(P^*)$  holds true and $\mbox{Aff}(E(P^*))=\mbox{Aff}(D)$. Note that the proof of  \citep[Lemma 2.2]{Bay17} relies on the convexity of $\mathcal{Q}$. Now Proposition \ref{singleP} (for $T=1$) shows that $0 \in \mbox{Ri} \left({\mbox{Conv}}(E(P^*))\right)$. 
\end{proof}\\

\subsection{Multi-period model}
\label{multi}
First we define the distance of a point $x \in \mathbb{R}^d$ to a set $F \subset \mathbb{R}^d$ by $d(x,F):=\inf \{|x-f|,\; f \in F\}$  and the Hausdorff distance between two sets  $F,G \subset \mathbb{R}^d$ by 
$d(F,G)=\sup_{x \in \mathbb{R}^d}|d(x,F)-d(x,G)|.$ 
\subsubsection{Proof of  Theorem \ref{ARBequival}}
\label{seoneARB}
\begin{proof}
Fix some $ 0 \leq t \leq T-1$ and  $\o^{t} \in \O^{t}.$ We say that the $ NA(\mathcal{Q}_{t+1}(\o^{t}))$ condition holds true if $h \Delta S_{t+1}(\o^{t},\cdot) \geq 0$ $\mathcal{Q}_{t+1}(\o^{t})$-q.s. for some $h \in \mathbb{R}^{d}$ implies that $h \Delta S_{t+1}(\o^{t},\cdot) = 0$ $\mathcal{Q}_{t+1}(\o^{t})$-q.s.  Proposition \ref{Arbeqone} implies that  the $NA(\mathcal{Q}_{t+1}(\o^{t}))$ condition  is equivalent to \eqref{vamakiARB} and \eqref{valakiARB} for any $\o^{t} \in \O^{t}$.
Then Theorem \ref{bnlocal} shows that Definition  \ref{NAQTARB} is equivalent to the fact that
\begin{align}
\label{OTNA}
\O^{t}_{NA}=\{ \o^{t} \in \O^{t},\; \;NA(\mathcal{Q}_{t+1}(\o^{t})) \mbox{ holds true}\}
\end{align}
 is a $\mathcal{Q}^{t}$-full measure set and belongs to $\mathcal{B}_{c}(\O^{t})$ for all $0 \leq t \leq T-1$. Thus,  for all $0 \leq t \leq T-1,$ one may choose $\O^{t}_{NA}= \O^{t}_{qNA}=\O^{t}_{gNA}.$  Furthermore,  Proposition \ref{Arbeqone} shows that one can take $\beta_{t}(\o^{t})={\varepsilon_{t}}(\o^{t})/{2}$ for $\o^{t} \in  \O^{t}_{NA}$.
\end{proof}
\subsubsection{Proof of  Proposition \ref{alphatmesARB}}
\label{proofmesepsi}

\noindent{\it Proof of Proposition \ref{alphatmesARB}}\\
\begin{proof}
Fix some $0\leq t\leq T-1$. We set  $\Gamma^{t+1}(\o^{t})=\emptyset$ for $\o^{t} \notin \O_{NA}^{t}$ and  for all $\o^{t}  \in \O_{NA}^{t}$
\begin{align*}
\Gamma^{t+1}(\o^{t})&:=\left\{ \varepsilon \in \mathbb{Q},\; \varepsilon>0,\; B(0,\varepsilon) \cap \mbox{Aff}\left(D^{t+1}\right)(\o^{t}) \subset {\mbox{Conv}}\left({D}^{t+1}\right)(\o^{t})\right\}
\\ &=\left\{ \varepsilon \in \mathbb{Q},\; \varepsilon>0,\; B(0,\varepsilon) \cap \mbox{Aff}\left(D^{t+1}\right)(\o^{t}) \subset  \overline{\mbox{Conv}}\left({D}^{t+1}\right)(\o^{t})\right\},
\end{align*}
where the equality comes from  Lemma \ref{util}.
 Assume for a moment that  $\mbox{graph} \; \Gamma^{t+1} \in  \mathcal{B}_{c}(\O^{t}) \otimes \mathcal{B}(\mathbb{R}^{d})$ has been  proved. The Aumann Theorem implies the existence of  a $\ \mathcal{B}_{c}(\O^{t})$-measurable selector ${\varepsilon}_{t}: \{ \Gamma^{t+1} \neq \emptyset \}  \to \mathbb{R}$ such that ${\varepsilon}_{t}(\o^t) \in \Gamma^{t+1}(\o^t)$ for every $\o^t \in  \{ \Gamma^{t+1} \neq \emptyset \}$. Now,  Theorem \ref{ARBequival} and \eqref{vamakiARB} imply that $\O^{t}_{NA} = \left\{ \Gamma^{t+1}(\o^{t}) \neq \emptyset \right\}$ (recall that  $\Gamma^{t+1}(\o^{t})=\emptyset$ outside $\O_{NA}^{t}$).
Setting $\varepsilon_{t}=1$ outside  $\O^{t}_{NA},$ $\varepsilon_{t}$ is $\mathcal{B}_{c}(\O^{t})$-measurable  and Proposition \ref{alphatmesARB} is proved as we can choose  $\beta_{t}={\varepsilon_{t}}/2$ (see Theorem \ref{ARBequival}). \\
It remains to  show   that $\mbox{graph} \; \Gamma^{t+1} \in  \mathcal{B}_{c}(\O^{t}) \otimes \mathcal{B}(\mathbb{R}^{d})$.  For all $\varepsilon>0$, $\varepsilon \in \mathbb{Q}$, let
$$A_{\varepsilon}:=\left\{ \o^{t} \in \O_{NA}^{t},\;  B(0,\varepsilon) \cap \mbox{Aff}\left(D^{t+1}\right)(\o^{t}) \subset \overline{\mbox{Conv}}\left({D}^{t+1}\right)(\o^{t})\right\}.$$
As $\mbox{graph} \; \Gamma^{t+1} =\bigcup_{\varepsilon \in \mathbb{Q},\; \varepsilon>0} A_{\varepsilon} \times \{\varepsilon\},$
it is enough to prove that $A_{\varepsilon} \in  \mathcal{B}_{c}(\O^{t})$.  Let  $h: \mathbb{R}^d \times \O^t$ be defined by 
$$h(x,\o^t) :=d\left(x,B(0,\varepsilon) \cap \mbox{Aff}\left(D^{t+1}\right)(\o^{t}) \right)- d\left(x,\overline{\mbox{Conv}}\left({D}^{t+1}\right)(\o^{t}) \right).$$
Then \citep[Theorem 18.5 p595]{Hitch} and Lemma \ref{Dmeasurability} show that  for all $ x\in \mathbb{R}^d$  $h(x,\cdot)$ is $\mathcal{B}_c(\O^{t})$-measurable and  that $h(\cdot,\o^t)$ is continuous for all $\o^t \in \Omega^t$. 
As $ \overline{\mbox{Conv}}\left({D}^{t+1}\right)(\o^{t})$ is  closed-valued, 
$$A_{\varepsilon}=\left\{ \o^{t} \in \O_{NA}^{t},\;  h(x,\o^t)\geq 0, \, \forall x \in \mathbb{R}^d\right\} =
\cap_{q \in \mathbb{Q}^d}\left\{ \o^{t} \in \O_{NA}^{t},\;  h(q,\o^t)\geq 0\right\} \in\mathcal{B}_c(\O^{t}).$$

\end{proof}

\subsubsection{Proof of Theorem \ref{PPstar}}
\label{setwoARB2}
As mentioned in Remark \ref{OW}, our proof uses similar ideas as the one   used  in the proof of  \citep[Theorem 3.1]{OW18} and relies crucially  on the measurability and convexity of Graph($\mathcal{Q}_{t+1}(\o^t))$  (see Assumption \ref{QanalyticARB}). \\
\begin{proof}\\
{\it Reverse implication.}\\
Fix some $0 \leq t \leq T-1$ and $\o^{t} \in \O^{t}_{NA}.$   
As $P^* \in \mathcal{Q}^{T},$ Remark \ref{DomainInc} implies that 
${D}_{P^{*}}^{t+1}(\o^{t})  \subset  {D}^{t+1}(\o^{t}).$ 
As $\mbox{Aff}\left(D^{t+1}\right)(\o^{t})=\mbox{Aff}\left(D_{P^*}^{t+1}\right)(\o^{t})$ and $0 \in \mbox{Ri} \left({\mbox{Conv}}(D_{P^*}^{t+1})\right)(\o^{t})$, there exists some $\varepsilon>0$ such that
   \begin{align*}
 B(0,\varepsilon) \bigcap \mbox{Aff}\left(D^{t+1}\right)(\o^{t}) &=B(0,\varepsilon) \bigcap \mbox{Aff}\left(D_{P^*}^{t+1}\right)(\o^{t})
   \subset  {\mbox{Conv}}(D_{P^*}^{t+1})(\o^{t}) \subset {\mbox{Conv}}(D^{t+1})(\o^{t})
  \end{align*}
 and NA$(\mathcal{Q}^{T})$ follows from Theorem \ref{ARBequival}.\\

{\it Direct implication.}\\
For all $0 \leq t \leq T-1$, let $\mathcal{E}_{t+1}: \O^{t} \twoheadrightarrow \mathfrak{P}(\O_{t+1})$ be defined by $\mathcal{E}_{t+1}(\o^{t})=\emptyset$ if $\o^{t} \in \O^{t}_{NA}$ and if $\o^{t} \in \O^{t}_{NA}$ 
\begin{align*}\mathcal{E}_{t+1}(\o^{t}):=\{p  \in \mathcal{Q}_{t+1}(\o^{t}),\;  0 \in \mbox{Ri} \left({\mbox{Conv}}(E^{t+1})\right)(\o^{t},p) \mbox{ and } 
&  \mbox{Aff}\left(E^{t+1}\right)(\o^{t},p)=\mbox{Aff}\left(D^{t+1}\right)(\o^{t}) \}.  \end{align*}
Theorem \ref{bnlocal} and Proposition \ref{bay} show that $\O^{t}_{NA}=\{ \mathcal{E}_{t+1} \neq \emptyset\}$. Assume for a moment that we have proved the existence 
 of  $p_{t+1}^{*}\in \mathcal{SK}_{t+1}$  such that  $p_{t+1}^{*}(\cdot,\o^{t}) \in \mathcal{E}_{t+1}(\o^{t})$ for all $\o^{t} \in \O^{t}_{NA} $.   Let 
$P^*:=p^*_1 \otimes \cdots \otimes p^*_{T}.$
Then,  $P^* \in \mathcal{Q}^{T}$ (see \eqref{QstarARB}), \eqref{DvsE} implies that $$\mbox{Aff}\left(D_{P^*}^{t+1}\right)(\o^{t})=\mbox{Aff}\left(E^{t+1}\right) (\o^{t},p^*_{t+1}(\cdot,\o^{t})) \subset \mbox{Aff} \left(D^{t+1}\right)(\o^{t})
=\mbox{Aff}\left(E^{t+1}\right)(\o^{t},p^*_{t+1}(\cdot,\o^{t}))$$ and 
$0 \in \mbox{Ri} \left({\mbox{Conv}}(D_{P^*}^{t+1})\right)(\o^{t})$  for all  $\o^t \in \O_{NA}^{t}$.
 \\
So it remains to prove the existence of $p^*_{t+1}$. Fix some  $0 \leq t \leq T-1$ and let
\begin{align*}
B&:=\{(\o^t,p) \in \O^{t} \times \mathfrak{P}(\O_{t+1}),\;  \mbox{Ri} \left(\overline{\mbox{Conv}}(E^{t+1}) \right)(\o^{t},p)\cap \{0\}\neq \emptyset\},\\
C&:=\{(\o^t,p)\in \O^{t} \times \mathfrak{P}(\O_{t+1}),\; \mbox{Aff}\left(E^{t+1}\right)(\o^{t},p)=\mbox{Aff}\left(D^{t+1}\right)(\o^{t})  \}.
\end{align*}
\citep[Lemma 5.6]{Artstein} and  Lemma \ref{Dmeasurability} show that $\mbox{Ri} \left(\overline{\mbox{Conv}}(E^{t+1})\right)$ is $\mathcal{B}(\O^{t}) \otimes \mathcal{B}( \mathfrak{P}(\O_{t+1}))$-measurable and $B\in \mathcal{B}(\O^{t}) \otimes \mathcal{B}( \mathfrak{P}(\O_{t+1}))$ follows. 
Let $h$ be defined by 
$$h(\o^t,p):=d\left(\mbox{Aff}\left(E^{t+1}\right)(\o^{t},p),\mbox{Aff}\left(D^{t+1}\right)(\o^{t})\right).$$
Note that $C=\{h^{-1}(0)  \}$. Then \citep[Theorem 18.5 p595]{Hitch} and Lemma \ref{Dmeasurability} show that  $x \in \mathbb{R}^d$ 
$(\o^t,p) \to d\left(x,\mbox{Aff}\left(E^{t+1}\right)(\o^{t},p)\right)$ is $\mathcal{B}(\O^{t}) \otimes \mathcal{B}( \mathfrak{P}(\O_{t+1}))$ measurable and $\o^t \to d\left(x,\mbox{Aff}\left(D^{t+1}\right)(\o^{t})\right)$ is $\mathcal{B}_c(\O^{t})$ measurable. 
They also show \\
$x \to |d\left(x,\mbox{Aff}\left(E^{t+1}\right)(\o^{t},p)\right)-d\left(x,\mbox{Aff}\left(D^{t+1}\right)(\o^{t})\right)|$ is continuous. Thus 
\begin{align}
h(\o^t,p)= 
 \sup_{x \in \mathbb{Q}^d}|d\left(x,\mbox{Aff}\left(E^{t+1}\right)(\o^{t},p)\right)-d\left(x,\mbox{Aff}\left(D^{t+1}\right)(\o^{t})\right)|
 \end{align}
 and $h$ is  $\mathcal{B}_c(\O^{t}) \otimes \mathcal{B}( \mathfrak{P}(\O_{t+1}))$ measurable. It follows that $C \in 
\mathcal{B}_c(\O^{t}) \otimes \mathcal{B}( \mathfrak{P}(\O_{t+1})). $ 
\citep[Theorem 6.3 p46]{cvx}, Assumption \ref{QanalyticARB} 
   and Lemma \ref{easy} show that 
$$ \mbox{graph}\left(\mathcal{E}_{t+1}\right)=\mbox{graph}\left(\mathcal{Q}_{t+1}\right) \cap B \cap C  \in  \mathfrak{A}\left(\mathcal{B}_{c}(\O^{t}) \otimes  \mathfrak{P}(\O_{t+1})\right),
$$
where for some Polish space $X$ and  some paving $\mathcal{J}$ (i.e.  a non-empty collection of subsets of $X$ containing the empty set),  $\mathfrak{A}(\mathcal{J})$ denotes the set of all nuclei of Suslin Scheme on $\mathcal{J}$ (see \citep[Definition 7.15 p157]{BS}).
Now  \citep[Lemma 4.11]{BN} (which relies on \citep{Lee78})  gives the existence of 
  $p_{t+1}^{*}\in \mathcal{SK}_{t+1}$  such that  $p_{t+1}^{*}(\cdot,\o^{t}) \in \mathcal{E}_{t+1}(\o^{t})$ for all $\o^{t} \in \O^{t}_{NA}   = \{\mathcal{E}_{t+1} \neq \emptyset\}$. The proof is complete.
  \end{proof}\\
   
\noindent The following  lemma was used in the previous proof.

\begin{lemma}
\label{easy}
Let $X,Y$ be two Polish spaces. Let $\Gamma_{1} \in \mathcal{A}(X\times Y)$ and $\Gamma_{2} \in \mathcal{B}_{c}(X) \otimes \mathcal{B}(Y)$. Then $\Gamma_{1} \cap \Gamma_{2} \in  \mathfrak{A}\left(\mathcal{B}_{c}(X) \otimes \mathcal{B}(Y)\right)$.
\end{lemma}
\begin{proof}
 \citep[Proposition 7.35 p158, Proposition 7.41 p166]{BS} imply that   $\Gamma_{1} \in \mathcal{A}(X\times Y)=\mathfrak{A}(\mathcal{B}(X)\otimes \mathcal{B}(Y)) \subset \mathfrak{A}\left(\mathcal{B}_{c}(X) \otimes \mathcal{B}(Y)\right)$ and   $\Gamma_{2} \in \mathcal{B}_{c}(X) \otimes \mathcal{B}(Y) \subset \mathfrak{A}\left(\mathcal{B}_{c}(X) \otimes \mathcal{B}(Y)\right)$ and thus  $\Gamma_{1} \cap \Gamma_{2} \in  \mathfrak{A}\left(\mathcal{B}_{c}(X) \otimes \mathcal{B}(Y)\right)$.
\end{proof}\\

\subsubsection{Proof of Theorem \ref{TheoS}}
\label{setwoARB}
\begin{proof} 
{\it Step 1: Reverse implication.}\\
Lemma \ref{exex} implies that the $NA(\mathcal{P}^{T})$ condition  holds true and  Lemma  \ref{PRpolar} shows that the $NA(\mathcal{Q}^{T})$ is satisfied.\\

{\it Step 2: Direct implication.}\\
Theorem \ref{PPstar} implies that there exists  some  $P^*\in \mathcal{Q}^{T}$ with the fixed disintegration  $P^*:=P_1^{*}\otimes p_{2}^{*}\otimes \cdots \otimes p_{T}^{*}$  such that   $\mbox{Aff}\left(D_{P^*}^{t+1}\right)(\o^{t})=\mbox{Aff} \left(D^{t+1}\right)(\o^{t})$ and $0 \in \mbox{Ri} \left({\mbox{Conv}}\left(D_{P^*}^{t+1}\right)\right)(\o^{t})$
  for all  $\o^t  \in \O^{t}_{NA}$ and all $0 \leq t \leq T-1$. 
The direct implication holds true if  {\it i), ii)} and {\it iii)} below are proved\footnote{Note that $i$) and $ii)$   are  true if we only  assume that $P^{*} \in \mathcal{Q}^{T}$.}. \\
\indent {\it i) $\mathcal{P}^{t} \subset \mathcal{Q}^{t}$  for all $1 \leq t \leq T$}.\\
\noindent  This follows by induction from    \eqref{PstarARB}, $p^{*}_{t+1}(\cdot,\o^{t}) \in \mathcal{Q}_{t+1}(\o^{t})$ and  the convexity of $\mathcal{Q}_{t+1}(\o^{t})$.\\
 \indent  {\it ii)  $\mathcal{Q}^{t}$ and $\mathcal{P}^{t}$ have the same polar-sets  for all $1 \leq t \leq T$}.\\
Fix some $1 \leq t \leq T$.   As  $\mathcal{P}^{t} \subset \mathcal{Q}^{t}$,  it is clear that  a $\mathcal{Q}^{t}$-polar set is also a  $\mathcal{P}^{t}$-polar set.
To establish the other inclusion,  we  prove by induction  that  for all $1 \leq t \leq T$  and all $Q^{t} \in \mathcal{Q}^{t}$, there exists some $(\l_{1}^t,  \cdots, \l_{2t}^t) \in (0,1]^{2t}$ such that $\sum_{i=1}^{2t} \l_{i}^t=1$ and some  $\left(R_{i}^t\right)_{3 \leq i \leq 2t} \subset \mathcal{Q}^{t}$ (if $ t \geq 2$)  such that 
\begin{align}
\label{reccu}
Q^{t} << P^{t}:= \l_{1}^t P^{*t}+ \l_{2}^t Q^{t} + \sum_{i=3}^{2t} \l_{i}^t R_{i}^t \in \mathcal{P}^{t}.
\end{align} 
For $t=1$, let $Q_{1} \in \mathcal{Q}^{1}$ and $P_{1}:= (P_{1}^*+Q_{1})/2$. Then, $Q_{1} << P_{1}$ and $P_{1}\in \mathcal{P}^{1},$ see \eqref{PstarARB}. Now assume that the property is true for some  $t \geq1.$ Let 
$Q^{t+1}\in \mathcal{Q}^{t+1}$ with the fixed disintegration $Q^{t+1}:=Q^{t} \otimes q_{t+1}$ where 
$Q^{t}\in \mathcal{Q}^{t}$ and   $q_{t+1}(\cdot,\o^t) \in \mathcal{Q}_{t+1}(\o^t)$ for all $\o^t \in \Omega^t$. Then, there exists some $\left(R_{i}^t\right)_{3 \leq i \leq 2t} \subset \mathcal{Q}^{t}$, $(\l_{1}^t,  \cdots, \l_{2t}^t)  \in (0,1]^{2t}$  such that  \eqref{reccu} holds true. Let 
\begin{eqnarray*}
P^{t+1} &:= & P^{t} \otimes \frac{1}{2}(p_{t+1}^{*}+q_{t+1})  \\
 R_{i}^{t+1} &:=& R_{i}^t\otimes \frac{1}{2}(p^{*}_{t+1}+q_{t+1})  
 \quad \quad 
\l_{i}^{t+1}:= \l_{i}^t  \quad \quad  \forall 3\leq i\leq 2t \\
 R_{2t+1}^{t+1}&:= &Q^{t} \otimes p^{*}_{t+1}  \quad \quad R_{2t+2}^{t+1}:=P^{*t} \otimes q_{t+1}\\
 \l_{1}^{t+1}&:= &\frac{\l_{1}^t}{2} \quad \quad  \l_{2}^{t+1}:=\frac{\l_{2}^t}{2} \quad \quad \l_{2t+1}^{t+1}:= \frac{\l_{2}^t}{2} \quad \quad  \l_{2t+2}^{t+1}:=\frac{\l_{1}^t}{2}. 
\end{eqnarray*}
Then  $
P^{t+1}  \in \mathcal{P}^{t+1}$ (see  \eqref{PstarARB}),  $(R_{i}^{t+1})_{3 \leq i \leq 2(t+1)} \subset  \mathcal{Q}^{t+1},$ 
 $\sum_{i=1}^{2(t+1)} \l_{i}^{t+1}=1$ and 
\begin{align*}
 P^{t+1}
 = \l_{1}^{t+1}P^{*t+1}+ \l_{2}^{t+1} Q^{t+1} + \sum_{i=3}^{2(t+1)} \l_{i}^{t+1} R_{i}^{t+1}.
 \end{align*}   
As $Q^{t+1} << P^{t+1},$ the induction is proven.\\
 
 \indent {\it iii) The $sNA(\mathcal{P}^{T})$ condition holds true}. \\
  Fix  some $P \in \mathcal{P}^{T} \subset  \mathcal{Q}^{T}$, some $0 \leq t \leq T-1$ and $\o^{t} \in \O_{NA}^{t}$. We  establish that    $0 \in  \mbox{Ri} \left({\mbox{Conv}}\left(D_{P}^{t+1}\right)\right)(\o^{t})$. Then  $P^{t}\left(\O^{t}_{NA}\right)=1$  and Proposition  \ref{singleP} shows that $NA(P)$ holds true and $iii)$ follows. 
Remark \ref{DomainInc} and \eqref{PstarARB} imply  that 
${D}_{P^{*}}^{t+1}(\o^{t}) \subset {D}_{P}^{t+1}(\o^{t}) \subset  {D}^{t+1}(\o^{t})$. Thus, $0 \in {\mbox{Conv}}(D_{P^*}^{t+1})(\o^{t}) \subset {\mbox{Conv}}(D_{P}^{t+1})(\o^{t})$. We  have that    
$$\mbox{Aff} \left(D^{t+1}\right)(\o^{t}) = \mbox{Aff}\left(D_{P^*}^{t+1}\right)(\o^{t}) \subset \mbox{Aff} \left(D_P^{t+1}\right)(\o^{t}) \subset \mbox{Aff} \left(D^{t+1}\right)(\o^{t}).$$ 
As  $0 \in \mbox{Ri} \left({\mbox{Conv}}(D_{P^*}^{t+1})\right)(\o^{t})$, there exists some $\varepsilon>0$ such that
\begin{small}
   \begin{align*}
 B(0,\varepsilon) \bigcap \mbox{Aff}\left(D_{P}^{t+1}\right)(\o^{t}) &=B(0,\varepsilon) \bigcap \mbox{Aff}\left(D_{P^*}^{t+1}\right)(\o^{t})
   \subset  {\mbox{Conv}}(D_{P^*}^{t+1})(\o^{t}) \subset {\mbox{Conv}}(D_{P}^{t+1})(\o^{t}).
  \end{align*}
  \end{small}
  
\end{proof}


\subsubsection{Proof of Proposition \ref{finally}}
\label{setwoARBe}
\begin{proof}
Theorem \ref{PPstar} implies that there exists  some $P^*\in \mathcal{Q}^{T}$ with the fixed disintegration  $P^*=P_1^{*}\otimes p_{2}^{*}\otimes \cdots \otimes p_{T}^{*}$ such   that    $\mbox{Aff} \left(D_{P^*}^{t+1}\right)(\o^{t})=\mbox{Aff} \left(D^{t+1}\right)(\o^{t})$  and $0 \in \mbox{Ri} \left({\mbox{Conv}}(D_{P^*}^{t+1})\right)(\o^{t} )$ for all $\o^{t} \in \O^{t}_{NA}$ and all $0 \leq t \leq T-1$.
To   find a $\mathcal{B}_{c}(\O^{t})$-measurable version of $\beta_{t}$ and ${\kappa}_{t}$ in \eqref{valakiARB} we follow the same idea as in \citep[Proposition 3.7]{BCR18}. Fix some $0 \leq t \leq T-1$. Set $n_t(\o^{t}):=\inf\{n \geq 1, \;A^{P^*}_{n}(\o^{t}) =\emptyset\}$ where for  all $n \geq 1$  $A^{P^*}_{n}(\o^{t})=\emptyset$ if $\o^{t} \notin \O^{t}_{NA}$ and if $\o^{t} \in \O^{t}_{NA}$,
\begin{align}
\label{ANP}
A^{P^*}_{n}(\o^{t}):=\left\{h \in \mbox{Aff} \left(D_{P^*}^{t+1}\right)(\o^{t}),\; |h|=1, \, p_{t+1}^{*}\left(h \Delta S_{t+1}(\o^{t},\cdot) < -\frac{1}{n}, \o^{t}\right) < \frac{1}{n}\right\}.
\end{align}
For all $\o^{t} \in \O^{t}$, as in the proof of Proposition \ref{Arbeqone}, $n_{t}(\o^{t})<\infty$ and one may set  $ \kappa_{t}(\o^{t})={\beta}_{t}(\o^{t}):={1}/{n_t(\o^{t})} \in (0, 1)$.
 Then, by definition of $A^{P^*}_{n}$,  \eqref{valakiARB} is  true with $P_{h}(\cdot)=p^{*}_{t+1}(\cdot,\o^{t})   \in \mathcal{Q}_{t+1}(\o^{t})$ since   $\mbox{Aff} \left(D_{P^*}^{t+1}\right)(\o^{t})=\mbox{Aff} \left(D^{t+1}\right)(\o^{t})$  for all $\o^{t} \in \O^{t}_{NA}$. 
\\
To prove that  $ \kappa_{t}={\beta}_{t}$ is $\mathcal{B}_{c}(\O^{t})$-measurable,    we show that  $\{A^{P^*}_{n} \neq \emptyset\} \in \mathcal{B}_{c}(\O^{t})$  since for all $k \geq 1$, 
$$\{n_{t} \geq k\} = \O^{t}_{NA} \cap \left( \bigcap_{1 \leq j \leq k-1} \{A^{P^*}_{j}=\emptyset\}\right).$$
Fix some $n \geq 1.$   As $p_{t+1}^{*}$ is only universally-measurable,  we use Lemma \ref{AAA} to prove that  $\{A^{P^*}_{n} \neq \emptyset\} \in \mathcal{B}_{c}(\O^{t})$. 
Fix  $P \in \mathfrak{P}(\O^{t})$.    First,   applying \citep[Lemma 7.28 p173]{BS}, there exists $p_{t+1}^{P}$ a Borel-measurable stochastic kernel on $\O_{t+1}$ given $\O^{t}$ and  $\O^{t}_{P} \in \mathcal{B}(\O^{t})$ such that $P(\O^{t}_{P})=1$ and  $p_{t+1}^{P}(\cdot,\o^{t})=p_{t+1}^{*}(\cdot,\o^{t})$ for all $\o^{t} \in \O^{t}_{P}$.
  Set $A^{P}_{n}$ as in \eqref{ANP} replacing $p^{*}_{t+1}$ with $p^{P}_{t+1}$ if $\o^{t} \in \O^{t}_{NA}$ (and  $A^{P}_{n}(\o^{t})=\emptyset$ if $\o^{t} \notin \O^{t}_{NA}$). 
Then 
$$\{A^{P^{*}}_{n} \neq \emptyset\} \cap \O^{t}_{P}=\{A^{P}_{n} \neq \emptyset\} \cap \O^{t}_{P}$$  and it remains to establish that 
$\{A^{P}_{n} \neq \emptyset\} \in \mathcal{B}_{c}(\O^{t})$.  
Remark that $$\mbox{graph}\left(A^{P}_{n}\right)= \mbox{graph}\left(\mbox{Aff}\left(D_{P^*}^{t+1}\right)\right) \bigcap \left\{ (\o^{t},h),\; |h|=1,\; p_{t+1}^{P}\left(h \Delta S_{t+1}(\o^{t},\cdot) < -\frac{1}{n},\o^{t}\right) < \frac{1}{n}\right\}.$$
 Lemma \ref{Dmeasurability} implies that $\mbox{graph}\left(\mbox{Aff}\left(D_{P^*}^{t+1}\right)\right) \in \mathcal{B}_{c}(\O^{t}) \otimes \mathcal{B}(\mathbb{R}^{d}).$  As $(\o^{t},h,\o_{t+1}) \to h \Delta S_{t+1}(\o^{t},\o_{t+1})$ and $p_{t+1}^{P}$ are Borel-measurable, \citep[Proposition 7.29 p144]{BS} implies that $ (\o^{t},h) \to p_{t+1}^{P}\left(h \Delta S_{t+1}(\o^{t},\cdot) < -{1}/{n},\o^{t}\right)$ is  $\mathcal{B}(\O^{t}) \otimes \mathcal{B}(\mathbb{R}^{d})$-measurable. Thus, applying the Projection Theorem,   
$ \mbox{Proj}_{\O^{t}}\left(\mbox{graph}\left(A^{P}_{n}\right)\right)=\{A^{P}_{n} \neq \emptyset\} \in \mathcal{B}_{c}(\O^{t})$ and the proof is complete.
\end{proof}
 \begin{lemma}
\label{AAA}
Let $X$ be a Polish space. Let $A \subset X$. Assume that for all $P \in \mathfrak{P}(X)$ there exists  some $A_{P} \in \mathcal{B}_{c}(X)$ and some  $P$-full measure set $X_{P} \in   \mathcal{B}(X)$ such that $A \cap X_{P}= A_{P} \cap X_{P}$. Then $A \in \mathcal{B}_{c}(X)$.
\end{lemma}
\begin{proof}
Fix some $P \in  \mathfrak{P}(X).$ We show that $A \in \mathcal{B}_{P}(X),$ the completion of $\mathcal{B}(X)$ with respect to $P.$ As this is true for all $P \in \mathfrak{P}(X)$, $A \in \mathcal{B}_{c}(X)$ will follow.\\
There exists $A_{P} \in \mathcal{B}_{c}(X)$ and  $X_{P} \in \mathcal{B}(X)$ such that $P(X_P)=1$ and  $A \cap X_P =A_{P} \cap X_P $. As  $A_{P} \cap X_P \in \mathcal{B}_{c}(\O^{t})  \subset \mathcal{B}_{P}(X)$ there exists a $P$-negligible set $N_P$ and $\tilde{A}_{P} \in \mathcal{B}(X)$ such that $A_{P} \cap X_P= \tilde{A}_{P} \cup N_{P}$. Now, let $M_P:=A \cap \left( X \backslash{X_P} \right) \subset X \backslash{X_P}$. As $X \backslash{X_P} \in \mathcal{B}(X)$ and  $P(X \backslash{X_P})=0$, $M_{P}$ is  a $P$-negligible set and  
$$A= \left(A \cap X_{P} \right) \cup \left(A \cap \left(X \backslash{X_P}\right)\right)=\tilde{A}_{P} \cup N_{P} \cup M_{P} \in \mathcal{B}_{P}(X).$$
\end{proof}
\subsubsection{Proof of Proposition \ref{CasD}}
\label{setwoARBef}
\begin{proof}
Lemma \ref{exex} implies that $NA(\widehat{P})$ and $NA(\mathcal{Q}^{T})$  are equivalent. 
Fix some disintegration of $\widehat{P} \in \mathcal{Q}^{T},$  $\widehat{P}:=\widehat{P}_{1}\otimes \widehat{p}_{2} \otimes \cdots \otimes \widehat{p}_{T}$ and some $1 \leq t \leq T$.   As $\widehat{P}^{t}$ dominates $\mathcal{Q}^{t}$ 
 Proposition \ref{singleP} implies  that 
$$0 \in \mbox{Ri}\left(\mbox{Conv}\left(D_{\widehat{P}}^{t+1}\right)\right)(\cdot) \; \; \mathcal{Q}^{t}\mbox{-q.s.}$$
Lemma \ref{Selection} below provides  a $\mathcal{Q}^{t}$-full measure set $\O^{t} \backslash{\O^{t}_{nd}}$ such that  $\widehat{p}_{t+1}(\cdot,\o^{t})$ dominates $\mathcal{Q}_{t+1}(\o^{t})$ for all $\o^{t} \in\O^{t} \backslash{\O^{t}_{nd}}$. 
Thus ${D}^{t+1}(\o^{t}) \subset {D}_{\widehat{P}}^{t+1}(\o^{t})$ and the equality follows from  \eqref{DvsE} as $\widehat{P} \in \mathcal{Q}^{T}.$ 
\end{proof}

\subsection{Proof of Proposition \ref{DOM}}
 \label{PRDOM}
 The proof of Proposition \ref{DOM}  follows directly from    Lemma \ref{Selection}. Indeed assume that the set $\mathcal{Q}^{T}$ is  dominated. As $\O^{t}_{N} \subset \O^t_{nd}$, $\O^{t}_{N}$ is a $\mathcal{Q}^{t}$-polar set which contradicts  $\widetilde{P}^t(\O^{t}_{N})>0$.\\  
 
The proof of Lemma \ref{Selection} is fairly technical and needs the introduction of the Wijsman topology as well as Lemma 
\ref{LK}. Note that  the reverse implication in Proposition \ref{DOM}  seems intuitive but  raises   challenging technical issues.\\  
Let $(X,d)$ be a Polish space  and $\mathcal{F}$ be  the set of non-empty closed subsets of $X$.  The Wijsman topology on $\mathcal{F}$ denoted by $\mathcal{T}_{W}$ is such that 
$$ F_{n} \underset{n\rightarrow +\infty}{\overset{\tau_w}{\longrightarrow}} F \iff d(x,F_{n})  \underset{n\rightarrow +\infty} \longrightarrow d(x,F) \; \mbox{for all $x \in X$,}$$
where $d(x,F):=\inf \{d(x,f),\; f \in F\}$.  
Note  that $\mathcal{F}$  endowed with $\mathcal{T}_{W}$ is a Polish space (see \citep{BE91}).  
\begin{lemma}
\label{LK}
  The function $(F,x) \in \mathcal{F} \times X \to 1_{F}(x)$ is $\mathcal{B}(\mathcal{F})\otimes \mathcal{B}(X)$-measurable.
\end{lemma}
\begin{proof}
The function $d:(x,F)  \in    X  \times \mathcal{F} \to d(x,F)$ is separately continuous. 
Indeed for all fixed $x \in X$, $d(x,\cdot)$ is continuous  by definition of $\mathcal{T}_{W}$ and  \citep[Theorem 3.16]{Hitch} implies that   $ d(\cdot,F)$ is  continuous for all fixed $F \in \mathcal{F}$. Using  \citep[Lemma 4.51 p153]{Hitch} $d$ is  $\mathcal{B}(X)\otimes \mathcal{B}(\mathcal{F}) $-measurable. We conclude since  
$x \in F$ if and only if $d(x,F)=0$.
\end{proof}\\

 \noindent 
 \begin{lemma}
 \label{Selection}
Assume that Assumption \ref{QanalyticARB} holds true and that $\mathcal{Q}^{T}$ is dominated by $\widehat{P} \in \mathfrak{P}(\O^{T})$ with the fix disintegration $\widehat{P}:=\widehat{P}_{0}\otimes \widehat{p}_{1}  \otimes \cdots \otimes \widehat{p}_{T}$ where $\hat p_{t} \in  \mathcal{SK}_{t}$ for all $1 \leq t \leq T$. Then    $$\O^{t}_{nd}:=\left\{\o^{t} \in \O^{t},\;  \mbox{$\mathcal{Q}_{t+1}(\o^{t})$ is not dominated by $\widehat{p}_{t+1}(\cdot,\o^{t})$}\right\} \in \mathcal{B}_{c}(\O^{t})$$ and is a $\mathcal{Q}^{t}$-polar set for all $0 \leq t \leq T-1$.
\end{lemma}
\begin{proof}
Fix some $0 \leq t \leq T-1$. We proceed in two steps. \\
{\it Step 1:   $\O^{t}_{nd} \in \mathcal{B}_{c}(\O^{t}).$}\\
To prove Step $1$, we   use  Lemma \ref{AAA} and fix  $R \in \mathfrak{P}(\O^{t})$. 
Applying \citep[Lemma 7.28 p174]{BS},  there exists $p^{R}_{t+1}$ a Borel-measurable stochastic kernel on $\O_{t+1}$ given $\O^{t}$  and a $R$-full-measure set $\O^{t}_{R} \in \mathcal{B}(\O^{t})$ such that \begin{align}
\label{Pr}
 p^{R}_{t+1}(\cdot,\o^{t})=\widehat{p}_{t+1}(\cdot,\o^{t}) \; \mbox{ for all $\o^{t} \in \O^{t}_{R}$.}
\end{align} 
Let $\mathcal{F}_{t+1}$ be the set of non-empty and  closed subsets of  $\O_{t+1}$ 
and let $N^{R}_{t}: \Omega^t \twoheadrightarrow  \mathfrak{P}(\O_{t+1})\times \mathcal{F}_{t+1}$ be  defined for all $\o^{t} \in \O^{t}$ by
\begin{align}
\label{NRs}
N^{R}_{t}(\o^{t}):=\left\{(q,F) \in \mathfrak{P}(\O_{t+1})\times \mathcal{F}_{t+1}, q \in \mathcal{Q}_{t+1}(\o^{t}),\; p^{R}_{t+1}\left(F,\o^{t}\right)=0,\; q\left(F\right)>0\right\}.
\end{align}
We first claim  that \begin{align}
\label{Ond}
\O^t_{nd} \cap \O^{t}_{R} = \{N^{R}_{t} \neq \emptyset\} \cap \O^{t}_{R}.
\end{align}
Let $\o^{t} \in \O^t_{nd} \cap \O^{t}_{R}$. As $\mathcal{Q}_{t+1}(\o^{t})$ is not dominated by $\widehat{p}_{t+1}(\cdot,\o^{t})=p^{R}_{t+1}(\cdot,\o^{t})$, there exists some $q \in \mathcal{Q}_{t+1}(\o^{t})$ and some $A \in \mathcal{B}(\O^{t})$ such that $p^{R}_{t+1}(A,\o^{t})=0$ and  $q(A)>0$. As  $q \in \mathfrak{P}(\O_{t+1})$ is inner-regular (see \citep[Definition 12.2 p435, Theorem 12.7  p438, Lemma 12.3 p435]{Hitch}), 
there exists some $F \in \mathcal{F}_{t+1}$, $F \subset A$ such that $q(F)>0$ and 
$(q,F) \in N_{t}^{R}(\o^{t})$ follows. The reverse inclusion is clear.\\ 
Thus Lemma \ref{AAA} applies and Step 1 is completed if 
$\{N^{R}_{t} \neq \emptyset\}  =  \mbox{Proj}_{\O^{t}} \left(\mbox{graph}\left( N^{R}_{t} \right) \right) \in \mathcal{B}_{c}(\O^{t}).$ 
This will follows from  Jankov-von Neumann Theorem (see  \citep[Proposition 7.49 p182]{BS}) if 
\begin{align}
\label{horor1}
\mbox{graph}(N^{R}_{t}) \in  \mathcal{A}\left(\O^{t} \times \mathfrak{P}(\O_{t+1}) \times \mathcal{F}_{t+1} \right).
\end{align}
This follows from $\mbox{graph}(N^{R}_{t}) =A \cap B \cap C$ where   \begin{align*}
A&:=  \mbox{graph}(\mathcal{Q}_{t+1}) \times \mathcal{F}_{t+1}  \in \mathcal{A}\left(\O^{t} \times \mathfrak{P}(\O_{t+1}) \times \mathcal{F}_{t+1} \right), \\
B&:=\{(\o^{t},q,F),\; p^{R}_{t+1}(F,\o^{t})=0\} \in \mathcal{B}(\O^{t})\otimes \mathcal{B}(\mathfrak{P}(\O_{t+1}) \otimes \mathcal{B}(\mathcal{F}_{t+1}),\\
C&:=\{(\o^{t},q,F),\; q(F)>0\} \in  \mathcal{B}(\O^{t})\otimes \mathcal{B}(\mathfrak{P}(\O_{t+1}) \otimes \mathcal{B}(\mathcal{F}_{t+1}),
\end{align*}
see  Assumption \ref{QanalyticARB} for the  measurability of $A$. For $B$ and $C$,  Lemma \ref{LK} together with  \citep[Proposition 7.29 p144]{BS} imply that $(\o^{t},q,F) \to p^{R}_{t+1}(F,\o^{t})$ and 
$(\o^{t},q,F) \to q(F)$ are  Borel-measurables (recall that  $p_{t+1}^{R}(d\o_{t+1}|\o^{t},q,F)=p^{R}_{t+1}(d\o_{t+1},\o^{t})$ 
and $q(d\o_{t+1}|\o^{t},q,F)=q(d\o_{t+1})$ are Borel-measurable stochastic kernels).\\

{\it Step 2: $\O^{t}_{nd}$ is a $\mathcal{Q}^{t}$-polar set.}\\
We proceed by contradiction and assume that there exists  some  $\overline{P} \in \mathcal{Q}^{T}$ such that  $\overline{P}^t(\O^{t}_{nd})>0$.
We choose $R=\widehat{P}^t$ in \eqref{Pr} and  \eqref{NRs} and we denote by 
$$\O^t_{nd1}:=\O^t_{nd} \cap \O^{t}_{\widehat{P}^t}= \{N^{\widehat{P}^t}_{t} \neq \emptyset\} \cap \O^{t}_{\widehat{P}^t} \in \mathcal{B}_{c}(\O^{t}),$$ see \eqref{Ond} and Step 1.  
The Jankov-von Neumann Theorem and \eqref{horor1} also give the existence  of  ${q}^{\widehat{P}}_{t+1}$ a universally-measurable  stochastic kernel on $\O_{t+1}$ given $\O^{t}$  and a universally measurable  function ${F}^{\widehat{P}}_{t+1}: \O^{t} \to \mathcal{F}_{t+1}$ such that  $({q}^{\widehat{P}}_{t+1}(\cdot,\o^{t}), F^{\widehat{P}}_{t+1}(\o^{t})) \in N^{\widehat{P}^t}_{t}(\o^{t})$  for all $\o^{t} \in \O^{t}_{nd1}$. For $\o^{t} \notin \  \O^{t}_{nd1}$ we set $F^{\widehat{P}}_{t+1}(\o^{t})=\emptyset$ and ${q}^{\widehat{P}}_{t+1}(\cdot,\o^{t})=q_{t+1}(\cdot,\o^{t})$  where $q_{t+1}$ is a given universally-measurable selector of $\mathcal{Q}_{t+1}$. \\
Note that as  $\widehat{P}^{t}$ dominates $\mathcal{Q}^{t}$, $1=\widehat{P}^t( \O^{t}_{\widehat{P}^t})=\overline{P}^t( \O^{t}_{\widehat{P}^t})$ and  $\overline{P}^t( \O^t_{nd1})>0$.\\
We now   build some $\widehat{Q} \in \mathcal{Q}^{T}$, $E \in \mathcal{B}_{c}(\O^{t+1})$ such that $\widehat{P}^{t+1}(E)=0$ but $\widehat{Q}^{t+1}(E)>0$ 
which contradicts the fact that $\widehat{P}$ dominates $\mathcal{Q}^{T}$.  
  Let   \begin{align*}
\widehat{Q}&:=\overline{P}^{t} \otimes {q}^{\widehat{P}}_{t+1} \otimes \overline{p}_{t+2} \otimes \cdots \otimes \overline{p}_{T} \in \mathcal{Q}^{T},\\
E&:=\left\{(\o^{t}, \o_{t+1}) \in \O^{t}\times \O_{t+1},\; \o^{t} \in \O^{t}_{nd1}, \; \o_{t+1} \in F^{\widehat{P}}_{t+1}(\o^{t}) \right\}=\varphi^{-1}(\{1\}) \cap \left(\O^{t}_{nd1} \times \O_{t+1}\right),\\
\varphi(\o^t,\o_{t+1})& :=1_{ F^{\widehat{P}}_{t+1}(\o^{t})}(\o_{t+1}).
\end{align*}
Lemma \ref{LK} implies that  $(F,\o_{t+1}) \to 1_{F}(\o_{t+1})$ is $\mathcal{B}(\mathcal{F}_{t+1})\otimes \mathcal{B}(\O_{t+1})$-measurable and as $(\o^{t},\o_{t+1}) \to ( F^{\widehat{P}}_{t+1}(\o^{t}), \o_{t+1})$ is $\mathcal{B}_{c}(\O^{t+1})$-measurable,   $\varphi$ 
is $\mathcal{B}_{c}(\O^{t+1})$-measurable by composition. Thus $E$ belong to $\mathcal{B}_{c}(\O^{t+1})$.
Let $\left(E\right)_{\o^{t}}:=\{ \o_{t+1} \in \O_{t+1},\; (\o^{t},\o_{t+1}) \in E\}$,  then 
\begin{align*}
\widehat{P}^{t+1}(E)=\int_{\O_{nd1}^{t}} \widehat{p}_{t+1}\left(\left(E\right)_{\o^{t}},\o^{t}\right) \widehat{P}^{t}(d\o^{t})
&=\int_{\O_{nd1}^{t}} \widehat{p}_{t+1}\left(F^{\widehat{P}}_{t+1}(\o^{t}),\o^{t}\right) \widehat{P}^{t}(d\o^{t})=0
\end{align*}
where we have used that for $\o^{t} \notin \O_{nd1}^{t}$ $\left(E\right)_{\o^{t}}=\emptyset$ and for $\o^{t} \in \O^{t}_{nd1}$  $\left(E\right)_{\o^{t}}=F^{\widehat{P}}_{t+1}(\o^{t})$ and that $ \widehat{p}_{t+1}\left(F^{\widehat{P}}_{t+1}(\o^{t}),\o^{t}\right)=p^{\widehat{P}^t}_{t+1}\left(F^{\widehat{P}}_{t+1}(\o^{t}),\o^{t}\right)  =0$. But 
\begin{align*}
\widehat{Q}^{t+1}(E)& = \int_{\O_{nd1}^{t}} {q}^{\widehat{P}}_{t+1}\left(\left(E\right)_{\o^{t}},\o^{t}\right) \overline{P}^{t}(d\o^{t})
=\int_{\O_{nd1}^{t}} {q}^{\widehat{P}}_{t+1}\left(F^{\widehat{P}}_{t+1}(\o^{t}),\o^{t}\right)\overline{P}^{t}(d\o^{t})>0
\end{align*}
since $\overline{P}^t\left(\O^{t}_{nd1}\right)>0$ and $ {q}^{\widehat{P}}_{t+1}\left(F^{\widehat{P}}_{t+1}(\o^{t}),\o^{t}\right)>0$ for all $\o^{t} \in \O^{t}_{nd1}$. This concludes the proof.
\end{proof}\\

\bibliography{biblioRomain}

\end{document}